\documentclass[twoside,11pt]{article}
\usepackage[preprint]{jmlr2e}

\usepackage{mathrsfs}
\usepackage{amsmath}
\usepackage{dsfont}

\usepackage{algorithm}
\usepackage{algpseudocode}
\algrenewcommand\algorithmicrequire{\textbf{Input:}}
\algrenewcommand\algorithmicensure{\textbf{Output:}}

\usepackage{enumitem}

\usepackage{subcaption} % for subfigure

\newcommand{\prob}{\mathbb{P}}
\newcommand{\mean}{\mathbb{E}}
\newcommand{\med}{\text{med}}
\newcommand{\Var}{\mathbb{V}}

\newcommand{\cov}{\text{cov}}
\newcommand{\mode}{\text{mode}}
\newcommand{\stdev}{\text{sd}}
\newcommand{\ud}{\mathop{}\!\mathrm{d}}
\newcommand{\cond}{\,|\,}

\newcommand{\e}{e} % e as exp()
\newcommand{\pr}{^{\backprime}} % variant for prime '

\DeclareMathOperator{\Normal}{\mathcal{N}}
\DeclareMathOperator{\logN}{\text{log}\mathcal{N}}

\DeclareMathOperator{\Poi}{{Poi}}

\DeclareMathOperator{\GP}{\mathcal{GP}}
\DeclareMathOperator{\Unif}{\mathcal{U}}

\newcommand{\Bb}{\mathbf{b}}

\newcommand{\Bk}{\mathbf{k}}

\newcommand{\Br}{\mathbf{r}}

\newcommand{\Bx}{\mathbf{x}}
\newcommand{\By}{\mathbf{y}}

\newcommand{\BB}{\mathbf{B}}

\newcommand{\BH}{\mathbf{H}}
\newcommand{\BK}{\mathbf{K}}
\newcommand{\BL}{\mathbf{L}}

\newcommand{\BR}{\mathbf{R}}
\newcommand{\BS}{\mathbf{S}}

\newcommand{\BSigma}{\boldsymbol{\Sigma}}

\newcommand{\Bpsi}{\boldsymbol{\psi}}
\newcommand{\Bbeta}{\boldsymbol{\beta}}
\newcommand{\Bmu}{\boldsymbol{\mu}}
\newcommand{\Btheta}{\boldsymbol{\theta}}

\newcommand{\BLambda}{\boldsymbol{\Lambda}}
\newcommand{\Bphi}{\boldsymbol{\phi}}

\newcommand{\epsi}{\varepsilon}

\newcommand{\reals}{\mathbb{R}}
\newcommand{\realsp}{\mathbb{R}_{+}}

\newcommand{\Dset}{\mathscr{D}}
\newcommand*{\half}[1][2]{\frac{1}{#1}}
\newcommand{\onevector}{\mathds{1}}

\newcommand{\zeromatrix}{\mathbf{0}}
\newcommand{\Bzeros}{\mathbf{0}}
\newcommand{\Id}{\mathbf{I}}

\newcommand{\bvn}{\textnormal{BvN}} % bivariate Normal

\newcommand{\myquad}{\quad\quad}
\newcommand{\eqdef}{\triangleq}
\newcommand{\T}{^{\top}}
\newcommand{\Tpr}{'{}^{\top}}
\newcommand{\Tinv}{^{-\top}}
\newcommand{\diag}{\operatorname{diag}}

\newcommand{\bigO}{\mathcal{O}} % Big O symbol
\newcommand{\indic}{\mathds{1}}
\newcommand{\simiid}{\sim_{\text{i.i.d.}}}

\newcommand{\bu}{\bullet}
\newcommand{\sigmaf}{\sigma_s}
\newcommand{\n}{b}

\newcommand{\xt}{t} % 1:t
\newcommand{\Bxobs}{\Bx_{\textnormal{o}}}

\newcommand{\blfi}{Bayesian LFI}
\newcommand{\BLFI}{BLFI}
\newcommand{\owen}{Owen's T function}
\newcommand{\TV}{\textnormal{TV}}
\newcommand{\hatgamma}{\hat{\gamma}}

\newcommand{\tildepiapprox}{\tilde{\pi}}

\newcommand{\tpi}{\tilde{p}}

\newcommand{\accratio}{\alpha}
\newcommand{\ddata}{\Dset}
\newcommand{\logu}{\tilde{u}}
\newcommand{\conderr}{\mathcal{E}} % conditional error
\newcommand{\unconderr}{E} % unconditional error
\newcommand{\accrejindic}{\kappa} % indicator whether acceptance/rejection
\newcommand{\BLL}{\BLambda} % matrix containing e.g. the length-scales
\newcommand{\AcqParams}{NextUnconfidentTransition}
\newcommand{\epoe}{EPoE}
\newcommand{\epoer}{EPoEr}
\newcommand{\naive}{naive}

\newcommand{\mh}{MH}
\newcommand{\mhblfi}{MH-BLFI}
\newcommand{\gpmh}{GP-MH}

\renewcommand\cite{\citep}
\newcommand{\appe}{\textnormal{Appendix}}

\def\app#1#2{%
  \mathrel{%
    \setbox0=\hbox{$#1\sim$}%
    \setbox2=\hbox{%
      \rlap{\hbox{$#1\propto$}}%
      \lower1.1\ht0\box0%
    }%
    \raise0.25\ht2\box2%
  }%
}

\allowdisplaybreaks % splitting align

\usepackage{aligned-overset}

\let\originalleft\left
\let\originalright\right
\renewcommand{\left}{\mathopen{}\mathclose\bgroup\originalleft}
\renewcommand{\right}{\aftergroup\egroup\originalright}

\ShortHeadings{Approximate Bayesian inference with GP emulated MCMC}{J\"{a}rvenp\"{a}\"{a} and Corander}
\firstpageno{1}

\begin{document}

%\title{Gaussian process emulated approximate MCMC with noisy likelihoods}
%\title{Efficient approximate posterior inference from noisy likelihoods with Gaussian process emulated MCMC}
%\title{Gaussian process emulated Metropolis-Hastings for approximate Bayesian inference}
\title{Approximate Bayesian inference from noisy likelihoods with Gaussian process emulated MCMC}

\author{\name Marko J\"{a}rvenp\"{a}\"{a} \email m.j.jarvenpaa@medisin.uio.no \\
       \addr Department of Biostatistics, University of Oslo, Norway
       \AND
       \name Jukka Corander \email jukka.corander@medisin.uio.no \\
       \addr Department of Biostatistics, University of Oslo, Norway\\
       Department of Mathematics and Statistics, University of Helsinki, Finland\\
       Wellcome Trust Sanger Institute, United Kingdom}

% \author[1]{Marko J\"{a}rvenp\"{a}\"{a}}
% \author[1,2,3]{Jukka Corander}
% \affil[1]{Department of Biostatistics, University of Oslo, Norway}
% \affil[2]{Department of Mathematics and Statistics, Helsinki Institute of Information Technology (HIIT), University of Helsinki, Finland}
% %\affil[3]{Pathogen Genomics, Wellcome Trust Sanger Institute, United Kingdom}
% \affil[3]{Wellcome Trust Sanger Institute, United Kingdom}
% \date{\today}

\editor{}

\maketitle

\begin{abstract}%   <- trailing '%' for backward compatibility of .sty file
We present a framework for approximate Bayesian inference when only a limited number of noisy log-likelihood evaluations can be obtained due to computational constraints, which is becoming increasingly common for applications of complex models. We model the log-likelihood function using a Gaussian process (GP) and the main methodological innovation is to apply this model to emulate the progression that an exact Metropolis-Hastings (\mh{}) sampler would take if it was applicable. Informative log-likelihood evaluation locations are selected using a sequential experimental design strategy until the \mh{} accept/reject decision is done accurately enough according to the GP model. The resulting approximate sampler is conceptually simple and sample-efficient. It is also more robust to violations of GP modelling assumptions compared with earlier, related ``Bayesian optimisation-like'' methods tailored for Bayesian inference. We discuss some theoretical aspects and various interpretations of the resulting approximate \mh{} sampler, and demonstrate its benefits in the context of Bayesian and generalised Bayesian likelihood-free inference for simulator-based statistical models. 
\end{abstract}


\section{Introduction} \label{sec:intro}

% general about Bayesian inference and comput. challenges
Standard computational methods for Bayesian inference, such as those based on Markov chain Monte Carlo (MCMC), require a large number of likelihood evaluations and are hence poorly suited as such when the likelihood function is expensive to evaluate. Additional challenge is brought by noisy evaluations which occurs, for example, in the context of likelihood-free inference (LFI) where the intractable likelihood function is itself estimated from forward simulations of the statistical model, see \citet{Marin2012,Sisson2019,Cranmer2019}. In particular, in the synthetic likelihood (SL) method \cite{Wood2010,Price2018}, typically hundreds or thousands of repeated simulations are needed to approximate the likelihood function at each evaluation location which can be costly and produces noisy estimates. Other situations with expensive likelihood evaluations include parameter estimation in complex ordinary/partial differential equations with tractable noise models (e.g.~\citet{Cleary2021,Paun2022}) and Bayesian machine learning with large data sets (e.g.~\citet{Korattikara2014,Bardenet2017}). 
%The resulting noisy and often costly likelihood evaluations make conducting Bayesian inference challenging. 
%While we mainly consider LFI in this paper, our framework is directly applicable whenever the likelihood function, some generalisation of it \citep{Bissiri2016,Schmon2021,Pacchiardi2021} or an approximation such as SL is expensive to compute but possesses local regularity at least at the highest density region. 

% earlier work, BLFI
Approximate Bayesian inference methods based on GP surrogate models have been proposed to relieve the computational challenges. 
In particular, recently \citet{Jarvenpaa2019_sl} proposed such a framework when a limited number (e.g.~less than $10^3$) of noisy log-likelihood evaluations can only be obtained which we call ``\blfi{}'' (\BLFI{}) in this paper. %(\BLFI{}\footnote{We have adopted this name although it was not explicitly used by \citet{Jarvenpaa2019_sl}.}) by \citet{Jarvenpaa2019_sl}. 
\BLFI{} is related to so-called probabilistic numerics methods \citep{Hennig2015,Cockayne2017,hennig_osborne_kersting_2022}. The key idea is to frame the computation of the posterior density itself as a Bayesian inference task, in the same spirit as global, derivative-free optimisation is framed as an inference task in acclaimed Bayesian optimisation methodology. In \BLFI{} the log-likelihood function is modelled with GP which is further used to form an estimator for the posterior density. %The uncertainty in the posterior due to the limited number of log-likelihood evaluations could also be quantified. 
\emph{Active learning} is used to maximise the information brought by the limited budget of log-likelihood evaluations. 

%% difficulties with "Bayesian LFI/ABC" & "BOLFI"
% high-dim cases or when the prior is significantly more broad as the prior
% non-stationarity and Inf/NaN outputs near borders
% difficult to code information about unimodality and there is no suitable non-stationary kernel
% difficult to code information about initial point/location of the posterior
\BLFI{} and other related approaches such as the operationally similar \emph{Bayesian optimisation for likelihood-free inference\footnote{BLFI and especially BOLFI should not be confused with ``standard'' Bayesian optimisation methods that are designed only for global, derivative-free optimisation. Although we focus on LFI as the naming convention implies, the target log-likelihood can in principle be replaced by any log-density of interest and the associated prior density neglected which allows more general use of B(O)LFI and the methods proposed in this paper.}} (BOLFI) framework by \citet{Gutmann2016} use a global GP surrogate model over the whole parameter space to achieve the sample-efficiency. This however involves some difficulties: 
First, such global modelling is sensible when the target density is expected to be highly multimodal but often the posterior is actually close to Gaussian, unimodal or at least its highest density region\footnote{We use \emph{the highest (posterior) density region} rather loosely to refer to the region where the (posterior) density is non-negligible and where good approximation accuracy is most importantly needed.} defines a relatively small region of the parameter space. Such global modelling may hence involve wasted effort and becomes inefficient, for instance, when the parameter space is high-dimensional. % or when the prior density is substantially more broad than the posterior. %, one cannot hope to globally model and explore the whole parameter space efficiently. 
%In particular, when the prior density is substantially more broad than the posterior or when the parameter space is high-dimensional, one cannot hope to globally model and explore the whole parameter space efficiently. 
%Unless the target density is expected to be highly multimodal or similar to the prior, the global search tendency of B(O)LFI is undesirable and at odds with the local search behaviour of commonly used MCMC methods. 
Importantly, the likelihood function can also behave irregularly, as is often the case with nonlinear dynamic models used in ecology and epidemiology \citep{Fasiolo2016}, or be arbitrarily small near the parameter boundaries which results practical difficulties with GP fitting. %All in all, an inference method is needed which preserves the sample-efficiency and interpretability of B(O)LFI yet more effectively focuses the computations on the highest posterior density region. %The log-likelihood might be more easy to model there using a standard GP as noted already by \citet{Rasmussen2003}. % and facilitates efficient computations. 

% why MCMC-based approach -> initial point
%Sometimes the highest posterior density region is roughly known based on e.g.~pilot runs, expert knowledge, or earlier analyses, but incorporating such information in \BLFI{} is not straightforward. %For example, handling the nontrivial shape of a highest density region of a ``banana shaped'' posterior in \BLFI{} would be difficult.
%Conveying information, e.g.~about potential unimodality of the posterior or non-stationarity at the boundary regions to the GP prior is likewise challenging. 

% comput. difficulties with acq functions
The current active learning strategies for B(O)LFI are also not fully satisfactory: First, they all require global optimisation of some so-called \emph{acquisition function} at each iteration of the algorithm. Especially the one-step ahead optimal Bayesian experimental design strategies by \citet{Jarvenpaa2019_sl} are rather costly. Although this is only a minor concern when the likelihood evaluations are truly expensive, it still complicates the inference pipeline. %\citet{Jarvenpaa2019_sl} developed also some cheap heuristics but did not always work reliably in their experiments. 
The more heuristic but cheaper acquisition functions based on {uncertainty sampling}, also by \citet{Jarvenpaa2019_sl}, did not work consistently. The ones borrowed from Bayesian optimisation literature (see \citet{Gutmann2016}) are not theoretically sound as they are not explicitly designed to estimate the posterior distribution \citep{Kandasamy2015,Jarvenpaa2018_acq,Jarvenpaa2019_sl} and tend to produce redundant evaluations near the boundaries as observed e.g.~in \citet{Jarvenpaa2018_acq,Picchini2020}. % which can aggravate the aforementioned challenges with GP fitting. 

% contributions of this paper:
In this paper we develop and analyse a new inference framework \gpmh{} that directly combines \mh{} sampling with the benefits of the B(O)LFI approach. 
%facilitating high sample-efficiency which is important when the likelihood evaluations are expensive. 
In addition to other advantages and new theoretical insights obtained as a by-product (some of which might be of independent interest), this new framework alleviates the aforementioned difficulties. In particular, the resulting \gpmh{} algorithm models and explores the parameter space in a more local fashion by emulating the progression of an exact but directly inapplicable \mh{} sampler. This allows to either avoid or robustly manage the problematic likelihood evaluations especially near the parameter boundaries which are often redundant anyway. 
While possible weak information about the shape or location of the highest density region based on e.g.~pilot runs, expert knowledge, or iterative model building, is difficult to code into the GP model, it can be more easily captured in an initialisation scheme, or the proposal density. 
Sequential Bayesian experimental design strategies are developed here to select informative evaluation locations and they provide fairly similar sample-efficiency as the B(O)LFI methods but are more computationally efficient and interpretable. % and would ease up handling e.g.~discrete parameters. 
The resulting \gpmh{} implementation, similarly to the current B(O)LFI methods, is mainly intended for models with low-dimensional parameters and require assumptions regarding the smoothness of the log-likelihood function (but not on its shape).

%Sequential experimental design strategies are used to gather new evaluation locations optimally in the sense of Bayesian decision theory which leads to fairly similar sample-efficiency as the B(O)LFI methods but are more interpretable. Also, computational challenges related to the GP-based methods themselves are alleviated. For example, optimising the design criterion (acquisition function) can be done more efficiently or possibly even avoided entirely.

%The focus of this paper is on approximate Bayesian inference in the context of intractable statistical models. However, unnormalised, expensive-to-evaluate densities emerge also in other applications. All the sampling algorithms considered in this paper can in principle be also used in other contexts by replacing the log-likelihood function in our analysis by the log-density of interest (and by neglecting the prior density when appropriate). In particular, the application of generalised Bayesian inference \citep{Bissiri2016,Schmon2021,Pacchiardi2021} is obtained by replacing the log-likelihood with some appropriate loss function. 

% This paper is organised as follows...
The rest of this paper is organised as follows. 
We first provide background on \mh{} sampling, \BLFI{} and previous methods in statistics and machine learning literature that use GPs for more efficient Bayesian inference in Section \ref{sec:background}.
We then develop our framework (Section \ref{sec:gpe}) and derive sequential experimental design strategies for it (Section \ref{sec:acq}). 
In Section \ref{sec:probinterp} we observe that some common \mh{} samplers are special cases of \gpmh{} and discuss how \gpmh{} can be interpreted as a special case of \BLFI{}, which motivates an alternative implementation called \mhblfi{}, or as a heuristic yet tractable estimate to a conceptual ``Bayesian'' \mh{} sampler. 
Section \ref{sec:theory} outlines some theoretical aspects of \gpmh{}. 
In section \ref{sec:exp} we investigate the proposed implementations numerically first with synthetically constructed target densities and then with realistic simulator-based models with SL and generalised Bayesian updating in the spirit of \citet{Bissiri2016,Pacchiardi2021}.
Summary and discussion about future research directions conclude the paper. 
%Mathematical derivations, additional analysis, implementation details and additional numerical results can be found in \appe{}. 

% Notes that could be added to the paper:
% - reveals some insights regarding how parameter space is explored in standard practice of MCMC and "Bayesian optimisation-like" inference methods such as B(O)LFI
% - awkwardness with first learn the density and then plug it in to a standard MCMC algorithm!
% - framework vs. particular implementation of GP-MH
% - suitable initial point already known
% -... using (un)cond.error because it contains minimum information needed to determine the correct progression of MH sampler

\section{Background} \label{sec:background}

% basics; parameter, model, prior, likelihood, data etc. 
We denote the observed data as $\Bxobs\in\reals^d$, the unknown parameters of the statistical model as $\Btheta\in\Theta\subset\reals^p$ and the associated likelihood function as $\pi(\Bxobs\cond\Btheta)$. We focus on continuous parameter spaces but most of the analysis extends to the case where some components of $\Btheta$ are discrete. The prior density is denoted by $\pi(\Btheta)$ and is assumed tractable. 
Bayes' theorem gives the posterior distribution $\pi(\Btheta\cond\Bxobs) = \pi(\Btheta)\pi(\Bxobs\cond\Btheta)/\pi(\Bxobs)$, where $\pi(\Bxobs) = \int_{\Theta} \pi(\Btheta\pr)\pi(\Bxobs\cond\Btheta\pr) \ud \Btheta\pr$ is the marginal likelihood. 
Throughout this paper the main goal will be on sampling from $\pi(\Btheta\cond\Bxobs)$ or computing some posterior expectations of the form
\begin{equation}
\bar{h} \eqdef
\mean_{\Btheta\cond\Bxobs}h(\Btheta)
= \int_{\Theta} h(\Btheta) \pi(\Btheta\cond\Bxobs) \ud \Btheta
= \frac{\int_{\Theta}h(\Btheta)\pi(\Btheta)\pi(\Bxobs\cond\Btheta)\ud\Btheta}{\int_{\Theta}\pi(\Btheta\pr)\pi(\Bxobs\cond\Btheta\pr)\ud\Btheta\pr}, 
\label{eq:posterior_mean}
\end{equation}
where $h:\Theta\rightarrow\reals$ is some known, tractable function that is cheap to evaluate.

%\subsection{Applications} \label{subsec:applications}

% exact lik/log-lik --> evaluations exact so unbiased lik/log-lik --> OK for both
% ABC --> unbiased likelihood --> no immedate unbiased log-lik
% SL --> ok application
% latent variable models (pseudo-marg MH applies) --> unbiased likelihood --> no immedate unbiased log-lik
% Bayesian inference with tall data --> unbiased log-lik but tractable
% Generalized (loss-based) Bayesian inference --> ok application

\subsection{Metropolis-Hastings algorithm} \label{subsec:mh}

% M-H - general
Metropolis-Hastings sampler \citep{Hastings1970} is widely used for Monte Carlo integration. 
Algorithm \ref{alg:mh} shows in a compact form how \mh{} is used to draw samples from $\pi(\Btheta\cond\Bxobs)$. Under certain technical conditions, this \mh{} sampler produces a Markov chain whose stationary distribution is the posterior $\pi(\Btheta\cond\Bxobs)$. 
The algorithm starts from an initial point $\Btheta^{(0)}$. At each iteration $i$ a new parameter denoted %\footnote{We always use $\Btheta'^{(i)}$, or just $\Btheta'$, to denote the proposed point of the exact or some approximate \mh{} sampler while $\Btheta\pr$ is used to denote an arbitrary parameter in other contexts.}
by $\Btheta'^{(i)}$ is drawn from the proposal density $q(\Btheta'^{(i)}\cond\Btheta^{(i-1)})$ and is then accepted with probability $\accratio(\Btheta^{(i-1)},\Btheta'^{(i)})$, where
\begin{align}
    \accratio(\Btheta,\Btheta') \eqdef \min\{ 1, \gamma(\Btheta,\Btheta') \}, 
    \quad 
    \gamma(\Btheta,\Btheta') \eqdef \frac{\pi(\Btheta')\pi(\Bxobs\cond\Btheta')q(\Btheta\cond\Btheta')}{\pi(\Btheta)\pi(\Bxobs\cond\Btheta)q(\Btheta'\cond\Btheta)},
\label{eq:accratio}
\end{align}
and otherwise the current point $\Btheta^{(i-1)}$ is kept. 
%The likelihood function $\pi(\Bxobs\cond\Btheta)$ in (\ref{eq:accratio}) needs to be evaluated exactly (but only up to normalisation). 
The initial samples (e.g.~the first half) are often discarded as ``burn-in''. The remaining samples, here denoted as $\Btheta^{(0)},\ldots,\Btheta^{(n)}$, are approximately distributed as $\pi(\Btheta\cond\Bxobs)$ and can be used to estimate (\ref{eq:posterior_mean}) as 
\begin{equation}
\bar{h} \approx 
\hat{\bar{h}}_{n+1}
\eqdef \frac{1}{n+1}\sum_{i=0}^{n} h(\Btheta^{(i)}). \label{eq:posterior_mean_mcmc_approx}
\end{equation}
See e.g.~\citet{Robert2004} for a more detailed treatment of MCMC methods.

% ALGORITHM: Metropolis-Hastings
\begin{algorithm}[htb]
\caption{Metropolis-Hastings (\mh{}) sampler for $\pi(\Btheta\cond\Bxobs)$} \label{alg:mh}
 \begin{algorithmic}[1]
 \Require Prior $\pi(\Btheta)$, likelihood $\pi(\Bxobs\cond\Btheta)$, initial point $\Btheta^{(0)}$, proposal $q(\Btheta'\cond\Btheta)$, no.~samples $i_{\text{MH}}$
 \Ensure Samples $\Btheta^{(1)},\ldots,\Btheta^{(i_{\text{MH}})}$ %$(\Btheta^{(i)})_{i=1}^{i_{\text{MCMC}}}$
 \For{$i=1:i_{\text{MH}}$}
 \State Draw $\Btheta'^{(i)} \sim q(\cdot\cond \Btheta^{(i-1)})$ and $u^{(i)} \sim \Unif([0,1])$ \label{line:q_u_gen}
 \State Set $\Btheta^{(i)} \leftarrow \Btheta'^{(i)}\indic_{\alpha(\Btheta^{(i-1)},\Btheta'^{(i)})\geq u^{(i)}}+\Btheta^{(i-1)}\indic_{\alpha(\Btheta^{(i-1)},\Btheta'^{(i)})< u^{(i)}}$ %\Comment{$\alpha(\Btheta^{(i-1)},\Btheta'^{(i)])$ is computed using (\ref{eq:accratio}).} 
 \label{line:newtheta}
 \EndFor
 \end{algorithmic}
\end{algorithm}

% pseudo-marginal M-H
The \mh{} sampler in Algorithm \ref{alg:mh} uses an exact target density evaluation (up to normalisation) in the \mh{} acceptance test (\ref{eq:accratio}). 
However, if a noisy but unbiased likelihood evaluation is used in place of the exact $\pi(\Bxobs\cond\Btheta')$, the old realisation of $\pi(\Bxobs\cond\Btheta)$ is carried on from the previous iteration instead of recomputing it and if certain technical conditions hold, the resulting modified sampler is a \emph{pseudo-marginal} \mh{} \citep{Beaumont2003,Andrieu2009} whose target density is still the exact posterior. Otherwise additional approximation error is introduced, see e.g.~\citet{Alquier2016} for analysis of the resulting \emph{``noisy''} \mh{} sampler. 
%Pseudo-marginal/noisy \mh{} has been combined with SL by \citet{Price2018}. We call the resulting method SL-MCMC. 
All these samplers require a new likelihood evaluation at each iteration and typically a large number of iterations to converge which makes them prohibitively expensive as such when the evaluations are costly.
The mixing of pseudo-marginal \mh{} is also slowed down by the ``sticky'' behaviour of the chain due to the noisy evaluations.

\subsection{Bayesian approach to Bayesian inference with noisy, expensive likelihoods} \label{subsec:blfi}

% "BLFI"
In the \BLFI{} framework the log-likelihood function, from now on denoted by $f:\Theta\rightarrow\reals$ so that $f(\Btheta)\eqdef\log\pi(\Bxobs\cond\Btheta)$  where the dependence on the fixed data $\Bxobs$ is suppressed for brevity, is itself treated as an unknown random function to be estimated using (another level of) Bayesian inference. 
A GP prior is placed on $f$ and a Gaussian noise model is assumed for the log-likelihood evaluations. Some $t$ parameter values are chosen to evaluate the log-likelihood $f$ which results the data $\ddata_t$. Then $f\cond\ddata_t$ follows a GP which further induces a log-GP posterior for the \emph{unnormalised} posterior 
\begin{equation}
\tildepiapprox_{f}(\Btheta) \eqdef \pi(\Btheta)\e^{f(\Btheta)}. \label{eq:unn_post}
\end{equation}
An estimate for (\ref{eq:unn_post}) is obtained analytically by using the properties of GP models. This GP-based estimate defines an approximation to the exact unnormalised posterior $\tildepiapprox(\Btheta\cond\Bxobs) \eqdef \pi(\Btheta)\pi(\Bxobs\cond\Btheta)$ and is plugged in to a standard MCMC algorithm. This last step does not require further log-likelihood evaluations as it is solely based on the GP posterior. Finally, the resulting samples can be used to approximate $\bar{h}$ as usual, via (\ref{eq:posterior_mean_mcmc_approx}). 

%% Bayesian experimental design strategy --> sample-efficiency
%Informative data $\ddata_t$ can be collected e.g.~using sequential Bayesian experimental design: At each step of the \BLFI{} algorithm, a new parameter $\Btheta^*$ for evaluating the log-likelihood $f$ is then chosen as the minimiser of an expected loss function, where the loss quantifies the uncertainty associated with the unnormalised posterior (\ref{eq:unn_post}) and the expectation is taken with respect to a hypothetical future evaluation based on the GP model. 
The new evaluation location is chosen as the global optimum of an acquisition function at each iteration of \BLFI{}. In the case of sequential Bayesian experimental design, the acquisition function is defined as an expected loss function where the loss quantifies the uncertainty associated with the unnormalised posterior (\ref{eq:unn_post}) and the expectation is taken with respect to a hypothetical future evaluation based on the current GP model. 
%This is repeated until the computational budget is depleted. %For details, see \citet{Jarvenpaa2019_sl}.
A general implementation of B(O)LFI in a concise form is outlined as Algorithm \ref{alg:blfi}. 

% GENERAL BLFI ALGORITHM 
\begin{algorithm}[htb]
\caption{General form of B(O)LFI implementation} \label{alg:blfi}
 \begin{algorithmic}[1]
 \Require Prior $\pi(\Btheta)$, Bayesian model (e.g.~GP) for the target log-likelihood $f(\Btheta)$, no.~initial evaluations $t_{\text{init}}$, no.~total evaluations $t_{\text{max}}$, no.~MCMC~samples $s_{\text{MCMC}}$ 
 \Ensure Samples $\Btheta^{(1)},\ldots,\Btheta^{(s_{\text{MCMC}})}$ % from a model-based estimate of (\ref{eq:unn_post})
  \For{$j=1:t_{\text{init}}$} \Comment{Collect data for the initial model fitting.}
 \State Select initial evaluation location $\Btheta_{j}$ \Comment{E.g.~draw $\Btheta_{j}\simiid\Unif(\Theta)$ or $\Btheta_{j}\simiid\pi(\Btheta)$.}
 \State Set $y_{j}$ $\leftarrow$ log-likelihood evaluation at $\Btheta_j$ %\Comment{Computed using e.g.~SL.}
 \EndFor
 \State Set $t \leftarrow t_{\text{init}}$ and $\ddata_t \leftarrow \{(y_{j}, \Btheta_{j})\}_{j=1}^{t}$
 % fit GP:
 \State Fit the Bayesian model for $f$ using $\ddata_t$
 %
 % MAIN LOOP:
 \For{$t=t_{\text{init}}\!+\!1:t_{\text{max}}$} \Comment{Collect data for improving the model.}
   \State Select informative evaluation location $\Btheta^*$ \label{line:blfiacq} \Comment{By optimising an acquisition function.}%\Comment{Based on the current model for $f$.}
   \State Set $y^*$ $\leftarrow$ log-likelihood evaluation at $\Btheta^*$ %\Comment{Computed using e.g.~SL.}
   \State Set $\ddata_t \leftarrow \ddata_{t-1} \cup \{(y^{*}, \Btheta^{*})\}$ 
   % fit GP:
   \State Refit Bayesian model for $f$ using $\ddata_t$
 \EndFor
 \State Sample $\Btheta^{(1)},\ldots,\Btheta^{(s_{\text{MCMC}})}$ from a model-based estimate of (\ref{eq:unn_post}) with MCMC %\Comment{Any standard MCMC method applies.}
 \end{algorithmic}
\end{algorithm} 

The approximation error can also be assessed probabilistically, at least in principle. For example, the expectation $\bar{h}$ in (\ref{eq:posterior_mean}), which is here interpreted as a functional of $f$ and rewritten as   
\begin{equation*}
\bar{h}_f \eqdef \int_{\Theta}h(\Btheta)\pi_{f}(\Btheta)\ud\Btheta, \quad \pi_{f}(\Btheta)
\eqdef \frac{\pi(\Btheta)\e^{f(\Btheta)}}{\int_{\Theta}\pi(\Btheta\pr)\e^{f(\Btheta\pr)}\ud\Btheta\pr},
%\label{eq:quantity_of_interest}
\end{equation*}
%
% %
% \begin{equation}
% \bar{h}_f 
% %
% \eqdef \int_{\Theta}h(\Btheta)\pi_{f}(\Btheta)\ud\Btheta
% %
% = \frac{\int_{\Theta}h(\Btheta)\pi(\Btheta)\e^{f(\Btheta)}\ud\Btheta}{\int_{\Theta}\pi(\Btheta\pr)\e^{f(\Btheta\pr)}\ud\Btheta\pr}. \label{eq:quantity_of_interest2}
% \end{equation}
% % 
follows a posterior distribution induced by the GP of $f\cond\ddata_t$ via the non-linear mapping $f\mapsto\bar{h}_f$. 
The resulting density $\pi(\bar{h}\cond\ddata_t)$ can be assessed numerically (see \citet{Jarvenpaa2020_babc}) though this approach is itself approximate and works only in low dimensions. %computationally challenging due to the nonlinear relationship between $\bar{h}_f$ and $f$. % which makes the aforementioned point estimation version of \BLFI{} more practical. 

Thanks to the more intelligent use of the limited budget of log-likelihood evaluations, \BLFI{} typically needs only a few hundred evaluations for reasonable posterior approximations. This is significantly less than using (pseudo-marginal/noisy) \mh{}. \BLFI{} is conceptually similar to Bayesian optimisation \cite{Hennig2012,Frazier2018,Garnett2022}, adaptive warped Bayesian quadrature \citep{Osborne2012,Gunter2014,Chai2019} and GP-based noisy level set estimation \citep{Bect2012,Lyu2021} developed for other related numerical tasks involving expensive functions.

\subsection{Other related literature} \label{subsec:literature}
% discuss probabilistic numerics, BQ etc. here?

% GP-accelerated MCMC with deterministic models/exact lik evaluations
Accelerating MCMC by using GPs or other surrogate models, also called emulators or metamodels, has been widely considered in literature, see \citet{Llorente2021} for a recent survey. % (which appeared after the first version of this paper). 
% earlier methods for exact likelihood evaluations
Notably, \citet{Rasmussen2003,Christen2005,Bliznyuk2008,Fielding2011,Conrad2016,Zhang2017,Sherlock2017,Zhang2019} develop asymptotically exact MCMC algorithms mainly in the context of exact but costly likelihood evaluations often resulting from complex ODE or PDE systems with tractable observation models. Sometimes derivative evaluations are also available to aid GP fitting \citep{Lan2016,Paun2022}. Different from these studies, we instead focus on expensive stochastic models whose likelihood function is estimated using forward simulations. We also aim for the best possible sample-efficiency (instead of merely improving over standard MCMC) while accepting some bias due to the full use of the GP surrogate. 
Related techniques that assume expensive likelihood evaluations but which are not directly based on MCMC include \citet{Kandasamy2015,Wang2018,Acerbi2018,Alawieh2020}. These methods use global GP modelling and we expect them hence to suffer from similar practical modelling challenges as B(O)LFI. Moreover, the numerical experiments by \citet{Jarvenpaa2019_sl,Acerbi2020} suggest that the active learning strategies of these papers do not work well in the noisy setting. 

% "tall data"; MCMC subsampling
Bayesian inference using \mh{} sampling in the case of ``tall data'' \citep{Korattikara2014,Angelino2016,Bardenet2017,Zhang2020talldata} is another related and likewise challenging task. While the underlying statistical model is tractable and typically relatively cheap, the very large number of data points makes likelihood evaluations costly. A key idea is to use unbiased log-likelihood evaluations obtained by subsampling the data points in the \mh{} accept/reject test. Although the existing methods are better tailored for this specific problem, our proposed approach in principle also applies there. 

% ABC, LFI context
In addition to B(O)LFI, other inference frameworks based on GP surrogate modelling have been proposed for LFI. 
% connection to Meeds and Welling 2014 UAI
Our approach most closely resembles the GPS-ABC algorithm by \citet{Meeds2014} where a related approximate \mh{} framework is considered. However, a major difference is that in GPS-ABC individual summary statistics are modelled with independent GPs in the context of Approximate Bayesian Computation (ABC) while we model the log-likelihood with GP (not necessarily in the ABC scenario). In addition, we provide substantially more comprehensive analysis of the main idea and extend it in various ways in our setting. %For instance, we develop sequential Bayesian experimental design strategies for improved sample-efficiency.
In \citet{Wilkinson2014} the difficulties with global GP modelling are partially eluded by classifying problematic parameter regions as implausible at each ``wave'' of their algorithm and fitting the GP only to its complement. However, this approach seems cumbersome and is difficult to automatise. % involving supervision while we aim for a more automatic solution. 
Finally, GP-accelerated MCMC methods in the context of noisy log-likelihood evaluations have been considered by \citet{Drovandi2018,Wiqvist2018} while variational inference is used by \citet{Acerbi2020}. However, these methods are quite convoluted featuring multiple stages and are not designed for maximal sample-efficiency.

\section{Gaussian process emulated \mh{} with noisy likelihood evaluations} \label{sec:gpe}

In this section we develop our framework for emulating the progression of an exact \mh{} when one only has access to a limited number of noisy log-likelihood evaluations. 
We model the log-likelihood using a probabilistic surrogate model and explicitly treat $\gamma(\Btheta^{(i-1)},\Btheta'^{(i)})$---and consequently also $\accratio(\Btheta^{(i-1)},\Btheta'^{(i)})$---as random variables. 
In Section \ref{subsec:approx_mh} we discuss how the \mh{} accept/reject decisions, that control the progression of the \mh{} algorithm, are made in the presence of uncertainty of the exact value of $\accratio(\Btheta^{(i-1)},\Btheta'^{(i)})$ in the light of Bayesian decision theory. Then we present a particular GP surrogate model for the log-likelihood function (Section \ref{subsec:gp_model}), combine this model with the preceding theory (Section \ref{subsec:alpha_gp}) and finally form a practical implementation of the framework (Section \ref{subsec:thealgorithm}). 
%Optimal strategies for gathering the log-lik evaluations for fitting the GP model are developed in Section \ref{sec:acq}. 

\subsection{Uncertainty in the \mh{} acceptance ratio} \label{subsec:approx_mh} \label{subsec:cond_uncond_errors}

% general theory part
\sloppy
Let us revisit the \mh{} sampler in Algorithm \ref{alg:mh}. 
%Equivalently, at each iteration we generate $u^{(i)}\sim\Unif([0,1])$ and accept the proposal $\Btheta'^{(i)}$ if $\gamma(\Btheta^{(i-1)},\Btheta'^{(i)})\geq u^{(i)}$ and reject otherwise. 
An essential initial observation is that $\alpha(\Btheta^{(i-1)},\Btheta'^{(i)})$ in line \ref{line:newtheta} can equivalently be replaced by the slightly simpler quantity $\gamma(\Btheta^{(i-1)},\Btheta'^{(i)})$. %This simplifies our analysis in the following. 
Whether to accept or reject a proposed $\Btheta'^{(i)}$ at iteration $i$, when the current point is $\Btheta^{(i-1)}$ and when there is uncertainty about the corresponding likelihood values, is treated as a problem of Bayesian decision theory. % (see e.g.~\citet{Robert2007} for background). 
Potential previous or future decisions are not taken into account for simplicity. 
Let $\hatgamma=\hatgamma(\Btheta^{(i-1)},\Btheta'^{(i)})$ be an estimator for the random variable $\gamma=\gamma(\Btheta^{(i-1)},\Btheta'^{(i)})$ for making the decision. % in the presence of uncertainty of the exact value of $\gamma$. 
We consider a loss function %$l_u(\gamma,\hatgamma)$ so that
\begin{align*}
    l_u(\gamma,\hatgamma) \eqdef \indic_{\gamma<u, \hatgamma\geq u} + \indic_{\gamma\geq u, \hatgamma<u} %\label{eq:firstloss}
\end{align*}
with a fixed $u=u^{(i)}\in[0,1]$. 
The loss is $1$ if we choose $\hatgamma\geq u$ while in reality $\gamma<u$ or if we choose $\hatgamma<u$ while $\gamma\geq u$, and $0$ otherwise. Both type of errors are hence considered equally undesirable. 
The expected loss is then
\begin{align}
\begin{split}
    \mean_{\gamma}l_u(\gamma,\hatgamma) &= \int_{\reals}(\indic_{\gamma<u}\indic_{\hatgamma\geq u} + \indic_{\gamma\geq u}\indic_{\hatgamma<u}) \ud F_{\gamma}(\gamma) \\
    &= \indic_{\hatgamma\geq u}\int_{\reals}\indic_{\gamma<u}\ud F_{\gamma}(\gamma) + \indic_{\hatgamma<u} \int_{\reals} \indic_{\gamma\geq u} \ud F_{\gamma}(\gamma) \\
    &= \prob(\gamma < u\cond u)\indic_{\hatgamma \geq u} + \prob(\gamma \geq u\cond u)\indic_{\hatgamma < u}, \label{eq:expected_loss}
\end{split}
\end{align}
where $F_{\gamma}(\gamma)$ is the cumulative distribution function (CDF) of $\gamma$. 
We also define an alternative loss function $l(\gamma,\hatgamma) \eqdef \int_0^1 l_u(\gamma,\hatgamma) \ud u$. Using Fubini's theorem we obtain the expected loss $\mean_{\gamma}l(\gamma,\hatgamma) = \int_0^1\mean_{\gamma}l_u(\gamma,\hatgamma)\ud u$, whose integrand is given by (\ref{eq:expected_loss}). 

% conditional/unconditional error which coincide expected loss
Similarly to \citet{Meeds2014}, we may also consider the probability of making an error in the \mh{} acceptance test. This is done either conditionally on $u$ so that
\begin{align}
\begin{split}
    % unconditional error, \epsi_u
    \conderr_{u,\hatgamma} &\eqdef \prob(\textnormal{``Incorrect accept/reject decision''} \cond \hatgamma, u) \\
    %
    %&= \prob(\{ \alpha < u, \hatgamma \geq u \} \cup \{ \alpha \geq u, \hatgamma < u \} \cond \hatgamma, u) \\
    %
    &= \prob(\{ \gamma < u, \hatgamma \geq u \} \cup \{ \gamma \geq u, \hatgamma < u \} \cond \hatgamma, u) \\
    &= \prob(\gamma < u, \hatgamma \geq u \cond \hatgamma, u) + \prob(\gamma \geq u, \hatgamma < u \cond \hatgamma, u) \\
    &= \prob(\gamma < u\cond u)\indic_{\hatgamma \geq u} + \prob(\gamma \geq u\cond u)\indic_{\hatgamma < u}, \label{eq:conderrgen}
\end{split}
\end{align}
or unconditionally by averaging over $u\sim\Unif([0,1])$ so that 
\begin{align}
    \unconderr_{\hatgamma} \eqdef \int_0^1\conderr_{u,\hatgamma} \Unif(u|[0,1])\ud u
    = \int_0^1\conderr_{u,\hatgamma} \ud u. \label{eq:unconderrgen}
\end{align}
We see that (\ref{eq:conderrgen}), which we refer simply as the \emph{conditional error} from now on, coincides with the expected loss (\ref{eq:expected_loss}). Similarly, the \emph{unconditional error} (\ref{eq:unconderrgen}) equals $\mean_{\gamma}l(\gamma,\hatgamma)$. 
%
% choosing estimator hatgamma
An optimal estimator ${\hatgamma}$ for making the accept/reject decision minimises the expected loss. Recall that the \emph{median} of a real-valued random variable $z$, which we denote by $\med(z)$, is defined as any value $m\in\reals$ satisfying $\prob(z\leq m)\geq 1/2$ and $\prob(z\geq m)\geq 1/2$. The median always exists but may not be unique.
\begin{proposition} \label{prop:med}
% 1st part:
Suppose $\gamma$ is a real-valued random variable.
Then the choice $\hatgamma = \med(\gamma)$ (where $\med(\gamma)$ can be any of its median values) minimises the unconditional error $\unconderr_{\hatgamma}$
% 2nd part:
and also the conditional error $\conderr_{u,\hatgamma}$ for each fixed $u\in[0,1]$. 
\end{proposition}
The proof for this and some other theoretical results are given in \appe{} \ref{appesec:proofs}. From now on we use $\hatgamma$ exclusively to denote this optimal estimator. It follows that the optimal decision is to choose the most probable action given $u$ because
\begin{align}
\begin{split}
    \conderr_{u,{\hatgamma}} &= \prob(\gamma < u\cond u)\indic_{\med(\gamma) \geq u} + \prob(\gamma \geq u\cond u)\indic_{\med(\gamma) < u} \\
    &= \min\{\prob(\gamma < u\cond u),\prob(\gamma \geq u\cond u)\}.
    \label{eq:minp1-p}
\end{split}
\end{align}
%
% connection to Meeds and Welling UAI paper
In the following sections we show that, unlike in the different surrogate modelling scenario of \citet{Meeds2014}, analytical formulas for the key quantities above can be obtained when the log-likelihood follows GP posterior.

\subsection{GP model for the log-likelihood} \label{subsec:gp_model}

We present a GP model which is similar to the one used by \citet{Jarvenpaa2019_sl} and is suitable for low-dimensional parameters ($p<10$). While we use this model in this paper, other choices may be more appropriate in some other settings. 
We first assume %a Gaussian measurement error model 
\begin{equation}
    y_j = f(\Btheta_j) + \epsilon_j, \quad \epsilon_j \sim \Normal(0,\sigma_n^2(\Btheta_j)), \quad j=1,\ldots,t,
    \label{eq:noisemodel}
\end{equation}
where $y_j\in\reals$ denotes a noisy evaluation of the log-likelihood function $f$ at some parameter $\Btheta_j\in\Theta$ and $\sigma_n^2: \Theta \rightarrow \realsp$ is the noise variance. 
The Gaussian noise assumption (\ref{eq:noisemodel}) is further discussed in \appe{} \ref{appsec:gpremarks}. % and some numerical justifications are provided by \citet{Jarvenpaa2019_sl}. 

% hierarchical GP prior for log-likelihood
We then place the following hierarchical GP prior for $f$:
\begin{align}\begin{split}
    f \cond \Bbeta \sim \GP(m_0(\Btheta),k_{\Bphi}(\Btheta,\Btheta\pr)), \quad %\\
    m_0(\Btheta) = \sum_{i=1}^q \beta_i h_i(\Btheta), \quad \Bbeta \sim \Normal(\Bb,\BB), \label{eq:gp_prior}
\end{split}\end{align}
where $k_{\Bphi}:\Theta\times\Theta\rightarrow\reals$ is a covariance (kernel) function with hyperparameters $\Bphi$ and $h_i:\Theta\rightarrow\reals$ denote fixed basis functions. The covariance function encodes the smoothness assumption of the likelihood function whereas possible prior assumptions on its shape can be encoded by specifying suitable basis functions. %Throughout the paper we assume $k_{\Bphi}, h_i$ and $\sigma_n^2$ are continuous functions. 
In our analysis we assume that $\Bphi$ and $\sigma_n^2(\Btheta)$ for each $\Btheta$ are known though in practice their values are obtained using point estimation. We omit $\Bphi$ from our notation for brevity.  
%In practice $\Bphi$ is determined using MAP estimation. We also assume that point estimates for $\sigma_n^2(\Btheta)$ are available and can be similarly used. 

% GP mean and covariance functions
As in \citet{OHagan1978,Rasmussen2006}, we integrate out $\Bbeta$ in (\ref{eq:gp_prior}). 
Given evaluations $\ddata_{\xt} \eqdef \{(y_j,\Btheta_j)\}_{j=1}^t$, the posterior of $f$ can be shown to be $f \cond \ddata_{\xt} \sim \GP(m_{\xt}(\Btheta),c_{\xt}(\Btheta, \Btheta\pr))$, where 
\begin{align*}
    % GP mean
    m_{\xt}(\Btheta) &\eqdef \Bk_{\xt}(\Btheta) \BK^{-1}_{\xt} \By_{\xt} 
    + \BR_{\xt}\T(\Btheta) \bar{\Bbeta}_{\xt}, %\label{eq:gp_mean} 
    \\ 
    % GP covariance
    \begin{split}
    c_{\xt}(\Btheta,\Btheta\pr) 
    &\eqdef k(\Btheta,\Btheta\pr) 
    - \Bk_{\xt}(\Btheta) \BK^{-1}_{\xt} \Bk\T_{\xt}(\Btheta\pr) 
    + \BR_{\xt}\T(\Btheta)[\BB^{-1} + \BH_{\xt}\BK^{-1}_{\xt}\BH_{\xt}\T]^{-1} \BR_{\xt}(\Btheta\pr), 
    \end{split} %\label{eq:gp_cov}
\end{align*}
with $[\BK_{\xt}]_{ij} \eqdef k(\Btheta_{i},\Btheta_{j}) + \indic_{i=j}\sigma_n^2(\Btheta_i)$ for $i,j=1,\ldots,t$, $\Bk_{\xt}(\Btheta) \eqdef (k(\Btheta,\Btheta_1),\ldots,k(\Btheta,\Btheta_t))$, 
$
	% gamma_bar
    \bar{\Bbeta}_{\xt} \eqdef [\BB^{-1} + \BH_{\xt}\BK^{-1}_{\xt}\BH_{\xt}\T]^{-1}(\BH_{\xt}\BK^{-1}_{\xt}\By_{\xt} + \BB^{-1}\Bb) %\\
    %
    % R
    %\BR_{\xt}(\Btheta) &\eqdef \BH(\Btheta) - \BH_{\xt}\BK^{-1}_{\xt}k_{\xt}\T(\Btheta).
$
and $\BR_{\xt}(\Btheta) \eqdef \BH(\Btheta) - \BH_{\xt}\BK^{-1}_{\xt}\Bk_{\xt}\T(\Btheta)$. 
The columns of $\BH_{\xt}\in\reals^{q\times t}$ consist of basis function values evaluated at $\Btheta_{1:t}=[\Btheta_1,\ldots,\Btheta_t]\in \reals^{p \times t}$ and $\BH(\Btheta)$ is the corresponding $q\times 1$ vector at $\Btheta$.  
We also have $\By_t=(y_1,\ldots,y_t)\T$ and we additionally denote the GP variance function as $s_{\xt}^2(\Btheta) \eqdef c_{\xt}(\Btheta,\Btheta)$.  
See \citet{Rasmussen2006} for further details on GP regression and \appe{} \ref{appsec:gpmodellingdetails} for remarks on modelling log-likelihood functions.

\subsection{Uncertainty in the \mh{} acceptance ratio based on GP surrogate} \label{subsec:alpha_gp}

We apply the analysis of Section \ref{subsec:cond_uncond_errors} on handling the uncertainty in the \mh{} accept/reject test when the log-likelihood function follows a GP posterior as in Section \ref{subsec:gp_model}. % conditioned on $\ddata_t$. 
Here $\Btheta$ denotes the current point at an arbitrary iteration of the \mh{} sampler and $\Btheta'$ is the corresponding proposal generated from $q(\Btheta'\cond\Btheta)$. 
We have 
\begin{align*}
    \begin{bmatrix}
    f(\Btheta) \\ f(\Btheta')
    \end{bmatrix} \!\cond \ddata_t
    \sim \Normal_2\left( \begin{bmatrix} m_t(\Btheta) \\ m_t(\Btheta') \end{bmatrix}, \begin{bmatrix} s_t^2(\Btheta) & c_t(\Btheta,\Btheta') \\ c_t(\Btheta,\Btheta') & s_t^2(\Btheta') \end{bmatrix} \right),
\end{align*}
which further implies 
\begin{align}
    f(\Btheta')-f(\Btheta)\cond\ddata_t \sim \Normal(m_t(\Btheta')-m_t(\Btheta), s_t^2(\Btheta')+s_t^2(\Btheta)-2c_t(\Btheta,\Btheta')).
\label{eq_fdiff}
\end{align}
Using (\ref{eq:accratio}), which we can rewrite as 
\begin{align*}
    \gamma_f(\Btheta,\Btheta') \eqdef \frac{\pi(\Btheta')\e^{f(\Btheta')}q(\Btheta\cond\Btheta')}{\pi(\Btheta)\e^{f(\Btheta)}q(\Btheta'\cond\Btheta)}
    = \frac{\pi(\Btheta')q(\Btheta\cond\Btheta')}{\pi(\Btheta)q(\Btheta'\cond\Btheta)}\e^{f(\Btheta')-f(\Btheta)},
\end{align*}
and (\ref{eq_fdiff}), it follows that $\gamma_f(\Btheta,\Btheta')$ given evaluations $\ddata_t$ follows log-Normal distribution:
\begin{align}
    \gamma_f(\Btheta,\Btheta')\cond\ddata_t &\sim \logN(\mu_t(\Btheta,\Btheta'), \sigma_t^2(\Btheta,\Btheta')), \label{eq:gammafdens} \\
    \mu_t(\Btheta,\Btheta') &\eqdef m_t(\Btheta')-m_t(\Btheta) + \log\left( \frac{\pi(\Btheta')q(\Btheta\cond\Btheta')}{\pi(\Btheta)q(\Btheta'\cond\Btheta)} \right), \label{eq:mu} \\
    \sigma_t^2(\Btheta,\Btheta') &\eqdef s_t^2(\Btheta')+s_t^2(\Btheta)-2c_t(\Btheta,\Btheta'). \label{eq:sigma}
\end{align}
%
% density of \alpha
Furthermore, $\alpha_f(\Btheta,\Btheta') \eqdef \min\{ 1, \gamma_f(\Btheta,\Btheta')\}$ given $\ddata_t$ follows a mixture density consisting of a log-Normal density in $[0,1)$ and a point mass at $1$. Its CDF is $F_{\alpha_f(\Btheta,\Btheta')\cond\ddata_t}(a) = \Phi((\log(a) - \mu_t(\Btheta,\Btheta'))/\sigma_t(\Btheta,\Btheta'))\indic_{a<1} + \indic_{a\geq 1}$ for $a>0$ and $F_{\alpha_f(\Btheta,\Btheta')\cond\ddata_t}(a)=0$ for $a\leq 0$, where $\Phi(\cdot)$ is the CDF of the standard Gaussian distribution. %The mean and variance of $\alpha_f(\Btheta,\Btheta')$ could be derived analytically but are not needed in this paper. % we defer the details to \appe{} \ref{appsec:extra_analysis}. 

% connection to error in making the accept/reject decision
Given the GP posterior of $f\cond\ddata_t$ and the optimal estimator
\begin{align}
\hatgamma = \hatgamma(\Btheta,\Btheta') =  \med_{f\cond\ddata_t}\gamma_f(\Btheta,\Btheta') = \e^{\mu_t(\Btheta,\Btheta')} \label{eq:taugpoptimal}
\end{align}
for $\gamma_f(\Btheta,\Btheta')$, the conditional and unconditional errors defined in Section \ref{subsec:cond_uncond_errors} are 
\begin{align}
    \conderr_{t,u,\hatgamma}(\Btheta,\Btheta') &= \Phi\left( -\frac{|\mu_t(\Btheta,\Btheta') - \log(u)|}{\sigma_t(\Btheta,\Btheta')} \right), \label{eq:conderr_gp} \\
    \unconderr_{t,\hatgamma}(\Btheta,\Btheta') 
    &= \int_0^1 \Phi\left( -\frac{|\mu_t(\Btheta,\Btheta') - \log(u)|}{\sigma_t(\Btheta,\Btheta')} \right) \ud u \label{eq:unconderr_gp_v0} \\
    %
    % &= \begin{cases} 
    % \Phi\left( -\frac{\mu_t(\Btheta,\Btheta')}{\sigma_t(\Btheta,\Btheta')}\right) - \e^{\mu_t(\Btheta,\Btheta')+\frac{\sigma_t^2(\Btheta,\Btheta')}{2}}\Phi\left( -\frac{\mu_t(\Btheta,\Btheta')+\sigma_t^2(\Btheta,\Btheta')}{\sigma_t(\Btheta,\Btheta')}\right), 
    % %
    % &\textnormal{if } \mu_t(\Btheta,\Btheta') \geq 0,
    % \\
    % \Phi\left( \frac{\mu_t(\Btheta,\Btheta')}{\sigma_t(\Btheta,\Btheta')}\right) + \e^{\mu_t(\Btheta,\Btheta')+\frac{\sigma_t^2(\Btheta,\Btheta')}{2}}\left[\Phi\left( -\frac{\mu_t(\Btheta,\Btheta')+\sigma_t^2(\Btheta,\Btheta')}{\sigma_t(\Btheta,\Btheta')}\right) -2\Phi(-\sigma_t(\Btheta,\Btheta')) \right],
    % %
    % &\textnormal{if } \mu_t(\Btheta,\Btheta') < 0, 
    % \end{cases}
    %
    &= \begin{cases} 
    \Phi\left( -{\mu_t}/{\sigma_t}\right) - \e^{\mu_t+{\sigma_t^2}/{2}}\Phi\left( -({\mu_t+\sigma_t^2})/{\sigma_t}\right) 
    &\textnormal{if } \mu_t \geq 0,
    \\
    \Phi\left( {\mu_t}/{\sigma_t}\right) + \e^{\mu_t+{\sigma_t^2}/{2}}\left[\Phi\left( -({\mu_t+\sigma_t^2})/{\sigma_t}\right) -2\Phi(-\sigma_t) \right]
    &\textnormal{if } \mu_t < 0, 
    \end{cases}
    \label{eq:unconderr_gp}
\end{align}
respectively. Note that we sometimes shorten $\mu_t(\Btheta,\Btheta')$ as $\mu_t$ and similarly $\sigma_t(\Btheta,\Btheta')$ as $\sigma_t$. 
Equations (\ref{eq:conderr_gp}) and (\ref{eq:unconderr_gp}) are derived in \appe{} \ref{appesec:proofs} and illustrated in Figure \ref{fig:cond_vs_uncond_errors}. 
%Equation (\ref{eq:conderr_gp}) shows that $\conderr_{t,u,\hatgamma}(\Btheta,\Btheta')\rightarrow 0$ as $|\mu_t(\Btheta,\Btheta')|\rightarrow\infty$ and $u>0$ so that whenever. Similar conclusion holds for $\unconderr_{t,\hatgamma}$.

\begin{figure}[hbtp]
\centering
\begin{subfigure}{0.32\textwidth}
\includegraphics[width=\textwidth]{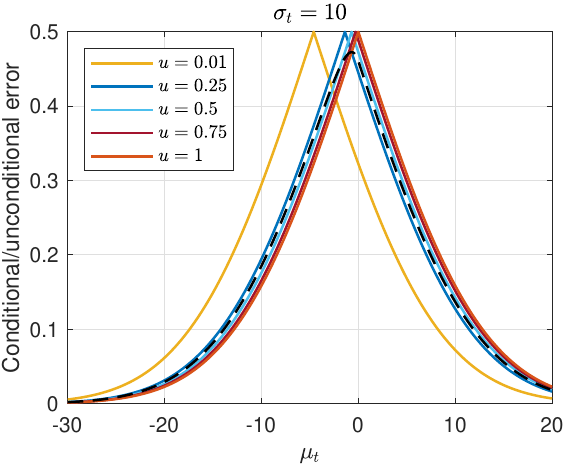}
\end{subfigure}
%\hspace{0.05cm}
\begin{subfigure}{0.32\textwidth}
\includegraphics[width=\textwidth]{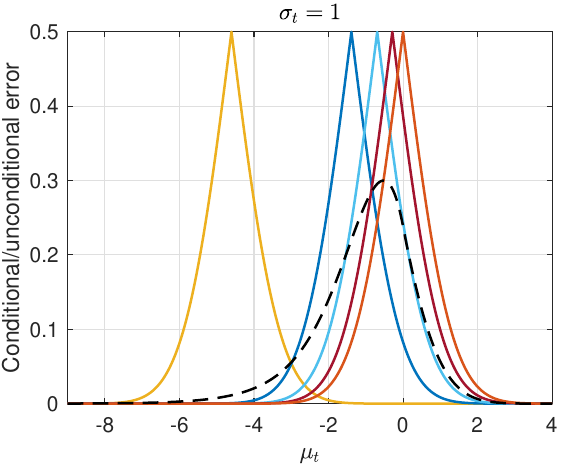}
\end{subfigure}
%\hspace{0.05cm}
\begin{subfigure}{0.32\textwidth}
\includegraphics[width=\textwidth]{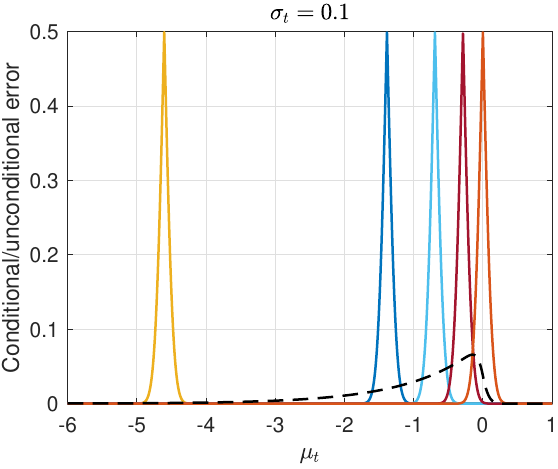}
\end{subfigure}
\caption{The coloured solid lines show the conditional error $\conderr_{t,u,\hatgamma}$ in (\ref{eq:conderr_gp}) with some choices of $u$ and the dashed black line is the unconditional error $\unconderr_{t,\hatgamma}$ in (\ref{eq:unconderr_gp}). 
%Different choices of $\mu_t$ and $\sigma_t$ are also considered. 
With each realisation of $u$, $\mu_t$ can be chosen so that $\conderr_{t,u,\hatgamma}$ equals its maximal value $\Phi(0)=1/2$. %On the other hand, $\unconderr_{t,\hatgamma}$ behaves differently. 
When the GP posterior becomes more accurate in the sense that $\sigma_t$ decreases, the intervals containing such $\mu_t$ values that lead to unconfident \mh{} accept/reject decision become narrower. Importantly, negligible (un)conditional error results when $|\mu_t|$ is large as compared to $\sigma_t$ in which case the accept/reject decision can be done with high confidence although the corresponding likelihood ratio is not necessarily known accurately.} \label{fig:cond_vs_uncond_errors}
\end{figure}

% indicator whether acceptance or not
We can also define 
\begin{align*}
   \accrejindic_{u,f}(\Btheta,\Btheta') \eqdef \indic_{\alpha_f(\Btheta,\Btheta')\geq u} = \indic_{\log\gamma_f(\Btheta,\Btheta')\geq \log(u)}
%\label{eq:kappadef}
\end{align*}
so that $\accrejindic_{u,f}(\Btheta,\Btheta')=1$ if $\Btheta'$ is to be accepted and $\accrejindic_{u,f}(\Btheta,\Btheta')=0$ otherwise (for a given $u\in[0,1]$ and $f$). 
We define $\logu\eqdef\log(u)$ and see immediately that $\prob_{f\cond\ddata_t}(\accrejindic_{u,f}(\Btheta,\Btheta') = 0) = \Phi((\logu-\mu_t(\Btheta,\Btheta'))/\sigma_t(\Btheta,\Btheta'))$ and $\mean_{f\cond\ddata_t}\accrejindic_{u,f}(\Btheta,\Btheta') = \prob_{f\cond\ddata_t}(\accrejindic_{u,f}(\Btheta,\Btheta') = 1) = \Phi((\mu_t(\Btheta,\Btheta')- \logu)/\sigma_t(\Btheta,\Btheta'))$. 
It further holds that 
\begin{align}
    % mean of \kappa
    %\mean_{f\cond\ddata_t}\accrejindic_{u,f}(\Btheta,\Btheta') &= \Phi((\mu_t(\Btheta,\Btheta')- \logu)/\sigma_t(\Btheta,\Btheta')), %\nonumber 
    %\\
    %
    % variance of \kappa
    \Var_{f\cond\ddata_t}\accrejindic_{u,f}(\Btheta,\Btheta') &= \Phi((\logu-\mu_t(\Btheta,\Btheta'))/\sigma_t(\Btheta,\Btheta'))\Phi((\mu_t(\Btheta,\Btheta')- \logu)/\sigma_t(\Btheta,\Btheta')). 
    \label{eq:kappavar}
\end{align}
We also see that 
\begin{align*}
    \conderr_{t,u,\hatgamma}(\Btheta,\Btheta') &= \min\{\prob(\accrejindic_{u,f}(\Btheta,\Btheta') = 0\cond\ddata_t), \prob(\accrejindic_{u,f}(\Btheta,\Btheta') = 1\cond\ddata_t)\},
\end{align*}
which shows that the most probable decision given the GP posterior is here made, as also implied by (\ref{eq:minp1-p}). 
Using the equations above, the fact $\Phi(z)=1-\Phi(-z)$ and the inequality $\min\{x,1-x\}\leq\sqrt{x(1-x)}$ for $x\in[0,1]$, we also see that the conditional error $\conderr_{t,u,\hatgamma}(\Btheta,\Btheta')$ is upper bounded by $({\Var_{f\cond\ddata_t}\accrejindic_{u,f}(\Btheta,\Btheta')})^{1/2}$.

\subsection{\gpmh{} implementation} \label{subsec:thealgorithm}

We combine the preceding analysis and GP model to form Algorithm \ref{alg:gpmh}. 
% maybe explain initial data collection (around the initial point \theta^(0))?
At each iteration $i$ of the resulting \gpmh{} algorithm, the proposal $\Btheta'^{(i)}$ is either accepted or rejected based on the GP posterior conditioned on the $t$ evaluations in $\ddata_t$ collected that far. The decision is made in a greedy optimal manner\footnote{This approach is greedy in the sense that the optimal decision is made at the iteration $i$ but its effect on the possible future decisions at later iterations is not taken into account.} 
based on the results in Section \ref{subsec:cond_uncond_errors} and \ref{subsec:alpha_gp}. In a similar spirit to \citet{Meeds2014,Korattikara2014}, new log-likelihood evaluations are acquired (lines \ref{line:acq}-\ref{line:sleval2}) until the estimated probability of making an incorrect \mh{} accept/reject decision is smaller than a pre-specified tolerance parameter $\epsi$ (line \ref{line:acc_check}). In most iterations, this is achieved without new log-likelihood evaluations which facilitates computational savings. 
Details on the selection of the evaluation locations %(also done in a greedy optimal manner) 
(line \ref{line:acq}) are given in Section \ref{sec:acq}. % but let us assume now that this can be done in a reasonable way, e.g.~either $\Btheta^{(i-1)}$ or $\Btheta'^{(i)}$ is selected. 
The GP model is updated after each new evaluation (line \ref{line:gpupd}). 
The outputted samples are finally used to approximate the posterior or some posterior expectations of interest via (\ref{eq:posterior_mean_mcmc_approx}). 
%We refer to the implementation in Algorithm \ref{alg:gpmh} as \gpmh{}. 
% notes
Note that two distinct parameter sets are maintained in the algorithm: $\Btheta_j$ in $\ddata_t$ denote the evaluation locations for the GP fit whereas $\Btheta^{(i)}$ denote the resulting approximate \mh{} samples. 
%Some implementation details are not explicitly shown in Algorithm \ref{alg:gpmh} for clarity. For example, if $\pi(\Btheta'^{(i)})=0$ then the while loop is skipped altogether and we set $\Btheta^{(i)} \leftarrow \Btheta^{(i-1)}$ on line \ref{line:accorrej} without evaluating $\hatgamma$. 
% handling NaN outputs etc.
Handling of possible problematic log-likelihood evaluations on lines \ref{line:sleval1} and \ref{line:sleval2} are described in \appe{} \ref{appsec:gpmodellingdetails}. 

% MAIN ALGORITHM
\begin{algorithm}[htb]
\caption{Approximate GP-emulated \mh{} (\gpmh{})} \label{alg:gpmh}
 \begin{algorithmic}[1]
 \Require Prior $\pi(\Btheta)$, GP model for ${f}$, no.~initial evaluations $t_{\text{init}}$, error tolerance $\epsi$, initial point $\Btheta^{(0)}$, proposal $q(\Btheta'\cond\Btheta)$, no.~MH~samples $i_{\text{MH}}$ 
 \Ensure Approximate \mh{} samples $\Btheta^{(1)},\ldots,\Btheta^{(i_{\text{MH}})}$
 %
 % INITIAL:
 \For{$j=1:t_{\text{init}}$} \Comment{Obtain evaluations for the initial GP fitting.} \label{line:line1}\label{line:initstart}
 \State Sample $\Btheta_{j} \simiid q(\cdot\cond\Btheta^{(0)})$ \Comment{Other initial points can also be used.\,\,}
 \State Set $y_{j}$ $\leftarrow$ log-likelihood evaluation at $\Btheta_j$ %\Comment{Computed using e.g.~SL.} 
 \label{line:sleval1}
 \EndFor \label{line:line4}
 \State Set $t \leftarrow t_{\text{init}}$ and $\ddata_t \leftarrow \{(y_{j}, \Btheta_{j})\}_{j=1}^{t}$
 % fit GP:
 \State Fit GP using $\ddata_t$ \label{line:initend}
 %
 % MAIN LOOP:
 \For{$i=1:i_{\text{MH}}$} 
  \State Sample $\Btheta'^{(i)} \sim q(\cdot\cond\Btheta^{(i-1)})$ and $u^{(i)}\sim\Unif([0,1])$ \label{line:newthetagpmh}
  \While{$\unconderr_{t,\hatgamma}(\Btheta^{(i-1)},\Btheta'^{(i)}) > \epsi$} \label{line:unconderr_comp} %%\Comment{Alternatively, $\unconderr_{t,\hatgamma}$ is used with $\epsi$ set accordingly.}
  \Comment{Alternatively, use $\conderr_{t,u^{(i)},\hatgamma}$.}
  \label{line:acc_check}
    % collect new data!
   \State Obtain $\Btheta^*$ as a solution to (\ref{eq:xiopt}) \Comment{{See Section \ref{sec:acq}}.\,\,}%\Comment{{New location for evaluating log-likelihood, see Section \ref{sec:acq}}.} 
   \label{line:acq}
   \State Set $y^*$ $\leftarrow$ log-likelihood evaluation at $\Btheta^*$ %\Comment{Computed using e.g.~SL.} 
   \label{line:sleval2}
   \State Set $t \leftarrow t+1$ and $\ddata_{t} \leftarrow \ddata_{t-1} \cup \{(y^{*}, \Btheta^{*})\}$
   % fit GP:
   \State Refit GP using $\ddata_{t}$ \label{line:gpupd} %\Comment{Alternatively, update $\Bphi$ only on some iterations.}
  \EndWhile
 \State Set $\Btheta^{(i)} \leftarrow \Btheta'^{(i)}\indic_{\hatgamma \geq u^{(i)}} + \Btheta^{(i-1)}\indic_{\hatgamma < u^{(i)}}$ \Comment{{Accept/reject $\Btheta'^{(i)}$; $\hatgamma$ computed using (\ref{eq:taugpoptimal}).}} \label{line:accorrej}
 \EndFor
 \label{line:post}
 \end{algorithmic}
\end{algorithm}

% proposal, adaptive Metropolis etc. 
In this paper we consider a random-walk Metropolis version of \gpmh{} with $q(\Btheta'\cond\Btheta)=\Normal_p(\Btheta'\cond\Btheta,\BSigma)$. 
It is often difficult to select a suitable proposal covariance $\BSigma$ \textit{a priori}. 
%Sometimes an optimisation algorithm is first employed to locate the MAP estimate to be used as an initial point $\Btheta^{(0)}$ and the Hessian at this point is further used to form a suitable $\BSigma$. In our setting this can be cumbersome and costly. 
In our implementation we specify an initial covariance matrix $\BSigma_{0}$ and update it based on the obtained samples as in the adaptive Metropolis algorithm by \citet{Haario2001}. 
We use the initial proposal density $\Normal_p(\Btheta'\cond\Btheta,\BSigma_0)$ also to obtain evaluations around $\Btheta^{(0)}$ for initial GP fitting (lines \ref{line:line1}-\ref{line:line4}). % although other initialisation strategies may be more suitable. 
Contrary to a typical MCMC use case, possible poor mixing is not a major concern in our setting where the parameter space is low-dimensional, time spent on evaluating the log-likelihood dominates and the \mh{} accept/reject decision is based solely on the GP on most iterations. Finding a good initial location and proposal covariance is still beneficial and pilot runs may be needed. 
% Maybe include for further details on how the covariance matrix is updated?
% mention that many GP quantities can be precomputed?

\section{One-step ahead optimal evaluation locations} \label{sec:acq}

% general
We choose the evaluation locations in a one-step ahead optimal manner in the sense of Bayesian experimental design theory, see e.g.~\citet[Chapter~5]{Garnett2022} for background. In this so-called ``myopic'' strategy only the effect of the next log-likelihood evaluation (or a batch of evaluations) is taken into account while potential additional future evaluations needed to make the (un)conditional error eventually smaller than $\epsi$ or possible later \mh{} transitions are not. This common strategy substantially simplifies the computations. % so that we consider only this one-step ahead strategy. 
We denote a collection of candidate evaluation locations as $\Btheta^*\in\reals^{p\times b}$ and the corresponding log-likelihood evaluations as $\By^*\in\reals^b$, where $b\geq 1$ is the batch size, that is, the number of simultaneous evaluations. We also denote $\ddata^* \eqdef \{(y_j^*,\Btheta_j^*)\}_{j=1}^b$. 
Algorithm \ref{alg:gpmh} is stated for the sequential case $b=1$ but we present the theory for the more general batch case $b\geq 1$ as this comes with little additional difficulty. %While using $b>1$ would allow concurrent log-likelihood evaluations, this approach is challenging because a high-dimensional global optimisation is needed when $b$ is large and possibly wasteful evaluations may be obtained whenever a single evaluation would be enough to make the error smaller than $\epsi$. 
%Detailed investigation of this, as well as the possibility of taking into account the potential future \mh{} transitions instead of just the current one, are left for future work.
%Detailed investigation of this is left for future work.

Specifically, such evaluation location(s) $\Btheta^*$ are treated optimal that minimise the expected loss where the loss function quantifies the uncertainty associated with the current \mh{} accept/reject decision. Suitable choices include the conditional and unconditional errors. The expectation is taken with respect to the predictive density of the future log-likelihood evaluation(s) at $\Btheta^*$ given by $\pi(\By^*\cond\Btheta^*,\ddata_t) = \Normal_b(\By^*\cond m_t(\Btheta^*),c_t(\Btheta^*,\Btheta^*) + \BLambda^{\!*})$ where $\BLambda^{\!*} \eqdef \diag(\sigma_n^2(\Btheta^*_1),\ldots,\sigma_n^2(\Btheta^*_b))$. %As the loss function we may choose the conditional or unconditional error, for example. $\mean_{\By^*\cond\Btheta^*,\ddata_t} \conderr_{t+b,u,\hatgamma}(\Btheta,\Btheta')$, expected unconditional error

\begin{proposition} \label{prop:error_formulas}
Suppose the GP model in Section \ref{subsec:gp_model} holds and consider the above set-up. 
The expected conditional error $L_t^{\conderr,u}(\Btheta,\Btheta';\Btheta^*)$, the expected unconditional error $L_t^{\unconderr}(\Btheta,\Btheta';\Btheta^*)$ and the expected variance of $\accrejindic_{u,f}$ denoted by $L_t^{\textnormal{v}}(\Btheta,\Btheta';\Btheta^*)$ are then given by  
\begin{align}
% conditional error:
    L_t^{\conderr,u}(\Btheta,\Btheta';\Btheta^*) &\eqdef \mean_{\By^*\cond\Btheta^*,\ddata_t} \conderr_{t+b,u,\hatgamma}(\Btheta,\Btheta') 
    = 2T\left( \frac{\logu-\mu_t(\Btheta,\Btheta')}{\sigma_t(\Btheta,\Btheta')}, \frac{\sqrt{\sigma_t^2(\Btheta,\Btheta') - \xi_t^2(\Btheta,\Btheta';\Btheta^*)}}{ \xi_t(\Btheta,\Btheta';\Btheta^*)} \right), \nonumber %\label{eq:exp_conderr} 
    \\
%
% unconditional error:
    L_t^{\unconderr}(\Btheta,\Btheta';\Btheta^*) &\eqdef \mean_{\By^*\cond\Btheta^*,\ddata_t} \unconderr_{t+b,\hatgamma}(\Btheta,\Btheta') 
    = \int_0^1 L_t^{\conderr,u}(\Btheta,\Btheta';\Btheta^*) \ud u, \label{eq:exp_unconderr} \\
%
% expected variance:
\begin{split}
    L_t^{\textnormal{V},u}(\Btheta,\Btheta';\Btheta^*) &\eqdef \mean_{\By^*\cond\Btheta^*,\ddata_t} \Var_{f\cond\ddata_t\cup\ddata^*}\accrejindic_{u,f}(\Btheta,\Btheta') \\
    &= 2T\left( \frac{\logu-\mu_t(\Btheta,\Btheta')}{\sigma_t(\Btheta,\Btheta')}, \sqrt{\frac{\sigma_t^2(\Btheta,\Btheta') - \xi_t^2(\Btheta,\Btheta';\Btheta^*)}{\sigma_t^2(\Btheta,\Btheta') + \xi_t^2(\Btheta,\Btheta';\Btheta^*)}} \right), \label{eq:exp_var}
\end{split}
\end{align}
respectively. Above $T(\cdot,\cdot)$ denotes the \owen{} \citep{Owen1956,Owen1980} and
\begin{align}
    \xi_t^2(\Btheta,\Btheta';\Btheta^*) &= \tau_t^2(\Btheta;\Btheta^*) + \tau_t^2(\Btheta';\Btheta^*) - 2\omega_t(\Btheta,\Btheta';\Btheta^*), \label{eq:xi} \\
    %
    % where:
    %\tau_t^2(\Btheta_{\bu};\Btheta^*) &= c_t(\Btheta_{\bu};\Btheta^*)[c_t(\Btheta^*,\Btheta^*) + \BLambda^{\!*}]^{-1}c_t(\Btheta^*,\Btheta_{\bu}), \label{eq:gptau} \\
    %
    \tau_t^2(\Btheta_{\bu};\Btheta^*) &= \omega_t(\Btheta_{\bu},\Btheta_{\bu};\Btheta^*), \label{eq:gptau} \\ %\quad  
    \omega_t(\Btheta,\Btheta';\Btheta^*) &= c_t(\Btheta,\Btheta^*)[c_t(\Btheta^*,\Btheta^*) + \BLambda^{\!*}]^{-1}c_t(\Btheta^*,\Btheta'). \label{eq:gpomega} % \\
    %
    %\BLambda^{\!*} &= \diag(\sigma_n^2(\Btheta^*_1),\ldots,\sigma_n^2(\Btheta^*_b)).
\end{align}
\end{proposition}
Note that above $\conderr_{t+b,u,\hatgamma}(\Btheta,\Btheta')$ and $\unconderr_{t+b,\hatgamma}(\Btheta,\Btheta')$ both depend on $\By^*$ and $\Btheta^*$ via $\ddata^*$ although this is not explicitly shown in the notation. 

%The above strategies are designed to sequentially maximally reduce the uncertainty regarding the \mh{} accept/reject decision. %which can frequently be done accurately even if substantial uncertainty about the exact value of $\gamma_f(\Btheta,\Btheta')$ remains. 
Another strategy would be to select such $\Btheta^*$ that minimises the future variance of the log-\mh{} ratio $\Var_{f\cond\ddata_t\cup\ddata^*}\log\gamma_f(\Btheta,\Btheta') = \sigma^2_{t+b}(\Btheta,\Btheta')=\sigma_t^2(\Btheta,\Btheta') - \xi_t^2(\Btheta,\Btheta';\Btheta^*)$ as this acquisition function conveniently does not depend on the unknown $\By^*$. % but only on $\Btheta^*$. Its formula follows directly from (\ref{eq:sigma}) and the proof of Proposition \ref{prop:error_formulas} in \appe{} \ref{appesec:proofs}. 
Interestingly, $\Var_{f\cond\ddata_t\cup\ddata^*}\log\gamma_f(\Btheta,\Btheta')$ and all three acquisition functions of Proposition \ref{prop:error_formulas} share the same global minimiser: % $\Btheta_{\textnormal{opt}}$:
\begin{proposition} \label{prop:xi}
\sloppy
The global minimum $\Btheta_{\textnormal{opt}}\in\Theta^b$ for the variance of the log-\mh{} ratio $\Var_{f\cond\ddata_t\cup\ddata^*}\log\gamma_f(\Btheta,\Btheta')$, for the expected conditional error $L_t^{\conderr,u}(\Btheta,\Btheta';\Btheta^*)$ with any $u\in[0,1]$, 
for the expected unconditional error $L_t^{\unconderr}(\Btheta,\Btheta';\Btheta^*)$ and also for the expected variance $L_t^{\textnormal{v},u}(\Btheta,\Btheta';\Btheta^*)$ of $\accrejindic_{u,f}$ with any $u\in[0,1]$, is given by 
\begin{align}
    \Btheta_{\textnormal{opt}} \in \arg\max_{\Btheta^*\in\Theta^b}\xi_t^2(\Btheta,\Btheta';\Btheta^*).
    \label{eq:xiopt}
\end{align}
\end{proposition}

We interpret $\arg\max_{\Btheta^*\in\Theta^b}\xi_t^2(\Btheta,\Btheta';\Btheta^*)$ as a set because the minimiser in (\ref{eq:xiopt}) may not be unique. 
When $b=1$ and $\Theta=\prod_{i=1}^p[a_i,b_i]$ where we allow $a_i=-\infty$ and $b_i=\infty$, the optimisation task (\ref{eq:xiopt}) could be simplified by replacing $\Theta$ with some bounded set $\tilde{\Theta}\subset\Theta$ located around $\Btheta$ and $\Btheta'$. For example, we can choose
\begin{equation}
    \tilde{\Theta} = \prod_{i=1}^p [\max\{\min\{\theta_i,\theta_i'\}-cl_i,a_i\},\min\{\max\{\theta_i,\theta_i'\}+cl_i,b_i\}], 
    \label{eq:opt_example_region}
\end{equation}
where $l_i$'s denote the lengths-cales of the GP covariance function and $c>0$ controls the size of the set. 
Of course, the set (\ref{eq:opt_example_region}) is not guaranteed to contain the global optimum unless $c$ is large enough. Another tempting choice is $\tilde{\Theta} = \{\Btheta,\Btheta'\}$
%
%\begin{equation}
%    \tilde{\Theta} = \{\Btheta,\Btheta'\} \label{eq:r}
%\end{equation}
%
which simplifies the implementation and facilitates fast computations because no auxiliary global optimisation is then needed. %This approach is especially helpful when some of the parameters are discrete. 
%However, as discussed in Section \ref{sec:theory}, the global optimum often does not belong to $\{\Btheta,\Btheta'\}$ so that poorer sample-efficiency can be expected. We investigate this empirically in Section \ref{sec:exp}. 

%Finally, we note that Algorithm \ref{alg:gpmh} features a built-in stopping rule: New log-likelihood evaluations are computed only if the uncertainty for making the accept/reject decision is too large. Eventually we can expect that no more evaluations are needed.
% could explain EPoE etc already here

\section{Interpretations of \gpmh{}} \label{sec:probinterp}
% - approximate MCMC
% - Bayesian LFI method where MH is implicitly used for data collection
% - heuristic estimate in the conceptual "Bayesian" MH framework

In Section \ref{subsec:mhasspecialcaseofgpmh} we reason that the standard and noisy/pseudo-marginal \mh{} samplers are sorts of special cases of the \gpmh{} framework when an uninformative GP prior is used. 
%We then discuss probabilistic interpretations of our \gpmh{} algorithm: 
In Section \ref{subsec:gpmcmcasblfi} we show a concrete relation between \gpmh{} and the \BLFI{} frameworks. This motivates an alternative, two-stage implementation of \gpmh{} called \mhblfi{}. %We also represent a reformulated algorithm where the fitted GP model is explicitly used to form an estimator for the posterior and the emulation of \mh{} accept/reject decisions implicitly defines an adaptive stochastic strategy with an embedded stopping rule for data collection. 
In Section \ref{subsec:probnum} we argue that when the randomness of the \mh{} sampler in Algorithm \ref{alg:mh} is considered fixed, the sampler can be treated as a deterministic mapping whose input is the likelihood function and whose output is a (multi)set of parameter values. The uncertainty in the output due to the use of the GP model in the place of the exact likelihood function can then be probabilistically quantified in principle. \gpmh{} can be viewed as a heuristic yet tractable estimate for the \mh{} output under this conceptual ``Bayesian'' \mh{} setting. %(in a somewhat similar spirit to probabilistic numerics methods \citep{Hennig2015,Cockayne2017}). 

\subsection{Standard and noisy/pseudo-marginal \mh{} as special cases of \gpmh{}} \label{subsec:mhasspecialcaseofgpmh}

Throughout this section the log-likelihood function $f$ is modelled with a white noise process. That is, the GP prior has mean $m_0(\Btheta)=0$ and covariance function $k_{\Bphi}(\Btheta,\Btheta\pr) = \sigma_s^2\indic_{\Btheta=\Btheta\pr}$ with $\Bphi=\sigma_s^2>0$. We further assume $\sigma_s^2$ is set arbitrarily large. 

% case 1,2
First, suppose $\sigma_n^2(\Btheta)=0$ for all $\Btheta$ and $0 \leq \epsi<1/2$. The evaluation at the current point $\Btheta$ is automatically reused from the previous iteration via the GP model\footnote{This holds also at the first \gpmh{} iteration if the initial GP fitting consists of a single evaluation at the initial point $\Btheta=\Btheta^{(0)}$. Other initialisations would not be useful anyway in this setting.} and is not recomputed. If the proposed point $\Btheta'$ has not been visited before, \gpmh{} algorithm needs to evaluate\footnote{Technically, this does not hold when $\pi(\Btheta')=0$. Although not explicitly shown in Algorithm \ref{alg:gpmh}, in such a case $\Btheta'$ is rejected without evaluating the (un)conditional error or the log-likelihood which is in line with the standard \mh{} case.} at $\Btheta'$ because otherwise the (un)conditional error would remain $1/2$. All of our acquisition functions will select $\Btheta'$ in this case and the (un)conditional error becomes $0$ after this exact evaluation, see Section \ref{sssec:xi} for further details on this. % so that a single evaluation suffices. 
If $\Btheta'$ has been visited before, its evaluation is also automatically reused via the GP and the (un)conditional error is $0$. Hence, we conclude that \gpmh{} reverts to standard \mh{} where the GP model has no role but to cache the old evaluations for reuse.

% case 3,4
Suppose next that the evaluations are noisy so that $\sigma_n^2(\Btheta)>0$ for all $\Btheta$ and $0<\epsi<1/2$. 
Noisy/pseudo-marginal \mh{} is then technically not a special case of Algorithm \ref{alg:gpmh} but falls under a slightly generalised version of our framework where one (and just one) noisy evaluation at $\Btheta'$ is acquired even when the resulting (un)conditional error is not smaller than $\epsi$. 
The GP model again caches the old evaluations. A disagreement however arises if the proposal $\Btheta'$ has been visited before. Then noisy/pseudo-marginal \mh{} would use a new noisy evaluation at $\Btheta'$ whereas the old evaluation at $\Btheta'$ would also be taken into account via the GP model in \gpmh{}. This disagreement however occurs with probability zero in the common case where $\Theta$ is uncountable and the proposal $q$ does not contain any point masses.
Finally, we observe that Algorithm \ref{alg:gpmh} with the white noise prior and $0<\epsi<1/2$ can be seen as a variant of the algorithm by \citet{Korattikara2014} where Bayesian analysis is used instead of their frequentist hypothesis testing approach to determine each \mh{} accept/reject decision.

\subsection{The relation between \gpmh{} and \BLFI{}} \label{subsec:gpmcmcasblfi}

Suppose no new log-likelihood evaluations are needed in Algorithm \ref{alg:gpmh} after some iteration $i$ because any possible later transition can be done confidently enough based on the current GP model. %, either because they are not needed or because a hard threshold on the maximum number of evaluations is additionally introduced, 
From iteration $i$ onwards Algorithm \ref{alg:gpmh} then acts as an exact \mh{} sampler that targets a density proportional to
\begin{equation}
    \med_{f\cond\ddata_t}(\tildepiapprox_f(\Btheta)) = \pi(\Btheta)\e^{m_t(\Btheta)}. \label{eq:medestim}
\end{equation}
This fact follows because in this case $\ddata_t$ does not change after iteration $i$ and because
\begin{align*}
%\begin{split}
    \hatgamma = \med_{f\cond\ddata_t}(\gamma_f(\Btheta,\Btheta')) &= \e^{\mu_t(\Btheta,\Btheta')}
    %
    %= \e^{m_t(\Btheta') - m_t(\Btheta) + \log\left( \frac{\pi(\Btheta')q(\Btheta\cond\Btheta')}{\pi(\Btheta)q(\Btheta'\cond\Btheta)} \right)} \\
    %
    = \frac{\pi(\Btheta')\e^{m_t(\Btheta')}q(\Btheta\cond\Btheta')}{\pi(\Btheta)\e^{m_t(\Btheta)}q(\Btheta'\cond\Btheta)}
    = \frac{\med_{f\cond\ddata_t}(\tildepiapprox_f(\Btheta'))q(\Btheta\cond\Btheta')}{\med_{f\cond\ddata_t}(\tildepiapprox_f(\Btheta))q(\Btheta'\cond\Btheta)}.
%\end{split}
\end{align*}
This situation occurs also if a hard threshold $t_{\text{max}}$ on the maximum number of log-likelihood evaluations, after which no more evaluations are collected even if some later proposed transition would be unconfident, is introduced to Algorithm \ref{alg:gpmh}.

\subsubsection{\gpmh{} as a \BLFI{} method with a stochastic acquisition function}

The above remark suggests an alternative, two-stage variant of \gpmh{}:~1) Run Algorithm \ref{alg:gpmh} until $t_{\text{max}}$ evaluations have been gathered or some other stopping criterion is met. 
2) Treat the estimate (\ref{eq:medestim}) based on the evaluations collected in stage 1 as an approximation to the unnormalised posterior and plug it in to some standard MCMC sampler. 
%Any standard MCMC method can be used to finally sample from (\ref{eq:medestim}). 
We write this idea in the form of Algorithm \ref{alg:mhblfi} which in fact is a special case of \BLFI{} (compare Algorithm \ref{alg:mhblfi} to Algorithm \ref{alg:blfi}!) and we call it \mhblfi{}. 
% stochastic acq function interpretation
The tail-recursive procedure \textproc{\AcqParams} together with line (\ref{line:mhblfiacq}) is interpreted as a sequential stochastic strategy for selecting the evaluation locations in the \BLFI{} framework. The recursion never completes if the condition $\unconderr_{t,\hatgamma}(\Btheta,\Btheta') \leq \epsi$ holds for all $\Btheta,\Btheta'\in\Theta$ so in practice Algorithm \ref{alg:mhblfi} would need to be further modified to terminate prematurely if the recursion becomes too deep. %Alternatively, $\epsi$ could be decreased adaptively.

% MAIN ALGORITHM --- THE OTHER FORMULATION
\begin{algorithm}[htb]
\caption{\gpmh{} reformulated in \BLFI{} framework (MH-BLFI)} \label{alg:mhblfi}
 \begin{algorithmic}[1]
 \Require Prior $\pi(\Btheta)$, GP model for ${f}$, error tolerance $\epsi$,  no.~initial evaluations $t_{\text{init}}$, no.~total evaluations $t_{\text{max}}$, initial point $\Btheta^{(0)}$, proposal $q(\Btheta'\cond\Btheta)$, no.~MCMC~samples $s_{\text{MCMC}}$ 
 \Ensure Samples $\Btheta^{(1)},\ldots,\Btheta^{(s_{\text{MCMC}})}$ % from the model-based estimate of (\ref{eq:unn_post})
 \item[{\footnotesize{1-6}}:\!\hspace{0.15cm}{Obtain initial GP and $\ddata_{t_{\text{init}}}$}] \Comment{Lines 1-6 are the same as those in Algorithm \ref{alg:gpmh}.}
\makeatletter
\setcounter{ALG@line}{6} % numbering now starts from 7
\makeatother
 \State Set $\Btheta \leftarrow \Btheta^{(0)}$ and sample $\Btheta' \sim q(\cdot\cond\Btheta)$ and $u\sim\Unif([0,1])$
 % MAIN LOOP:
 \For{$t=t_{\text{init}}\!+\!1:t_{\text{max}}$} 
  \State Set $(\Btheta,\Btheta',u) \leftarrow$ \textproc{\AcqParams}($\Btheta,\Btheta',u$)
   \State Obtain $\Btheta^*$ as a solution to (\ref{eq:xiopt}) using ($\Btheta,\Btheta',u$) \Comment{{As in Algorithm \ref{alg:gpmh}}.\label{line:mhblfiacq}}%\Comment{New location for evaluating log-likelihood.}
   \State Set $y^*$ $\leftarrow$ log-likelihood evaluation at $\Btheta^*$ %\Comment{Computed using e.g.~SL.}
   \State Set $\ddata_t \leftarrow \ddata_{t-1} \cup \{(y^{*}, \Btheta^{*})\}$ 
   % fit GP:
   \State Refit GP using $\ddata_t$
   %\Comment{Alternatively, update $\Bphi$ only on some iterations.}
 \EndFor
 \State Sample $\Btheta^{(1)},\ldots,\Btheta^{(s_{\text{MCMC}})}$ from (\ref{eq:medestim}) with MCMC \Comment{Alternatively, use (\ref{eq:mode_post_estim}).}
 \item[] % produces empty line, no number
 \Procedure{\AcqParams}{$\Btheta,\Btheta',u$} 
 %\Comment{GP and $\epsi$ are also available for the procedure.}
 \label{line:acqproc}
 \If{$\unconderr_{t,\hatgamma}(\Btheta,\Btheta') > \epsi$} %\Comment{Alternatively, $\conderr_{u^{(i)},\hatgamma}$ or $\Var_{f\cond\ddata}(\alpha_f)$ can be used.}
 \State \textbf{Return} $(\Btheta,\Btheta',u)$
 \ElsIf{$\hatgamma(\Btheta,\Btheta') \geq u$} \Comment{{Proposal $\Btheta'$ is accepted. $\hatgamma$ computed using (\ref{eq:taugpoptimal}).}}
 \State $\Btheta \leftarrow \Btheta'$
 \EndIf \Comment{Only the latest accepted $\Btheta$ needs to be stored.}
 \State Sample $\Btheta' \sim q(\cdot\cond\Btheta)$ and $u\sim\Unif([0,1])$
 \State \textbf{Return} \textproc{\AcqParams}($\Btheta,\Btheta',u$)
 \EndProcedure
 \end{algorithmic}
\end{algorithm}

A key difference between \gpmh{} in Algorithm \ref{alg:gpmh} and \mhblfi{} in Algorithm \ref{alg:mhblfi} is that in the former approach the GP model is continuously refined. % unlike in the latter approach which is a two-stage approach just like B(O)LFI. % where the estimator for the unnormalised posterior is first formed and only finally plugged-in to some auxiliary MCMC sampler. 
A potential concern of the former approach is that its convergence as an approximate \mh{} sampler is problematic to assess because the target density is slightly changing. %(Our experiments in Section \ref{sec:exp} however suggest that this is not a major concern.)
% the number of evals random
On the other hand, the fact that the GP surrogate can no longer be improved in the second stage of the latter approach can also be problematic. 
See also \citet[Section~3]{Llorente2021} for some generic discussion on such approaches. % for accelerating MCMC using surrogate models. 
% another difference: number of evals random <-> fixed ub
Another important difference is that Algorithm \ref{alg:gpmh} is run for $i_{\text{MH}}$ iterations and the number of log-likelihood evaluations is not known in advance whereas Algorithm \ref{alg:mhblfi} uses $t_{\text{max}}$ evaluations (or possibly less if an additional stopping criterion is used). 
%
%That is, the approximate \mh{} samples only take the role of ``input parameters'' $(\Btheta,\Btheta',u)$ for \textproc{\AcqParams}.

%% theoretical validity of MCMC not so relevant
%A potential advantage of Algorithm \ref{alg:mhblfi} is that the convergence of the implicit approximate \mh{} chain is not strictly required. As long as the highest density region of the posterior is sufficiently explored, the resulting simulation locations can be expected to result in a reasonable GP-based approximation for the posterior irrespective of whether the implicit chain is efficiently generating samples from the target. This is in contrast to standard MCMC methods and Algorithm \ref{alg:gpmh} which are typically run until convergence, assessing of which is however not straightforward in practice. Of course, if the implicit M-H is far from convergence, then some parts of the high-density region have likely not yet been visited and the resulting GP-based estimator can also be poor. 

\subsubsection{Robust estimator of the unnormalised posterior} \label{subsubsec:post_estimator}

%% point estimator & issue with uncertainty in the boundaries
An explicit estimator for the unnormalised posterior $\tildepiapprox_f(\Btheta)$ is needed in a two-stage approach such as \mhblfi{} in Algorithm \ref{alg:mhblfi}. The choice of this estimator can be based on Bayesian decision theory. As discussed by \citet{Jarvenpaa2020_babc}, the (marginal) median
$\hat{\tpi}_1(\Btheta) = \med_{f\cond\ddata_t}(\tildepiapprox_f(\Btheta)) = \pi(\Btheta)\e^{m_t(\Btheta)}$,
which was also already shown as (\ref{eq:medestim}), minimises the expected $L^1$-loss so that  $\hat{\tpi}_1(\Btheta) = \arg\min_{\tpi}\mean_{f\cond\ddata_t}\tilde{l}_1(\tildepiapprox_f,\tpi)$ where $\tilde{l}_1(\tildepiapprox_f,\tpi) = \int_{\Theta}|\tildepiapprox_f(\Btheta)-\tpi(\Btheta)|\ud\Btheta$. 
% explain that post mean follows from similar analysis?
This estimator is intuitive and easy to compute. % and not affected by the heavy tail of the log-GP of $\tildepiapprox_f(\Btheta)$. %; the intractable log-likelihood function is simply replaced by the GP mean function $m_t(\Btheta)$. 

% issue with med
Similarly to standard MCMC methods, the parameter regions far in the posterior tails are typically not extensively explored when Algorithm \ref{alg:gpmh} and \ref{alg:mhblfi} are used and so $\ddata_t$ mostly contains points in the highest density region. % which is often desirable. 
Consequently, the uncertainty of the likelihood function can remain large in such unexplored tail regions and the resulting GP-based estimator for $\tildepiapprox_f(\Btheta)$, such as the median (\ref{eq:medestim}), can unintuitively have a non-negligible value there. % as $\tildepiapprox_f(\Btheta)$ follows a log-GP with heavy right tail. 
This can produce difficult target densities for MCMC in the second stage of a two-stage approach such as \mhblfi{}. The issue has also been observed by \citet{Fielding2011,Drovandi2018} and is illustrated in our set-up in Figure \ref{fig:loglik_gp1}: The estimated SL posterior, the red line in Figure \ref{fig:loglik_gp1}b, has two distant modes. The possibility of the second mode at $\theta=30$ cannot be excluded based on the GP model fitted to the $9$ noisy evaluations. % but some further information needs to be incorporated. % which makes it a tricky target for MCMC. %as some MCMC chains might get stuck sampling in the boundary. 

\begin{figure}[hbtp] % loglik modelling/estimator demo
\centering
\includegraphics[width=0.8\textwidth]{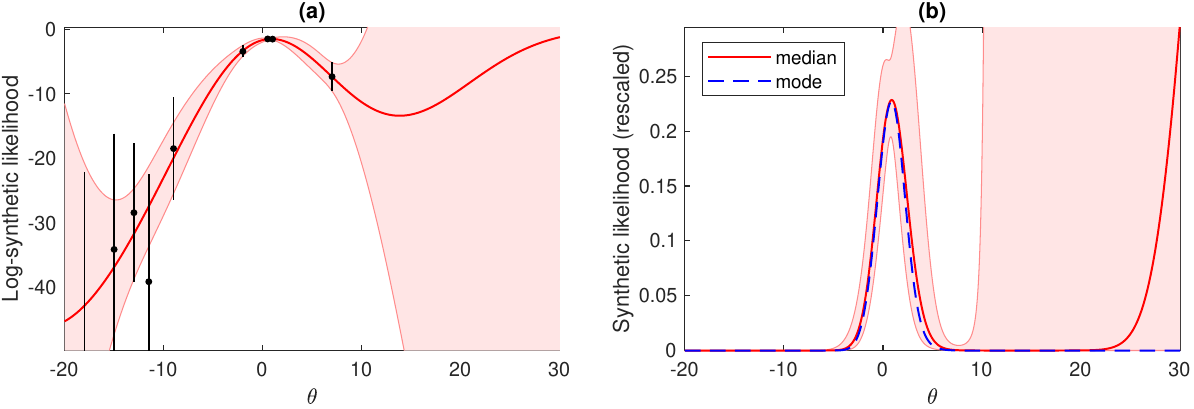}
\caption{(a) Log-SL of a simple 1D toy problem is modelled using a GP with zero-mean prior. Red line (red shaded region) shows the GP posterior mean ($95\%$ credible interval), black dots/lines the log-SL evaluations with corresponding observation errors. (b) The (rescaled) SL, which equals the (unnormalised) SL posterior when a uniform prior on $[-20,30]$ is used, then follows a log-GP. Red line shows the marginal median estimate (\ref{eq:medestim}) and dashed blue line the marginal mode (\ref{eq:mode_post_estim}). The lack of evaluations in $\theta\in[10,30]$ causes large uncertainty in this region as shown by the shaded red region.} \label{fig:loglik_gp1}
\end{figure}

% better loss function
In the practical application of standard MCMC methods all unvisited tail regions are implicitly and inevitably neglected as unlikely as only the visited locations may belong to the final set of samples. \mhblfi{} however requires different, more explicit solution. 
Ideally, the above practical issue is alleviated in \mhblfi{} by incorporating more prior information regarding the possibility of multimodality to the GP model or by meaningfully altering the original prior density $\pi(\Btheta)$ but this is cumbersome. Here we instead propose an estimator for the unnormalised posterior that is shrunk towards zero in regions with large uncertainty.

Consider the loss function
\begin{equation}
    \tilde{l}_g(\tildepiapprox_f,\tpi) 
    \eqdef 
    \int_{\Theta} (1-g(\Btheta))(\tildepiapprox_f(\Btheta)-\tpi(\Btheta))\indic_{\tildepiapprox_f(\Btheta)\geq\tpi(\Btheta)} 
    + g(\Btheta)(\tpi(\Btheta)-\tildepiapprox_f(\Btheta))\indic_{\tpi(\Btheta)>\tildepiapprox_f(\Btheta)} \ud \Btheta,
    \label{eq:better_loss}
\end{equation}
where $g:\Theta\rightarrow(0,1)$ is a weight function. Clearly, the choice $g(\Btheta)=1/2$ gives the $L^1$ loss. By changing the order of expectation and integration and then using Proposition 2.5.5 in \citet{Robert2007}, we can see that the posterior expected loss $\mean_{f\cond\ddata_t}\tilde{l}_g(\tildepiapprox_f,\tpi)$ is minimised when $\tpi(\Btheta)$ is the $(1-g(\Btheta))$-percentile of the log-Normal distribution of $\tildepiapprox_f(\Btheta)\cond\ddata_t$ for (almost) all $\Btheta\in\Theta$. The resulting estimator is hence $\pi(\Btheta)\exp({m_t(\Btheta)+\Phi^{-1}(1-g(\Btheta))s_t(\Btheta)})$.
If we exceptionally allow $g$ to depend on the posterior of $f$ and choose $g(\Btheta) = \Phi(s_t(\Btheta))$ so that the loss function (\ref{eq:better_loss}) penalises large posterior estimates in the regions with large uncertainty, we obtain the estimator
\begin{equation}
     \hat{\tpi}_g(\Btheta) = \mode_{f\cond\ddata_t}(\tildepiapprox_f(\Btheta)) = \pi(\Btheta)\e^{m_t(\Btheta) - s_t^2(\Btheta)}, \label{eq:mode_post_estim}
\end{equation}
which is the marginal mode. %\footnote{This estimator results also when $l_g'(\tildepiapprox_f,\tilde{d})\eqdef \int_{\{\Btheta\in\Theta:|\tildepiapprox_f(\Btheta)-\tilde{d}(\Btheta)|\geq\epsilon\}}\ud\Btheta$ and $\epsilon>0$ is set arbitrarily small.}. 
Estimator (\ref{eq:mode_post_estim}) behaves similarly as (\ref{eq:medestim}) in the highest density region where typically $s_t^2(\Btheta)\approx0$ or $s_t^2(\Btheta) \ll |m_t(\Btheta)|$ but meaningfully shrinks its value towards $0$ in regions with large uncertainty as seen in Figure \ref{fig:loglik_gp1}b (the dashed blue line)\footnote{In \appe{} \ref{appsec:add_viz} we show that a more suitable GP model or an additional evaluation near the right boundary also remove the problematic, second mode in this particular 1D case. %However, trusting the appropriateness of the GP model or the sufficiency of the observations can make the algorithm fragile especially in higher dimensions. A special estimator such as (\ref{eq:mode_post_estim}) is thus beneficial.
}. 
% discuss other choices and warn that estimate may be poor (part of the posterior mass is missed)
%In the experiments in Section \ref{sec:exp} we use (\ref{eq:mode_post_estim}) instead of (\ref{eq:medestim}). % in Algorithm \ref{alg:mhblfi}. 

\subsection{``Bayesian'' \mh{} sampler} \label{subsec:probnum}

% M-H with fixed randomness
We take here a more conceptual approach than in other sections and again revisit the \mh{} sampler in Algorithm \ref{alg:mh}. 
Samples from $q(\Btheta'\cond\Btheta)$ are often obtained using a relation $\Btheta'=g(\Btheta,\Br)$, where $g:\Theta\times\reals^r\rightarrow\Theta$ is a known function and $\Br$ follows some standard distribution. 
For example, the independent \mh{} sampler is obtained using $g(\Btheta,\Br) = \Br$ and the Gaussian proposal $q(\Btheta'\cond\Btheta)=\Normal_p(\Btheta'\cond\Btheta,\BSigma)$ follows as $g(\Btheta,\Br)=\Btheta+\Br, \Br\sim\Normal_p(\Bzeros,\BSigma)$. % (or alternatively as $g(\Btheta,\Br)=\Btheta+\BL\Br, \Br\sim\Normal_p(\Bzeros,\Id)$ where $\BL\in\reals^{p\times p}$ is such that $\BL\BL\T=\BSigma$). 
Although we could proceed more generally, in the following we assume the relation $\Btheta'=g(\Btheta,\Br)=\Btheta + \Br$, where $\Br$ follows some absolute continuous density (e.g.~Gaussian with a non-singular covariance matrix $\BSigma$). 
We disregard here any intricate issues related to the convergence and initialisation of \mh{}. That is, we consider an ideal scenario where the initial point $\Btheta^{(0)}$ is located in the highest density region, a suitable proposal $q$ is immediately available and a single chain (length $n$, no burn-in) long enough to produce a negligible Monte Carlo error is run. 

% main idea
A preliminary key observation is that instead of drawing $u^{(i)}$ and $\Btheta'^{(i)}$ at each iteration $i$ of Algorithm \ref{alg:mh}, $u^{(i)}$ and $\Br_i$ for $i=1,\ldots,n$ can be pre-generated. From now on we hence suppose $u^{(i)}$'s and $\Br_i$'s are fixed so we can exceptionally treat \mh{} as a deterministic algorithm (or mapping) whose input is the likelihood function and output a (multi)set of $n$ samples. The GP posterior $f\cond\ddata_t\sim\GP(m_t(\Btheta),c_t(\Btheta,\Btheta\pr))$ of the log-likelihood then induces a probability distribution over the output, the $n$ samples. % and not on the exact values of $f$ at each proposed point. 
The corresponding estimate $\hat{\bar{h}}_n$ in (\ref{eq:posterior_mean_mcmc_approx}) similarly follows a probability distribution induced by the GP.
We call this unusual approach ``Bayesian'' \mh{}. 
In this approach the posterior uncertainty of $f$ is accounted probabilistically but the sampling error due to the finite sample size $n$ is not. This makes Bayesian \mh{} fundamentally different from the related Bayesian quadrature methods \citep{OHagan1991,Rasmussen2003bmc,Briol2019}. % where neglecting sampling error is (typically) avoided. 

The proposal is either accepted or the current point is kept at each iteration of \mh{}. %The log-likelihood function $f$ is not known exactly at these points but the probability of either case given the GP posterior can be computed. 
Hence, the possible states of \mh{} sampler at iteration $i$ are $\mathcal{S}_i \eqdef \{\Btheta^{(0)} + \sum_{j=1}^i e_j \Br_j \cond e_j\in\{0,1\}\forall j\in\{1,\ldots,i\}\}$ so that $S_0 = \{\Btheta^{(0)}\}$, $S_1 = \{\Btheta^{(0)}, \Btheta^{(0)} + \Br_1\}$, $S_2 = \{\Btheta^{(0)}, \Btheta^{(0)} + \Br_1, \Btheta^{(0)} + \Br_2, \Btheta^{(0)} + \Br_1 + \Br_2\}$ and so on. We also have $\mathcal{S}_i\subset\mathcal{S}_{i+1}, i\geq 0$. From now on we assume $|\mathcal{S}_i|=2^i$ for all $0\leq i\leq n$ which would happen with probability $1$ anyway as we already assumed that $\Br_i$'s are realisations of an absolute continuous density. 
We denote the (random) state of Bayesian \mh{} at iteration $i$ as $\Btheta_{(i)}\in\mathcal{S}_i$. 
The Bayesian \mh{} forms itself a discrete-time process with finite state space $\mathcal{S}_i$ whose cardinality grows exponentially as a function of iteration $i$. 
The process is not Markovian in general so that the probability mass function of each realisation only satisfies $p(\Btheta_{(0)},\ldots,\Btheta_{(n)}\cond\ddata_t) = \prod_{i=1}^n p(\Btheta_{(i)}\cond\Btheta_{(i-1)},\ldots,\Btheta_{(0)},\ddata_t)p(\Btheta_{(0)})$, where the initial probability mass function is $p(\Btheta_{(0)})=\smash{\indic_{\Btheta_{(0)}=\Btheta^{(0)}}}$. 
For example, the posterior probability that the true \mh{} chain would first stay at the initial point $\Btheta^{(0)}$ and then move to the proposed point would be $p(\Btheta^{(0)},\Btheta^{(0)},\Btheta^{(0)}+\Br_2\cond\ddata_t) = \prob_{f\cond\ddata_t}(\gamma_f(\Btheta^{(0)},\Btheta^{(0)}+\Br_1)<u^{(1)},\gamma_f(\Btheta^{(0)},\Btheta^{(0)}+\Br_2)\geq u^{(2)})$ which could be computed using the bivariate Gaussian CDF. 
There are exactly $2^n$ paths the true chain can take %Some of these can coincide but this happens with probability $0$. 
and the probability of other paths is hence $0$. For example, $p(\Btheta^{(0)},\Btheta^{(0)},\Btheta^{(0)}+\Br_1)=0$ since a transition from $\Btheta^{(0)}$ to $\Btheta^{(0)}+\Br_1$ can only happen at iteration $i=1$. 

The expectation of $\hat{\bar{h}}_{n+1}=\sum_{i=0}^n h(\Btheta_{(i)})/(n+1)$ with respect to $f\cond\ddata_t$ is computed as 
\begin{align*}
    \mean_{f\cond\ddata_t}(\hat{\bar{h}}_{n+1}) &= \sum_{(\Btheta_{(0)},\ldots,\Btheta_{(n)})\in\prod_{i=0}^n\mathcal{S}_i} \frac{1}{n+1}\sum_{i=0}^n h(\Btheta_{(i)})p(\Btheta_{(0)},\ldots,\Btheta_{(n)}\cond\ddata_t) \\
    &=
    %=
    \frac{1}{n+1}\sum_{i=0}^n \sum_{\Btheta_{(i)}\in\mathcal{S}_i}h(\Btheta_{(i)})p(\Btheta_{(i)}\cond\ddata_t).
\end{align*}
A formula for the variance of $\hat{\bar{h}}_{n+1}$ can also be derived. 
% computational challenges
Unfortunately, these computations seem to require repeated evaluations of multivariate Gaussian CDF and would not scale better than $\bigO(2^n)$ in any case. Even if some realisations could be neglected as impossible (e.g.~chains that lead outside the support of the prior density) or extremely unlikely based on the GP, the computations remain intractable. Generating a sample path of $(\Btheta_{(0)},\ldots,\Btheta_{(n)})$ is more feasible as this requires only an iterative generation of a GP sample path. The resulting $\bigO(n^3)$ cost still limits this approach to short chains only. % and an additional level of Monte Carlo error.
GP approximations \citep[see e.g.][]{Wilson2020} could be used to bypass the cubic cost but we do not consider them here. % (see however \citet[p.~5]{Jarvenpaa2020_babc} for some discussion and \citet{Wilson2020} for a recent approximate method for efficiently sampling from a GP).

We may now view Algorithm \ref{alg:gpmh}, as well as the related GPS-ABC algorithm by \citet{Meeds2014}, as a heuristic yet tractable approach for constructing an estimator for the true \mh{} chain. In these algorithms the most probable decision, either acceptance or rejection of the proposed point, is selected (un)conditionally on $u^{(i)}$ at each iteration $i$ without acknowledging possible future accept/reject decisions or the locations visited during iterations $0,1,\ldots,i-2$. 
Perhaps a more natural estimator would be the most probable set of samples but it also appears intractable. 
The uncertainty of the resulting set of samples (or of the estimate $\hat{\bar{h}}_n$) due to the GP posterior is not explicitly quantified. % which can be a reasonable practical compromise given that we may not have computational resources to explore all boundary regions and modelling it accurately is difficult

% evaluation locations
An ideal experimental design strategy for Bayesian \mh{} would be such that minimises the expected uncertainty of the samples or of $\hat{\bar{h}}_n$. %Such an ideal approach would also be intractable. 
Such a strategy is also intractable even under the one-step ahead simplification. Furthermore, this strategy would presumably feature similar practical challenges as mentioned in Section \ref{sec:intro}, e.g.~the need for accurate global modelling. 
The tractable sequential strategies developed in Section \ref{sec:acq} instead minimise the expected uncertainty regarding the current accept/reject decision without explicitly accounting for the aforementioned final goal of the inference. % which however facilitates tractable computations. 
%A more natural goal of selecting the evaluations to minimise the expected variance of $\hat{\bar{h}}_n$ under a given budget of log-likelihood evaluations, would be intractable. 

% discuss the term Bayesian MCMC used in one prob.numerics paper in Statistical science?

\section{Summary of theoretical analysis} \label{sec:theory}

% general + acq locations
We briefly outline two most important theoretical results regarding our \gpmh{} algorithm. % presented in detail in \appe{} \ref{appsec:extra_analysis}. 
The first key result is that the optimal evaluation location $\Btheta_{\textnormal{opt}}$ of (\ref{eq:xiopt}) does not generally coincide with $\Btheta$ or $\Btheta'$ unlike one might first intuitively expect, see \appe{} \ref{subsec:eval_loc_analysis} for our detailed analysis on this. This motivates the experimental design strategies of Section \ref{sec:acq} and suggests that restricting the optimisation in (\ref{eq:xiopt}) to $\tilde{\Theta} = \{\Btheta,\Btheta'\}$ may not produce the best possible sample-efficiency. We investigate this empirically in Section \ref{sec:exp}. 

% repeated evaluations and other results
In Section \ref{subsec:nrevals} we analyse the number of log-likelihood evaluations needed at an individual iteration of \gpmh{}. The results there demonstrate a potential shortcoming of \gpmh{}: Some individual step of the algorithm may need even hundreds of evaluations. However, these results concern the worst case situation and are not representative of practice where most accept/reject decisions are done based solely on the GP model and already during the early iterations. The \gpmh{} algorithm may still occasionally get ``stuck'' (somewhat similarly as pseudo-marginal \mh{} but for a completely different reason) especially if $\epsi$ is chosen ``too small''. % which may preclude computation savings. 
The numerical experiments in Section \ref{sec:exp} show that sample-efficiency comparable to \BLFI{} is reached with suitable choices of $\epsi$. 
%Also, the unconditional error in this sense behaves better than the conditional error and is hence adopted to our numerical comparisons in Section \ref{sec:exp}. 

%\subsection{Alternatives and extensions}

% why not target variance of \gamma
%However, computing the value of $\gamma_f(\Btheta,\Btheta')$ very accurately would be wasteful if it can be said with high confidence that $\gamma_f(\Btheta,\Btheta') > 1$ (so that $\Btheta'$ is to be accepted). 
%By not targeting $\gamma_f(\Btheta,\Btheta')$ potential difficulties caused by the heavy tail of its log-Normal distribution are also avoided.  
%
% why not target \alpha
%However the formula (\ref{eq:var_alpha}) in \appe{} \ref{appesec:proofs} is complicated.

\section{Numerical experiments}\label{sec:exp}

% general
In this section we investigate the effect of the tolerance parameter $\epsi$ and the developed sequential experimental design strategies on the quality of the resulting posterior approximation. We compare our \gpmh{} and \mhblfi{} implementations and also consider a \BLFI{} implementation with the theoretically well-motivated and best-performing integrated median interquantile range (IMIQR) strategy by \citet{Jarvenpaa2019_sl}. % whenever feasible 
% on test problems
We consider three scenarios: 1) synthetically constructed log-densities corrupted with additive Gaussian noise (Section \ref{subsec:toymodels}), 2) SL inference for simulator-based models (Section \ref{subsec:sim_models}), 3) likelihood-free generalised Bayesian inference (Section \ref{subsec:bacterial_model}). This allows us to understand how the algorithms perform both in ideal circumstances and in more realistic situations where the GP modelling assumptions are to some extent violated. 

% which methods we compare
Our \gpmh{} and \mhblfi{} implementations both use the same GP surrogate and mainly differ in how the final posterior approximation is formed, as discussed in Section \ref{subsec:gpmcmcasblfi}. 
Instead of using \mhblfi{} with pre-determined number of evaluations as in Algorithm \ref{alg:mhblfi}, it is implemented here so that a separate \mh{} sampler targeting the robust mode-based estimate (\ref{eq:mode_post_estim}) is run at various stages during Algorithm \ref{alg:gpmh}. This approach simplified our experiments and reduced their total computational cost. 
% epoe etc.
We use the unconditional error (\ref{eq:unconderr_gp}) and the two strategies of Section \ref{sec:acq} for selecting the evaluations: 1) ``\epoe{}'' which stands for expected probability of error and requires solving (\ref{eq:xiopt}) over the set (\ref{eq:opt_example_region}) with $c=3/4$, and 2) ``\epoer{}'' where the extra ``r'' informs that the optimisation in (\ref{eq:xiopt}) is restricted to $\tilde{\Theta} = \{\Btheta,\Btheta'\}$. We also consider a baseline ``\naive{}'', where the new evaluation location $\Btheta^*$ is the current point $\Btheta$ with probability $0.5$ and the proposed point $\Btheta'$ otherwise. 

% MCMC settings and initial point
In each experiment, the initial point $\Btheta^{(0)}$ is chosen to be near (but outside) the highest density region to represent the scenario where some weak information about the location is available, e.g.~as a result of pilot runs. %Obviously, starting near the MAP estimate, for example, might lead to better results overall. 
%
% AM-MCMC:
We use a Gaussian proposal whose covariance matrix is updated adaptively as mentioned in Section \ref{subsec:thealgorithm}. 
%The initial proposal covariance matrix is diagonal whose entries are chosen to roughly represent the expected variability. 
The first quarter of the (approximate) MCMC samples is always neglected as burn-in. 
%
% GP surrogate model
We use the same GP model of Section \ref{subsec:gp_model} for all of our experiments. In particular, we use basis functions $1, \theta_j, \theta_j^2$ for each dimension $j$ and we set $\Bb=\Bzeros$ and $B_{jk}=30^2\indic_{j=k}$ although this is likely unideal. We assume that the log-likelihood is a smooth function on the highest density region and use the squared exponential covariance function $k_{\Bphi}(\Btheta,\Btheta\pr) = \sigma_{s}^2 \exp(-\sum_{j=1}^p(\theta_j-\theta_j\pr)^2/(2l_j^2))$. % although other choices are possible. 
%
% GP hyperparam optim
We further use relatively uninformative hyperpriors for the GP hyperparameters $\Bphi=(\sigmaf^2,l_1,\ldots,l_p)$ (and $\sigma_n^2$ in Section \ref{subsec:toymodels} and \ref{subsec:bacterial_model}). The hyperparameters are re-estimated immediately after each new log-likelihood evaluation when $t\leq 300$ and after every $10$th evaluation otherwise.

% TV
We summarise the posterior approximation accuracy mainly via the average total variation distance over the coordinate-wise marginal densities as compared to the ground-truth\footnote{Mixing of (approximate) MCMC chains is not a major concern in our low-dimensional set-up and we consequently only analyse the eventual posterior approximation accuracy as a function of the limited number of log-density evaluations used. However, poor mixing or lack of convergence of (approximate) MCMC chains can definitely affect the results in other settings and should in practice be assessed using standard methods, see e.g.~\citet[Chapter~12]{Robert2004}, though their validity is questionable whenever the target density is slightly changing. Also, we do not use criteria such as the effective sample size per computational cost as a summary of performance because our algorithms are approximate and their computational cost varies as a function of iteration.}. In the 2D cell biology experiment we instead use the joint total variation. That is, in the former case we define 
$\TV(\pi,\pi\pr) \eqdef \sum_{j=1}^p\int_{\Theta_j}|\pi(\theta_j)-\pi\pr(\theta_j)|\ud\theta_j/(2p)$  
and in the latter case 
$\TV(\pi,\pi\pr) \eqdef \int_{\Theta}|\pi(\Btheta)-\pi\pr(\Btheta)|\ud\Btheta/2$
where $\pi, \pi\pr$ are pdfs both defined over $\Theta=\prod_{j=1}^p\Theta_j\subset\reals^p$ and $\pi(\theta_j), \pi\pr(\theta_j)$ denote their marginals. % as also used in previous work. 
We chose these criteria to be in line with previous work and because their values are easily approximated using the MCMC output, are not overly sensitive to inevitable small errors especially in the tails and are easy to interpret (as the overlap between the densities in the scale $[0,1]$). %We have also separately analysed the approximation of the correlation structure. %As the latter distance depends only on the marginals we investigate the accuracy of the correlation structure of the posterior approximation separately.  
Each run of the algorithms is repeated $100$ times ($50$ times for cell biology experiment) with different realisation of randomness to assess the variability. 
%
% software/computations
The experiments were performed using MATLAB 2022a. Some GP functionality was taken from GPstuff 4.7 \citep{Vanhatalo2013}.

\subsection{Synthetic log-densities} \label{subsec:toymodels}

% explain the synthetic log-densities
We first consider three 6D densities from \citet{Jarvenpaa2019_sl} with different characteristics: a Gaussian density called `Simple', a banana-shaped density `Banana' and a multimodal density `Multimodal'.
% noise levels
The variance of the log-density evaluations $\sigma_n^2(\Btheta)$ is here constant and estimated together with the GP hyperparameters $\Bphi$. We use $\sigma_n=1$ for Banana and Multimodal and $\sigma_n=2$ for the Simple log-density. % to make this test case more challenging. 
Further details of the target densities are given in \appe{} \ref{appsec:toymodel_details} and some additional results in \appe{} \ref{appsec:synth_add_res}. 
The number of initial evaluations is $t_{\text{init}}=10$. Algorithm \ref{alg:gpmh} is run for $i_{\text{MH}}=10^5$ iterations. 

% results: x-axis number of evals
Figure \ref{fig:res6d_eval} show the median accuracy and the variability of the final posterior approximation. A hypothetical best possible algorithm that always produces an exact posterior ($\TV=0$) without any log-density evaluations would appear in the left lower corner of the figure. We see that \epoe{} produces the best sample-efficiency, the \naive{} method is the worst and \epoer{} is roughly halfway between them. %All three methods significantly improve upon pseudo-marginal MCMC which would require at least $10^4$ evaluations to even allow one to check the convergence and produce reasonably small sampling error. 
Decreasing $\epsi$ leads to more accurate posterior approximations as expected. %At most a few hundred log-density evaluations are needed with suitable $\epsi$. 
Interestingly, a fairly large tolerance ($\epsi\approx 0.3$) is already sufficient to produce reasonable approximations. In this case \epoe{} needs at most a few hundred log-density evaluations. %Good posterior approximations with small enough number of evaluations is achieved with suitable choices of $\epsi$ in all three cases. %The Banana density requires more evaluations than the rest. 
Banana log-density is more challenging than the other two as all methods face some challenges in estimating its long tails. %This is unsurprising since even a well tuned random walk \mh{}, that has access to exact log-density evaluations, would need a large number of samples to sufficiently visit its tails. %Here the noisy evaluations further complicates the estimation. %The results are nevertheless good given that only several hundreds of evaluations are used as seen in Figure \ref{fig:res6d_eval}. 
\gpmh{} and \mhblfi{} implementations produce similar results which is because most log-density evaluations are already collected during the ``burn-in'', as seen in Figure \ref{fig:res6d_cdfevals} of \appe{}. 
The results also show that \epoe{} produces similar or slightly worse accuracy as \BLFI{} (IMIQR) with large values of $\epsi$ and performs slightly better with low $\epsi$. 

\begin{figure}[hbt]
\centering
\includegraphics[width=0.97\textwidth]{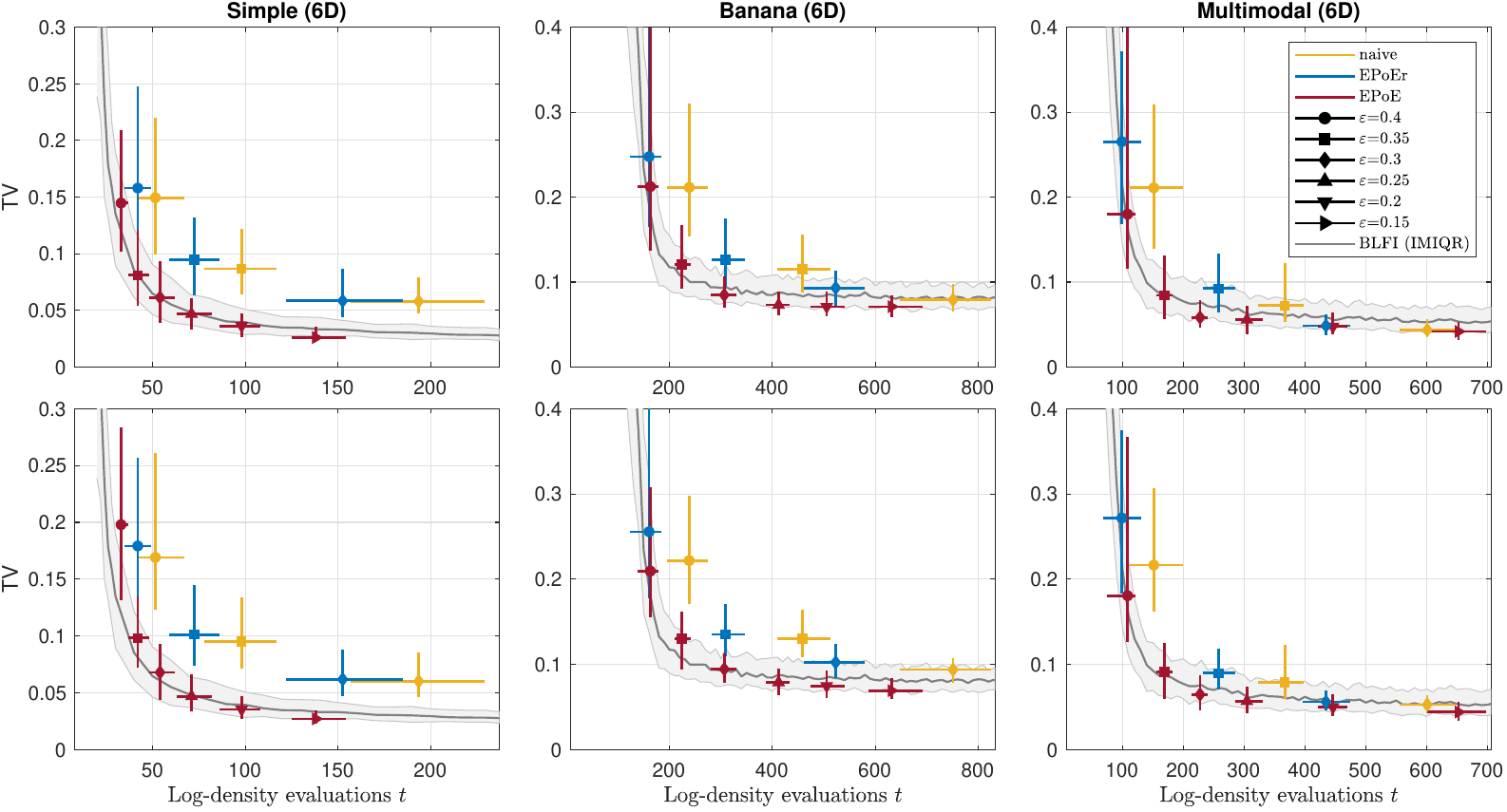}
\caption{Accuracy of the marginal posterior approximation as a function of the used log-density evaluations at the final iteration $i_{\text{MH}}=10^5$. The horizontal and vertical lines show the middle $75\%$ computed over $100$ repeated runs and the marker in the middle shows the corresponding (marginal) median. The grey lines/area show the median and the middle $75\%$ for \BLFI{} with IMIQR design strategy over $100$ repeats. \emph{Top row} shows \gpmh{} and the \emph{bottom row} the corresponding results by \mhblfi{}.} \label{fig:res6d_eval}
\end{figure}

% evaluation locations
Typical examples of selected evaluation locations are shown in Figure \ref{fig:simple_eval_example}. \epoe{} tends to select slightly more diverse evaluation locations as the two other methods. On the other hand, the evaluation locations of IMIQR are even more diverse and, despite the fairly large number of initial locations ($20$ for Simple and $50$ for Banana), some of its evaluations occur at the boundary of the parameter space. 
%(We do not investigate this further but only note that the evaluations of \gpmh{} could be made more diverse by setting a fixed proposal density that takes long jumps in the parameter space although this may not be advantageous.) 
When the initial point $\Btheta^{(0)}$ is far from the highest density region, evaluations intuitively occur on a path connecting the initial point and the highest density region as most clearly seen in the top left case of Figure \ref{fig:simple_eval_example}. 

\begin{figure}[hbt] % additional example of evaluation locations - Simple6D & Banana6D
\centering
\includegraphics[width=0.9\textwidth]{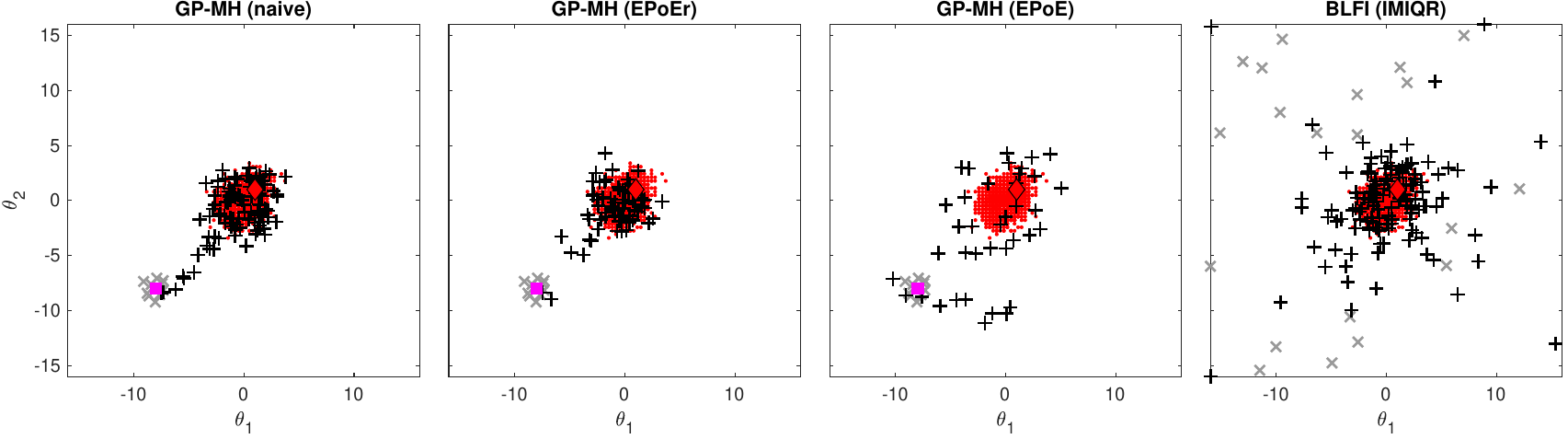}
\includegraphics[width=0.9\textwidth]{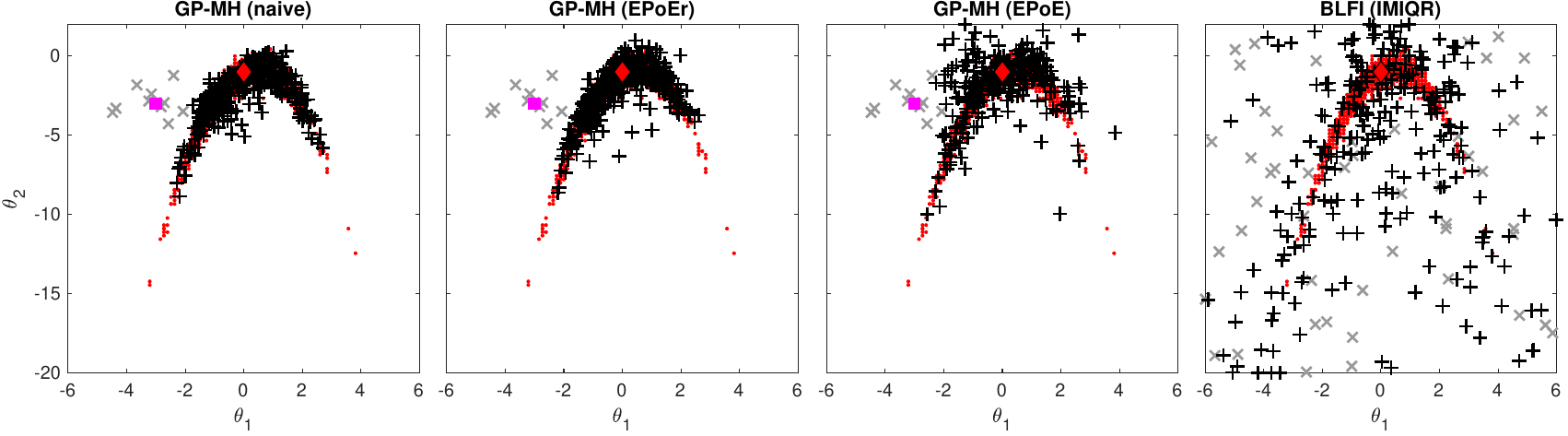}
\caption{Typical realisations of the log-density evaluation locations in the Simple (\emph{top row}) and Banana (\emph{bottom row}) experiments. The evaluation locations are projected to the first two components. Red dots in the background illustrate the true target density, the grey crosses ($\times$) the $t_{\text{init}}$ initial evaluation locations and the black plus signs ($+$) the evaluation locations. The magenta square shows the initial location.} \label{fig:simple_eval_example}
\end{figure}

\subsection{SL inference}\label{subsec:sim_models}

%Having confirmed the intuitive behaviour and the promising trade-off between accuracy and sample-efficiency of Algorithm \ref{alg:gpmh} with toy models, 
We next consider simulator-based models where SL is used as the target likelihood, see Section \ref{subsec:lfi} of \appe{} for details on SL. % from $N$ repeated simulations at each selected evaluation location. 
Although our methodology is particularly useful for expensive models, 
we here nevertheless use models that are not highly costly (although not very cheap either). 
This allows us to compare our results directly to reasonable ground-truth posteriors obtained via extensive computation using SL-MCMC \citep{Price2018}. 
For simplicity, we always evaluate log-SL using fixed number of repeated simulations $N$. 
%Although it might be possible to adjust $N$ adaptively, we use fixed $N$ to evaluate log-SL. %This approach would allow straightforward parallellisation. %has also the advantage of straightforward parallellisation for truly expensive simulation model. 

% practical difficulties in handling \sigma_n:
We use the bootstrap to estimate the noise variance $\sigma_n(\Btheta)$ at the evaluated locations in $\ddata_t$ but this does not work for \epoe{} which needs estimates of $\sigma_n(\Btheta)$ at each $\Btheta$. 
We first approximated $\sigma_n(\Btheta)$ for \epoe{} with a constant obtained near the MAP parameter value but this resulted underestimated $\sigma_n(\Btheta^*)$ and hence overestimated $\xi_t^2(\Btheta,\Btheta';\Btheta^*)$ in the tails which further caused overexploration of the tails. As a heuristic solution we set $\sigma_n(\Btheta)=0.1$ as if the evaluations were almost noiseless. %We also examined other choices but they did not bring improvements. 
More realistic estimates could possibly be obtained by modelling $\sigma_n(\Btheta)$ as a function of $\Btheta$ (nearest neighbour interpolation was used by \citet{Acerbi2020} for a similar goal but this approach would make $\sigma_n(\Btheta)$---and consequently \epoe{}---discontinuous) but our simple approach already produced improvements over the naive approach. \epoer{} was observed to be insensitive for this choice. We leave more detailed analysis on modelling the noise variance, and the log-likelihood function itself, an important topic for future work.

\subsubsection{Theta-Ricker model} \label{subsubsec:thetaricker}
% Note to me: Write Theta-Ricker when start a sentence, theta-Ricker otherwise

We consider theta-Ricker model, see e.g.~\citet{Polansky2009} and references therein for background.
In this model $N_t$ denotes the number of individuals in a population (or population density) at time $t$ which evolves according to the discrete time process
%$
\begin{equation*}
N_{t+1} = r N_{t} \exp(-\log(r)(N_{t}/K)^{\theta} + \epsi_t), 
\end{equation*}
%$
for $t=1,\ldots,T$. Parameter $K$ indicates the population size when the growth rate $r$ goes to zero and $\theta$ %(which should not be confused with $\Btheta$ that we use as a generic notation for any parameter vector) 
further controls the form of the growth rate. 
The Ricker model, a common LFI test problem that we also consider in \appe{} \ref{appsec:ricker_thetaricker}, is a special case with $\theta=1$ and $K=\log(r)$. Gaussian process noise model $\epsi_t \sim \Normal(0,\sigma^2_{\epsi})$ is assumed. 
A noisy measurement $x_t$ of the population size $N_t$ is further assumed to follow Poisson observation model
$
%\begin{equation}
x_t \cond N_{t}, \phi \sim \Poi(\phi N_{t}),
%\end{equation}
$
where $\phi$ is a scale parameter. 

% theta-Ricker - experimental details
The goal is to compute the SL posterior for the parameters\footnote{We use similar parametrisation for our (theta-)Ricker models as in the (scaled) Ricker model of \citet{Wood2010} and several of its follow-up articles for consistency. This is however different from the common definition for (theta-)Ricker model given e.g.~in \citet{Polansky2009} where $r$ is used in the place of our $\log(r)$.} $\Btheta=(\log(r),\theta,K,\phi,\sigma_{\epsi})$ given data $\Bx = (x_t)_{t=1}^T$ with $T=100$. 
The initial population size is $N_{0}=1$ and we use independent priors $\log(r)\sim\Unif([2,5]), \theta\sim\Unif([0.01,2]), K\sim\Unif([1,5]), \phi\sim\Unif([4,20]), \sigma_{\epsi}\sim\Unif([0,0.8])$. 
We use $N=100$ repeated simulations and the $13$ summary statistics proposed by \citet{Wood2010} to compute log-SL. We acknowledge that these statistics may be unideal as they are designed for the simpler Ricker model but this way we obtain a challenging target density which could be encountered during a typical LFI workflow. %Namely, in practice sensible summaries may not be initially available. % and their suitability may not evident. 
The ``true'' parameter for generating the data is $\Btheta_{\text{true}}=(3.5, 1.0, 3.5, 10, 0.3)$. 
We use $t_{\text{init}}=20$ and $i_{\text{MH}} = 2\cdot 10^5$. The initial location and initial proposal covariance are $\Btheta^{(0)}=(3.4,   0.9,  3.0,   8.0,  0.3)$ and $\BSigma_0=\diag(0.05,   0.1, 0.25,   0.5, 0.05)^2$, respectively. We observe noise level of $\sigma_n\gtrsim 1.0$ in the highest posterior density region.

% theta-Ricker - results
Figure \ref{fig:res_thetaricker_eval} shows that \epoe{} produces the best sample-efficiency (especially when $\epsi=0.35$ or $0.3$) though the improvements are not as substantial as in the more ideal GP modelling scenario of Section \ref{subsec:toymodels}. %A probable reason for this is the difficulty with modelling the noise variance of log-SL. 
We see that $\epsi=0.4$ does not produce accurate posterior ($\TV\gtrsim 0.3$) but already $\epsi=0.35$ produces reasonable approximations ($\TV\lesssim 0.1$) and requires only $300-500$ log-SL evaluations. %After around $600$ evaluations no improvements in accuracy are seen. 
Naive with $\epsi=0.25$ would have required more evaluations than our limit $10^3$ which we set to keep the computation time bounded.
\gpmh{} and \mhblfi{} again produce similar results.  

\begin{figure}[hbtp]
\centering
\includegraphics[width=0.75\textwidth]{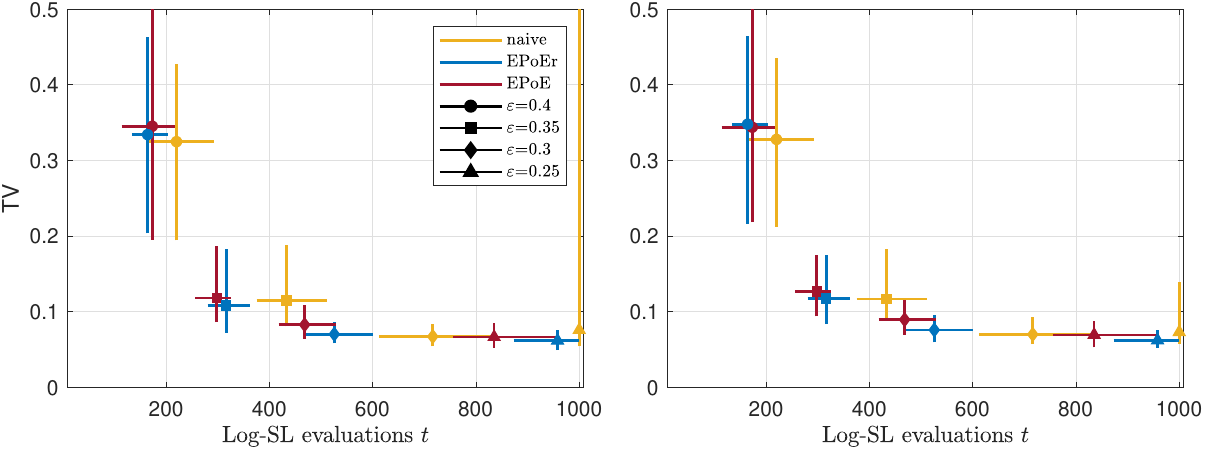}
\caption{Accuracy of the marginal posterior approximation in the theta-Ricker experiment as a function of log-likelihood evaluations at the final iteration $i_{\text{MH}}=2\cdot 10^5$. \emph{Left plot} shows \gpmh{} and the \emph{right plot} the corresponding results by \mhblfi{}.} \label{fig:res_thetaricker_eval}
\end{figure}

% theta-Ricker posterior example 
Figure \ref{fig:thetaricker_post_example} shows a typical estimated posterior with \epoe{} and $\epsi=0.3$. 
% K and \phi:
Parameters $K$ and $\phi$ cannot be fully identified which makes their true joint marginal posterior difficult to approximate also. The overall approximation quality is still good when $\epsi \lesssim 0.35$ and, most importantly, the non-identifiability of these two parameters is clearly captured. 
% typical failure (not shown):
When using $\epsi=0.4$ the approximation for $(\log(r),\theta,\sigma_e)$ was still reasonable but substantial amount of the probability mass of ($K,\phi$) was often missed. % (not shown). 

\begin{figure}[hbt] % example posterior - theta-Ricker
\centering
\includegraphics[width=0.78\textwidth]{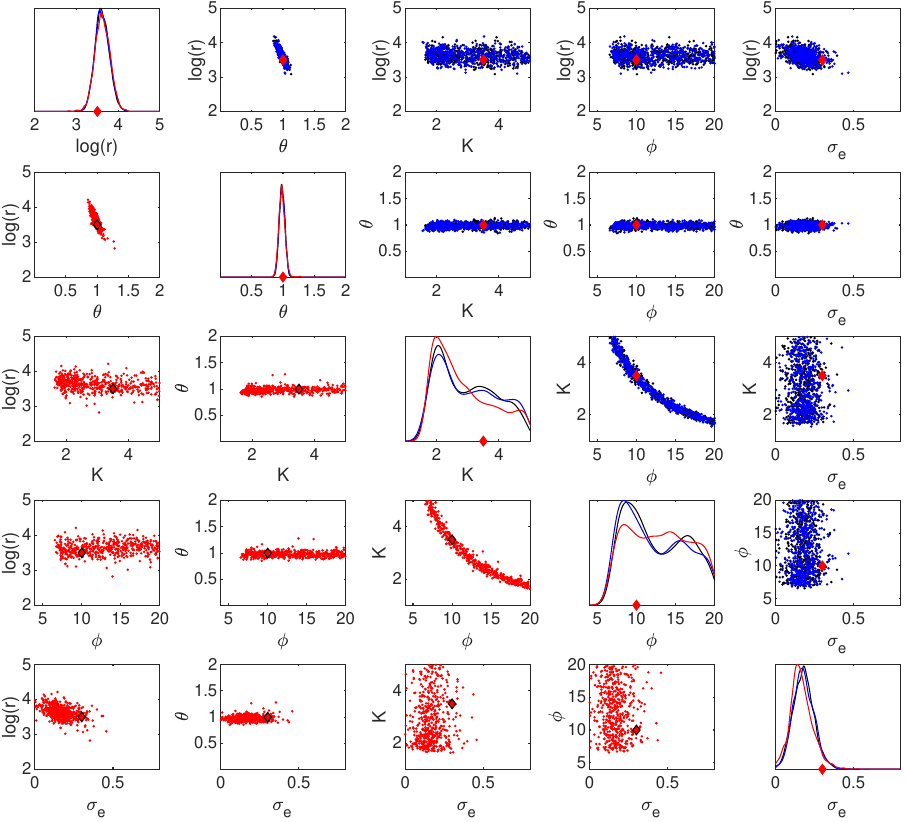}
\caption{Comparison of the ground-truth posterior (red dots/line) computed using SL-MCMC and a typical example of estimated posterior (blue dots/line: \gpmh{} posterior, black dots/line: \mhblfi{} posterior) in the case of theta-Ricker experiment. The red diamond shows the true parameter.} \label{fig:thetaricker_post_example}
\end{figure}

% DISCUSS HOW B(O)LFI DID NOT WORK HERE
We were unable to obtain reasonable posterior approximations in our theta-Ricker experiment using \BLFI{}. The main reason for this was that log-SL behaves irregularly on certain boundary regions of the parameter space where the method typically needs to evaluate. This produces a vicious cycle where a poor global GP fit causes next evaluation locations to be uninformative and the GP fit to remain poor. We unsuccessfully tried thresholding the log-SL values and special initialisations. Cropping the problematic parameter regions would help but is cumbersome due to the parameter correlations and would complicate optimising the acquisition functions. Similar problems emerge with other techniques relying on global GP surrogate. Especially current implementations of BOLFI typically excessively evaluate near the boundaries as seen e.g.~in \citet{Jarvenpaa2018_acq,Picchini2020}. 
However, when initialised near the highest posterior density region, \gpmh{} and \mhblfi{} produced accurate results and gracefully avoided the problematic regions. Namely, if a new parameter is proposed in such a region, a new log-likelihood evaluation is often triggered there. This point is then simply rejected without updating the GP or having to alter the acquisition function (see \appe{} \ref{appsec:gpmodellingdetails} for details). 
%This approach is still not fool-proof because it is inevitably possible that such a proposal is incorrectly accepted solely based on the GP model in which case the algorithm proceeds to the problematic region and may get stuck there. However, in such rare cases the algorithm can be restarted with different initialisation. %It does not seem straightforward to incorporate such systematic and general approach to B(O)LFI. 
In rare cases where the algorithm still enters such problematic regions and gets stuck there, the algorithm can be restarted with different initialisation.

\subsubsection{Cell biology model} \label{subsubsec:cellmodel}

% general about the model
We consider a model used for estimating cell motility and proliferation which are further needed when assessing the efficacy of certain medical treatments. Background and details of the model can be found in \citet{Price2018} and references therein. 
In the model T3T cells are represented as an $R\times C$ binary matrix at each time point and each $(x,y)$ location indicates whether a cell is present there or not. The cell dynamics are simulated over time using a random walk model which features two parameters that control the cell movement and reproduction, the probability of motility $P_m \in [0,1]$ and the probability of proliferation $P_p \in [0, 1]$. 
Observed data can be obtained using a scratch assay and consist of images (binary matrices), measured at some time points. %To create the observed data, the cells are placed on a two-dimensional discrete lattice.
%In this paper we however use simulated data to fit the model. 
The resulting likelihood function is intractable. % and it is not easy to design informative and low dimensional summary statistics. 
We consider similar set-up as in \citet{Price2018}. We use simulated data consisting of binary matrices over $145$ time instances, we set $R = 27, C = 36$ and we initially place  $110$ cells randomly in the rectangle with positions $x \in \{1, 2, \ldots, 13\}, y \in \{1, 2, \ldots, 36\}$. 

% experimental details
Although the model has only two parameters, the imposed small budget for log-SL evaluations ($\leq750$) and the relatively large noise level of the log-likelihood evaluations (we use $N=2500$ which results $\sigma_n\gtrsim 2.5$ in the highest density region) makes inference challenging. One log-SL evaluation takes approximately $7$s \citep[on a PC laptop with Intel Core i5 8265U and using the C-code by][]{Price2018} which makes SL-MCMC feasible but fairly expensive. For example, a single chain with length $10^5$ requires approximately $8$ days of computing. Using more image data, larger lattice, more complicated cell dynamics or less efficient implementation would make the inference even more expensive. 
We use the same $145$ summary statistics as \citet{Price2018}. These include the Hamming distances between all the subsequent binary matrices over the $144$ time intervals and the total number of cells in the final time period. 

% prior, bounding box, initial point etc.
We use $\Btheta_{\text{true}}=(0.35, 1.0\cdot10^{-3})$. % similarly to \citet{Price2018}. 
We first experimented with $\Unif([0,1]^2)$ prior but immediately observed that initialising our method in $\theta_2\gtrsim4.0\cdot10^{-3}$, where log-SL value is negligible and has very large variance, would not work. Similar difficulties would affect also SL-MCMC.
In fact, log-SL decreases very fast near all boundaries which is problematic for B(O)LFI. While it is feasible to construct a bounding box to crop such regions in this particular 2D case, this in general involves tedious manual work. % and is not always feasible due to the complex shape of the posterior as in the theta-Ricker experiment. 
While not absolutely necessary, we restricted the parameter space of $\theta_2$ and coded this into the prior $\Btheta\sim\Unif([0,1]\!\times\![0,4.0\cdot10^{-3}])$. 
We use $t_{\text{init}}=10$, $i_{\text{MH}} = 10^5$, $\Btheta^{(0)} = (0.5, 1.5\cdot10^{-3})$ and $\BSigma_0 = \diag(0.02, 2.0\cdot10^{-4})^2$.

% discussion on the results
Overall the results summarised as Figure \ref{fig:res_cell} are similar to those in the previous experiments. However, the difference between the methods is smaller and the variability in the number of used log-SL evaluations is substantially larger especially when $\epsi=0.2$. This variability is mostly caused by some individual iterations requiring around $30-70$ evaluations. Presumably the fairly large $\sigma_n$ and the low-dimensional $\Btheta$ make the progress of the algorithm more dependent on randomness as in our other experiments. 
% could try to artificially limit the number of evaluations/iteration
% 
We nevertheless obtain reasonable posterior approximations with only a few hundred log-SL evaluations. The computational cost of each run was at most one to two hours. % which is substantially less than using SL-MCMC. 
We observed isolated cases where the algorithm had traversed to the problematic boundary region and got stuck there. As these rare cases all happened when $\epsi\geq 0.3$, it is safe to say that our algorithm worked robustly in this experiment. % despite the problematic boundary regions. 

\begin{figure}[hbtp]
\centering
\includegraphics[width=0.75\textwidth]{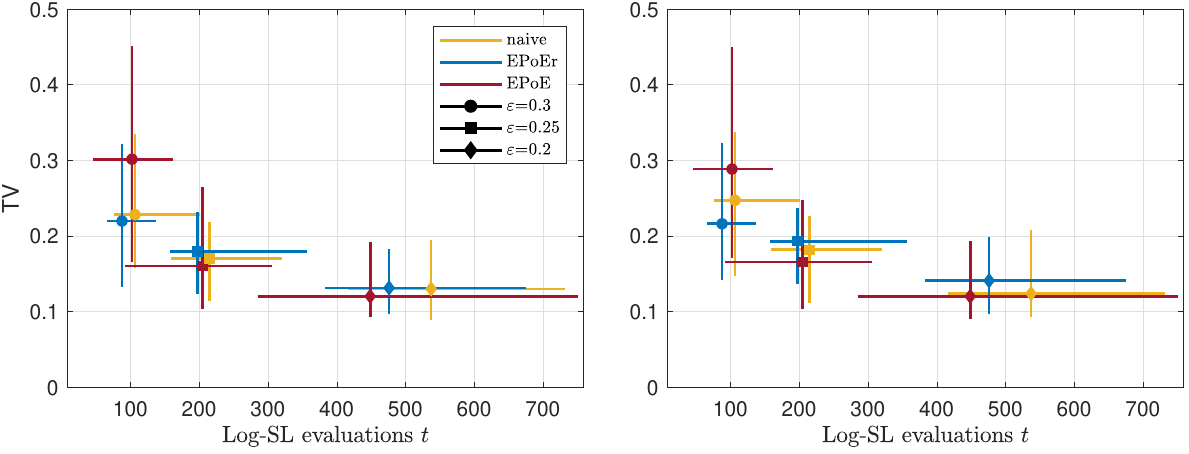}
\caption{Accuracy of the posterior approximation in the Cell biology experiment as a function of log-likelihood evaluations at the final iteration $i_{\text{MH}}=10^5$. \emph{Left plot} shows \gpmh{} and the \emph{right plot} the corresponding results by \mhblfi{}.} \label{fig:res_cell}
\end{figure}

\subsection{Likelihood-free generalised Bayesian inference} \label{subsec:bacterial_model}

We have now covered the case where one has access to a small number of noisy evaluations of the likelihood function or some approximation such as SL. However, \gpmh{} can be extended to other situations, in particular to the ABC case when a single model simulation is computationally costly and a surrogate model is formed for the ABC discrepancy. % $\Delta(\Bxobs,\Bx)$ where $\Bxobs$ is the observed data and $\Bx$ denotes data simulated from the intractable model. 
See Section \ref{subsec:abc} for details on ABC relevant to what follows and \citet{Meeds2014} for another approach that is also mentioned in Section \ref{subsec:literature}. %As mentioned,  have covered the situation where individual summary statistics are modelled with GPs, we here model the discrepancy. 

% using gpmh in the ABC setting
Suppose the discrepancy $\Delta(\Bxobs,\Bx)$ follows Gaussian distribution with mean $d(\Btheta)$ and variance $\sigma^2_n(\Btheta)$ when $\Bx\sim\pi(\Bx\cond\Btheta)$ and where the observed data $\Bxobs$ is again suppressed from the notation. The ABC likelihood (\ref{eq:abc_prob}) of \appe{} \ref{subsec:abc} is then proportional to $\Phi[(h-d(\Btheta))/\sigma_n(\Btheta)]$ where $h>0$ is the threshold. % and $\Phi$ is the CDF of standard Gaussian.
The unknown mean discrepancy $d(\Btheta)$ can be modelled with GP as first suggested by \citet{Gutmann2016}, though the likelihood ratio $\Phi[(h-d(\Btheta'))/\sigma_n(\Btheta')]/\Phi[(h-d(\Btheta))/\sigma_n(\Btheta)]$ in the \mh{} accept/reject test then depends on $d$ in a such way which does not seem to allow analytical computations. %The dependence on $\sigma_n$ is problematic too. 
Although efficient numerical approach may be feasible, in the following we instead consider an alternative target density which directly facilitates \gpmh{}. % which we then investigate numerically.

\subsubsection{Generalised posterior with expected ABC discrepancy as the loss}

% target density and some justifications for it
We consider a target density proportional to
\begin{align}
    %\pi(\Btheta\cond\Bxobs) \propto 
    \pi(\Btheta)\e^{-wl(\Btheta,\Bxobs)}, \label{eq:genpost_bact}
\end{align}
where $w>0$ is a scaling constant and the loss function $l(\Btheta,\Bxobs)$ is selected as 
\begin{align}
    l(\Btheta,\Bxobs) = l_{\Delta}(\Btheta,\Bxobs) 
    \eqdef \mean_{\Bx\cond\Btheta}\Delta(\Bxobs,\Bx). \label{eq:meandiscrloss}
\end{align}
%
%where $\Delta(\Bxobs,\Bx)$ is the discrepancy formed as in ABC. 
We assume $\mean_{\Bx\cond\Btheta}\Delta(\Bxobs,\Bx)$ is finite at each $\Btheta$ and (\ref{eq:genpost_bact}) can be normalised. 
This approach can be considered as an instance of generalised Bayesian inference by \citet{Bissiri2016} though the discrepancy $\Delta(\Bxobs,\Bx)$ in the ABC case typically does not satisfy a coherence property that \citet{Bissiri2016} used to provide a justification for the exponential form of (\ref{eq:genpost_bact}). %Its  justification is however perhaps less clear because we here consider non-i.i.d. LFI setting where the coherence property, saying that we should arrive the same posterior irrespective whether we first update the prior with a subset of data or the, does not hold in the ABC case.  
The density (\ref{eq:genpost_bact}) is also a solution to
\begin{align}
    \arg\min_{\nu}\int wl(\Btheta,\Bxobs)\nu(\ud\Btheta) + \text{KL}(\nu\,||\,\pi). \label{eq:genpost_variational}
\end{align}
For small $w$ the generalised posterior (\ref{eq:genpost_bact}) is similar to the prior $\pi(\Btheta)$ in the sense of Kullback-Leibler divergence ($\text{KL}$).
If $w$ is set arbitrarily large, then a point mass at $\arg\min_{\Btheta}l(\Btheta,\Bxobs)$ is a solution to (\ref{eq:genpost_variational}) and can in principle be computed using (standard) Bayesian optimisation. 

% my justification
The choice $l(\Btheta,\Bxobs)=-\log\pi(\Bxobs\cond\Btheta)$ with $w=1$ gives the exact posterior and the ABC posterior is formally obtained when $w=1$ and  
\begin{align}
l(\Bxobs,\Btheta) 
= l_{\textnormal{ABC}}(\Bxobs,\Btheta) 
\eqdef -\log\mean_{\Delta\cond\Btheta} K_h(\Delta) = -\log\int K_h(\Delta(\Bxobs,\Bx))\pi(\Bx\cond\Btheta)\ud\Bx. \label{eq:abcloss}
\end{align}
Assuming the kernel $K_h$ is strictly positive %(which does not hold for the uniform kernel $K_h(r)\propto\indic_{r\leq h}$ but holds e.g.~for Exponential and Gaussian kernels) 
and by using Jensen's inequality, we have $l_{\textnormal{ABC}}(\Btheta,\Bxobs) \leq -\mean_{\Delta\cond\Btheta}\log K_h(\Delta)$.
When the Exponential kernel $K_h(r)\propto\e^{-r/h}$ is used, it is easy to see that $-\mean_{\Delta\cond\Btheta}\log K_h(\Delta)$ equals $wl_{\Delta}(\Btheta,\Bxobs)$ with $w=1/h$ up to an additive constant. That is, in this case (\ref{eq:genpost_bact}) is also obtained from the ABC posterior by replacing the ``ABC-loss'' $l_{\textnormal{ABC}}(\Btheta,\Bxobs)$ with the larger loss $wl_{\Delta}(\Btheta,\Bxobs)$. 

% why this target?
A computational advantage of (\ref{eq:genpost_bact}) is that the mean discrepancy only needs to be accurately estimated %(which facilitates also non-probabilistic regression methods) 
though meaningful uncertainty estimates are still needed to apply \gpmh{}. The standard ABC approach instead requires the more difficult task of estimating the lower tail probability of the discrepancy. Also, (\ref{eq:genpost_bact}) may be more appropriate when the simulator-based model is supposedly grossly misspecified yet no better model is available. In this case the interpretation of the exact posterior is problematic. A downside is that (\ref{eq:genpost_bact}) does not have the interpretation as the exact posterior when the model is well-specified and careful selection of $w$ is required. % maybe write some more about selecting w?

% estimate of l
\gpmh{} is applied by replacing the target log-likelihood $f(\Btheta)$ with $-w l_{\Delta}(\Btheta,\Bxobs)$. The noisy evaluations are obtained as $-w\sum_{i=1}^M \Delta(\Bxobs,\Bx^{(i)})/M$, where $\Bx^{(1)},\ldots,\Bx^{(M)}$ are samples from $\pi(\Bx\cond\Btheta)$, which clearly defines an unbiased estimator for $-w l_{\Delta}(\Btheta,\Bxobs)$. Although the discrepancy is commonly formed to be non-negative, Gaussian distribution for $\Delta(\Bxobs,\Bx^{(i)})$ is a sensible general assumption \citep{Jarvenpaa2020_babc}. If a relatively large $M$ can be used, the central limit theorem further justifies the Gaussianity assumption. %though this is not the case here where the simulations are assumed expensive. 
% We observe that 
% %
% \begin{align*}
% %\Var_{\Bz^{(1)},\ldots,\Bz^{(M)}\cond\Btheta}(w\sum_{i=1}^M \Delta(\Bx,\Bz^{(i)})/M) = (w^2/M)\Var_{\Bz\cond\Btheta}(\Delta(\Bx,\Bz))
% \Var_{\Bz^{(1)},\ldots,\Bz^{(M)}\cond\Btheta}\left(\frac{w}{M}\sum_{i=1}^M \Delta(\Bx,\Bz^{(i)})\right) 
% = \frac{w^2}{M}\Var_{\Bz\cond\Btheta}\Delta(\Bx,\Bz)
% \end{align*}
% %
% so that increasing $w$ is expected to increase the required number of model simulations. % just like decreasing ABC threshold $h$... 

\subsubsection{Bacterial infections model}

% the bacterial infections model
\citet{Numminen2013} developed a simulator-based model to understand transmission dynamics of bacterial infections in day care centres. It is defined is terms of a latent continuous-time Markov process and an observation model. % The model features intractable likelihood and each simulation requires several seconds making inference with ABC feasible but costly. 
The model has been used previously by \citet{Gutmann2016,Jarvenpaa2020_babc} to demonstrate that GP-based ABC methods can substantially improve the computational efficiency over standard ABC samplers though here we instead briefly investigate the loss-based general Bayesian inference with \gpmh{} as an alternative to ABC. 
% data:
We estimate the three unknown parameters of the model, the internal, external and co-infection parameters $\beta,\Lambda$ and $\theta$, respectively. 
We use the same real data as \citet{Numminen2013} which describes colonisations with the bacterium \emph{Streptococcus pneumoniae} and consists of varying numbers of sampled attendees at $29$ day care centres at a single time point. 
The simulation time depends on the parameter value and takes several seconds for this data. Further details of the model and data can be found in \citet{Numminen2013}. 

% bacterial discrepancy:
The discrepancy is formed similarly as in \citet{Gutmann2016}: The data of each day care centre (a binary matrix whose elements inform whether each sampled attendee is infected with each of the bacterial strains at the observed time point) is summarised with four statistics. The $L^1$ distances between the empirical CDFs over the $29$ day care centres both for the observed and simulated data are then computed for each summary. 
The discrepancy is defined as the square root of the average of the resulting four $L^1$ distances (weighted by the magnitude of the observations). 
%As in the SL case, we consider only fixed $M$. % and some values $w$. 
%Analysis of suitable selection of $w$ or the loss function itself is left for future work. % we only investigate some possible choices numerically.
Also as in previous work, we use uniform priors $\beta\sim\Unif([0,11]),\Lambda\sim\Unif([0,2])$ and $\theta\sim\Unif([0,1])$. We treat $\sigma_n^2=\Var_{\Bx^{(1)},\ldots,\Bx^{(M)}\cond\Btheta}(-w\sum_{i=1}^M \Delta(\Bxobs,\Bx^{(i)})/M)$ as a constant with respect to $\Btheta$ and estimate it as in Section \ref{subsec:toymodels}. % although in reality $\sigma_n^2$ depends to some extent on the magnitude of the discrepancy. 
We use $t_{\text{init}}=10$, $i_{\text{MH}} = 10^5$, $\Btheta^{(0)}=(2.0, 0.4, 0.2)$ and $\BSigma_0=\diag(0.5, 0.1, 0.02)^2$ for each experiment. 

% discuss results
Figure \ref{fig:res_bacterial_eval} shows the results with some choices of $w$ and $M$. A few hundred ($w=10$ or $20$) or a thousand ($w=40$) model simulations are sufficient to produce useful approximations. \gpmh{} is hence a feasible option also in this setting. %and that \naive{} again uses more evaluations on average for a given fixed $\epsi$ as the other methods. 
Naive method interestingly works well and clearly produces the best final accuracy when $w=40$ and $M=4$. \epoer{} and especially \epoe{} rely more on the GP model than naive and their performance is presumably more affected by the misspecified constant noise assumption of the GP. While naive requires more evaluations than the other methods on average to make the unconditional error smaller than $\epsi$, it appears that the extra evaluations are here not redundant and cause later accept/reject decisions to be made more accurately than required. The accuracy of the accept/reject decisions by \epoe{} and \epoer{} are more closely centred near the upper bound $\epsi$. 

\begin{figure}[hbt]
\centering
\includegraphics[width=1\textwidth]{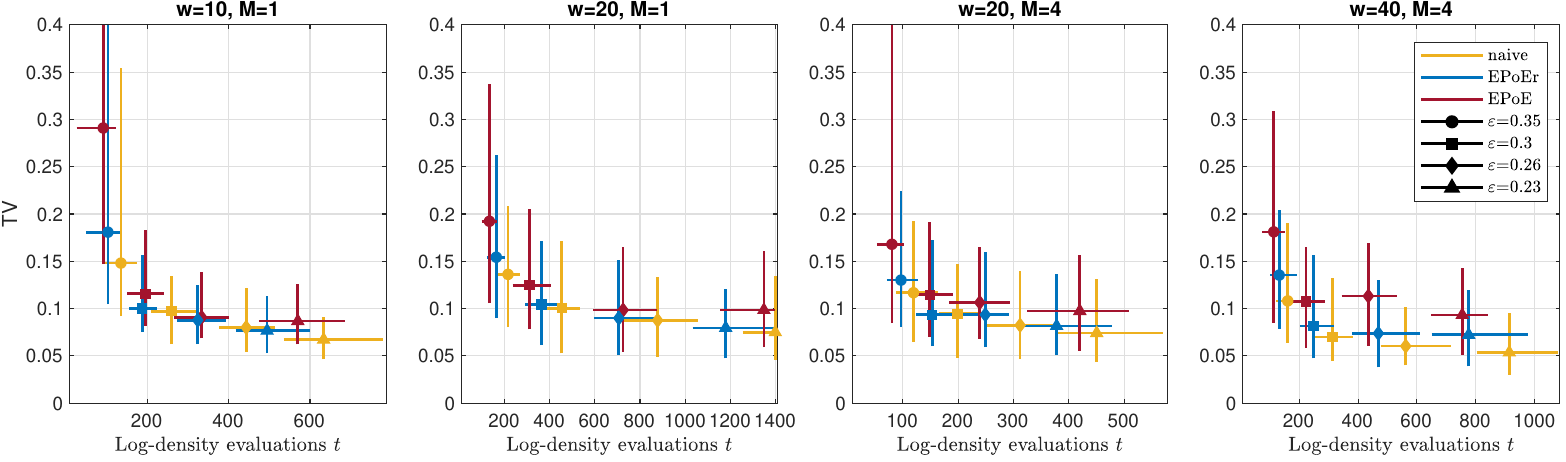}
\caption{Accuracy of the marginal generalised posterior approximation of \gpmh{} in the bacterial infections experiment as a function of log-density evaluations at the final iteration $i_{\text{MH}}=10^5$. The results for \mhblfi{} were similar.} \label{fig:res_bacterial_eval}
\end{figure}

% example genpost
Some estimated generalised posteriors are illustrated in Figure \ref{fig:bacterial_posts_example} in the \epoer{} case with $\epsi=0.23$, $M=1$ and $w=20$. We observe that the choice $w=20$ (which is approximately equal to $2(\Var_{\Bx\cond\Btheta}\Delta(\Bxobs,\Bx))^{-1/2}$ for any $\Btheta$ near the mode) produces fairly similar density estimate as the ABC posterior obtained using the ABC-MCMC sampler by \citet{Marjoram2003}, the same discrepancy as for \gpmh{} and a uniform kernel with threshold $h=1.0$. %$\hat{\Btheta}=(3.6, 0.6, 0.1)$
This suggests that, at least in some situations, our generalised posterior may be considered as an approximation to the ABC posterior itself. % though this requires further studies...
% comment sample-efficiency, ABC comparison.

\begin{figure}[hbtp] % example posterior - bacterial
\centering
\includegraphics[width=0.85\textwidth]{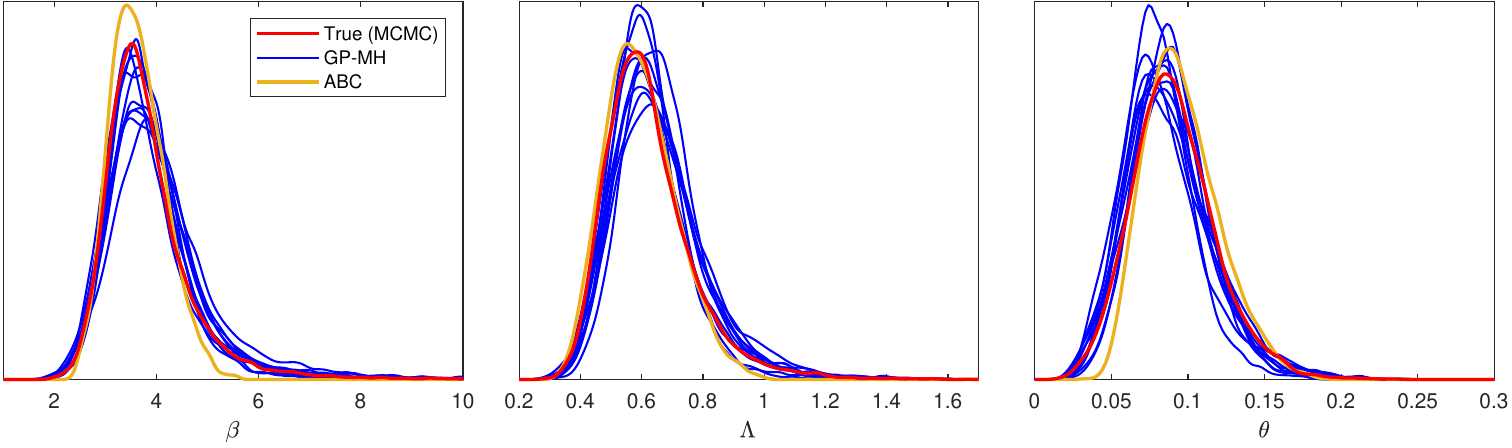}
\caption{Comparison of the ground-truth generalised posterior (red line) and $10$ examples of the estimated posterior produced by \gpmh{} (blue lines) in the case of bacterial infections model. The ABC posterior (yellow line) is also shown for comparison. 
} \label{fig:bacterial_posts_example}
\end{figure}

\section{Summary and discussion}\label{sec:disc}

% brief summary
We proposed a new sample-efficient framework for approximate Bayesian inference by combining Metropolis-Hastings sampling with Gaussian process emulation. % and sequential Bayesian experimental design. 
The resulting \gpmh{} implementation is mainly intended for low-dimensional problems ($p<10$) when only a small number of possibly noisy likelihood evaluations (e.g.~$t<10^3$) is available. The log-likelihood needs to be a smooth function in the highest density region and its evaluations need to be (approximately) Gaussian distributed. 
We gave a Bayesian decision-theoretic justification for various parts of the method and discussed the relationship between \gpmh{} and ``Bayesian optimisation-like'' B(O)LFI frameworks for approximate Bayesian inference by \citet{Gutmann2016,Jarvenpaa2019_sl}.

% discuss the experiments and the main results
Two particular implementations of the proposed framework were investigated numerically. 
% discuss the difference between approx.\mh{} and GP-based results
Our \gpmh{} and \mhblfi{} implementations produced similar posterior approximations which is unsurprising given their connection (Section \ref{subsec:gpmcmcasblfi}) and because most of the log-likelihood evaluations are typically collected during the early iterations which are neglected anyway as ``burn-in''. %For example, $60-80\%$ of the log-SL evaluations in our Ricker and theta-Ricker experiments occurred during the first $10^4$ iterations. 
The slightly changing target distribution of \gpmh{} thus does not appear problematic. 
% discuss EPoE etc. and their sample-efficiency
Our experiments show that Bayesian sequential design strategies can provide worthwhile improvements to sample-efficiency over a naive strategy. However, we believe that their full potential did not realise due to practical challenges with surrogate modelling and as the \epoe{} strategy does not account for the resulting posterior approximation directly \citep[unlike IMIQR by][]{Jarvenpaa2019_sl} but is designed to minimise the evaluations needed to make the current \mh{} accept/reject decision within the required accuracy. %Careful modelling of the variance of the log-likelihood evaluations seems to be needed for full advantage of such strategies 
This makes the more heuristic \epoer{} a practical choice for \gpmh{}. % and because it does not require auxiliary optimisation of the acquisition function.

% practical improvements over previous methods
%Our method requires substantially less simulations than pseudo-marginal MCMC. 
\gpmh{} aims to construct the GP surrogate around the highest posterior density region which makes \gpmh{} robust to possible violations of GP modelling assumptions near the parameter boundaries. This is in contrary to earlier B(O)LFI methods where the log-likelihood is modelled and the acquisition function optimised over the whole parameter space. On the other hand, posterior densities that do behave smoothly everywhere and feature many distant modes are more ideal for B(O)LFI. Computations needed to apply \epoe{}(r) methods are substantially more efficient than those developed by \citet{Jarvenpaa2019_sl}. % because the integration over the parameter space is avoided. % yet produces fairly similar sample-efficiency. 
For example, the GP computations in a typical run with the theta-Ricker model using $2\cdot10^5$ iterations and around $500$ evaluations took less than $15$ minutes while the global optimisation of the IMIQR acquisition function once in \BLFI{} already takes up to one minute. \epoer{} and naive methods do not even need auxiliary optimisation which also simplifies the implementation. 
% one variant also avoid the auxiliary optimisation -> good in high dimensions or when discrete parameter space
%Another key advantage of \gpmh{} is its conceptual simplicity. In particular, no additional variational inference approximation as in \citet{Acerbi2018,Acerbi2020} or preliminary stages with auxiliary MCMC as in \citet{Drovandi2018,Wiqvist2018} is needed. 

% importance of initialisation location
A potential downside over B(O)LFI is that \gpmh{} may require more careful initialisation. As with all \mh{} samplers, selection of good proposal covariance can be tricky. Also, if \gpmh{} is started far from the highest density region or if the log-likelihood behaves irregularly around the initial point, difficulties with traversing to the highest density region or GP fitting may still emerge. 
The trade-off between accuracy and computational cost needs to be adjusted (in a somewhat nontransparent fashion) using the parameter $\epsi$. 
A current good practice would be to first run the algorithm using a fairly large $\epsi$ and guided by the experiments of this paper. If the resulting posterior appears inaccurate or if less evaluations were used than anticipated, the algorithm can be rerun with decreased $\epsi$ and with the evaluations from the previous run used as initial data. 
%Selection of suitable GP prior for \gpmh{}, as well as for all GP-based inference methods, remains nontrivial. 

% future work
There is room to further enhance and extend the proposed approach. Detailed analysis of the interplay between the error tolerance $\epsi$, the number of total evaluations needed and the accuracy of the resulting posterior approximation would be beneficial but seems difficult to establish. 
%\gpmh{} inherits the well-known downsides of \mh{} such as its random-walk behaviour. This is not a major concern because most accept/reject decisions are done efficiently based on the GP surrogate alone. One could still investigate if the proposed approach could be used together with other MCMC samplers such as Hamiltonian Monte Carlo and if this leads to improvements. 
For global optimisation, one could extend \gpmh{} to work with simulated annealing. 
The use of \gpmh{} for the very challenging case of high-dimensional parameters, improved modelling of the log-likelihood variance and alternative criteria for controlling the accuracy of the \mh{} accept/reject test could also be investigated. 

%We also only focused on the scenario where a single chain is generated but generalisation seems possible. 
%We used adaptive Metropolis algorithm \citep{Haario2001} as the MCMC method but it could be possible to additionally use delayed rejection to improve adaptation \citep{Haario2006} in the beginning or parallel tempering as in ... to improve mixing with multimodal posteriors. 
%It is possible to remove some old noninformative points, e.g.~those that have been collected in the very beginning when the modal area has not yet been located and which have very small likelihood values. %, but we have not investigated this.
% as with all approximate inference methods, some care is needed when using the resulting approximate posterior.

%\subsubsection*{Acknowledgements}
\vbox{% forces the acknowledgements text + its title(!) to stay together on one page
\acks{Some of the computations were performed on resources provided by UNINETT Sigma2 - the National Infrastructure for High Performance Computing and Data Storage in Norway. This research was funded by the Norwegian Research Council FRIPRO grant no.~299941 and by the European Research Council grant no.~742158.}
}

%\newpage
%\onecolumn
\appendix
%\pagenumbering{Roman}
%\setcounter{page}{1}
\numberwithin{equation}{section}
\numberwithin{figure}{section}
\numberwithin{theorem}{section}

\section{Background on likelihood-free inference}

\subsection{Synthetic likelihood} \label{subsec:lfi}

LFI methods mainly differ on how the model simulations are used to approximate the intractable likelihood function. 
% SL
%The methodology developed in this paper is especially useful for likelihood-free inference where the analytical form of the likelihood function $\pi(\Bx\cond\Btheta)$ is either unavailable or too expensive to evaluate. See \citet{Marin2012,Lintusaari2016,Sisson2019,Cranmer2019} for recent reviews. 
%An application of this paper is the synthetic likelihood method \citep{Wood2010,Price2018} when the model simulations are computationally costly. 
%As mentioned, our framework however also applies in other scenarios with expensive and possibly noisy likelihood evaluations. 
%
Synthetic likelihood \citep{Wood2010,Price2018} is a parametric approximation to the likelihood $\pi(\Bxobs\cond\Btheta)$ formed by first replacing the full data $\Bxobs\in\reals^d$ with summary statistics $S(\Bxobs)$, where $S:\reals^d\rightarrow\reals^s, s < d$, and then assuming 
\begin{equation}
\pi(S(\Bxobs) \cond \Btheta) 
= \Normal_s(S(\Bxobs)\cond\Bmu_{\Btheta},\BSigma_{\Btheta}). 
%
% Normal pdf:
%= ({\det(2\pi\BSigma_{\Btheta})})^{-1/2}{\e^{-(S(\Bxobs)-\Bmu_{\Btheta})\T \BSigma^{-1}_{\Btheta} (S(\Bxobs)-\Bmu_{\Btheta})/2}}{}. 
\label{eq:SL}
\end{equation}
The approximation results from replacing the full data $\Bxobs$ with potentially insufficient summary statistics $S(\Bxobs)$ and from the Gaussianity assumption in (\ref{eq:SL}) that rarely holds exactly. 
The unknown expectation $\Bmu_{\Btheta}\in\reals^{s}$ and covariance matrix  $\BSigma_{\Btheta}\in\reals^{s\times s}$ in (\ref{eq:SL}) are estimated from $N$ repeated simulations for each proposed $\Btheta$ using
\begin{equation}
% mean:
\hat{\Bmu}_{\Btheta} = \frac{1}{N}\sum_{i=1}^{N}S(\Bx^{(i)}_{\Btheta}), 
%
% var:
\quad \hat{\BSigma}_{\Btheta} = \frac{1}{N-1}\sum_{i=1}^{N} (S(\Bx^{(i)}_{\Btheta}) - \hat{\Bmu}_{\Btheta}) (S(\Bx^{(i)}_{\Btheta}) - \hat{\Bmu}_{\Btheta})\T, \label{eq:SL_ML}
\end{equation}
where $\Bx^{(i)}_{\Btheta} \sim \pi(\cdot\cond\Btheta)$ for $i=1,\ldots,N$. An estimator for the likelihood function is then obtained by replacing the unknown $\Bmu_{\Btheta}$ and $\BSigma_{\Btheta}$ in (\ref{eq:SL}) with the point estimates $\hat{\Bmu}_{\Btheta}$ and $\hat{\BSigma}_{\Btheta}$ in (\ref{eq:SL_ML}). See Section \ref{appsec:gpremarks} for some remarks on this estimator. The resulting log-synthetic likelihood evaluations (abbreviated as log-SL) are noisy because $N$ cannot in practice be large for computational reasons. 
%
% other SL-like methods like LFIRE 
Various extensions of SL have also been proposed \citep{An2018,An2019,Nott2019,Thomas2018} which similarly produce noisy log-likelihood approximations.

\subsection{Approximate Bayesian computation} \label{subsec:abc}

The ABC likelihood is formed in a non-parametric fashion so that 
\begin{align}
    \pi(\Bxobs \cond \Btheta) \approx \pi_{\text{ABC}}(\Bxobs \cond \Btheta) \eqdef \mean_{\Bx\cond\Btheta}K_h(\Delta(\Bxobs,\Bx)) = \int K_h(\Delta(\Bxobs,\Bx))\pi(\Bx\cond\Btheta)\ud\Bx, \label{eq:abc_lik}
\end{align}
where the discrepancy $\Delta(\Bxobs,\Bx)$ quantifies the similarity between the observed data $\Bxobs$ and the pseudo-data $\Bx$ simulated from the model likelihood $\pi(\Bx\cond\Btheta)$. The data is commonly reduced to summary statistics (as in the SL case) so that, for example, $\Delta(\Bxobs,\Bx)=||S(\Bxobs)-S(\Bx)||$ where $||\cdot||$ is some suitable norm although other choices are also possible. The kernel $K_h(r)$ in (\ref{eq:abc_lik}) with with bandwidth $h>0$ further controls the accuracy of the ABC approximation. A common choice is the uniform kernel $K_h(r) \propto \indic_{r\leq h}$ so that
\begin{align}
    \pi_{\text{ABC}}(\Bxobs \cond \Btheta) \propto \prob_{\Bx\cond\Btheta}(\Delta(\Bxobs,\Bx) \leq h). \label{eq:abc_prob}
\end{align}
Under some regularity conditions and assuming $S$ above is a sufficient statistic, (\ref{eq:abc_prob}) converges to the exact likelihood as $h\rightarrow 0$. However, in practice, one usually needs to resort to non-sufficient statistics and $h$ cannot be taken too small because simulation of pseudo-data that gives a close match to the observed summary statistic $S(\Bx)$ is inefficient. 

% unbiased estimator for the ABC lik
An unbiased estimator for the ABC likelihood (\ref{eq:abc_lik}) is obtained as
\begin{align}
    \pi_{\text{ABC}}(\Bxobs \cond \Btheta) \approx \frac{1}{M}\sum_{i=1}^M K_h(\Delta(\Bxobs,\Bx^{(i)})), \label{eq:abc_unbiased}
\end{align}
where $\Bx^{(i)}\sim\pi(\cdot\cond\Btheta)$ for $i=1,\ldots,M$. 
Further details on ABC methods can be found e.g.~in \citet{Sisson2019}.

\section{Proofs and mathematical derivations of Section \ref{sec:gpe} and \ref{sec:acq}} \label{appesec:proofs}

%\begin{proof}[Proof of Proposition \ref{prop:med}]
\begin{proof}[Proposition \ref{prop:med}]
% conditional error - proof
We use the fact that the set of medians of a random variable is a closed interval which we denote as $[m_1,m_2]$ where $m_1\leq m_2$. That is, $m$ is a median of $\gamma$ $\iff$ $m\in[m_1,m_2]$. From the definition of the median we see that $\prob(\gamma < m) \leq 1/2$ for any median $m$. Suppose $m_1<m\leq m_2$. Since then $1/2\leq\prob(\gamma\leq m_1)\leq\prob(\gamma< m)$ it follows that $\prob(\gamma<m)=1/2$ in this case. 

We consider fixed $u$, treat the conditional error $\conderr_{u,\hatgamma}$ as a function of $\hatgamma$ and write it as 
\begin{align*}
    \conderr_{u,\hatgamma} = \begin{cases}\prob(\gamma<u) &\textnormal{if } \hatgamma \geq u, \\ \prob(\gamma\geq u) &\textnormal{if } \hatgamma < u, \end{cases}
\end{align*}
to ease up the analysis to follow. 
This function is clearly bounded and consists of two values depending whether $\hatgamma\geq u$ or $\hatgamma<u$. 

Suppose first $u=m_1$. Then $\prob(\gamma<m_1)\leq 1/2$ and consequently $\prob(\gamma\geq m_1)\geq 1/2$. Thus, if we choose $\hatgamma\geq m_1$ we get the minimum. In particular, we can choose $\hatgamma$ to be any median. 

Suppose now $m_1<u\leq m_2$. Then $\prob(\gamma<u)=\prob(\gamma\geq u) = 1/2$ so that any choice of $\hatgamma$ will do. Again, we can choose $\hatgamma$ to be any median. 

Suppose $u<m_1$. Since $u$ is not a median it must hold that $\prob(\gamma\geq u)<1/2$ or $\prob(\gamma\leq u)<1/2$. The former option is not possible because it would contradict with the facts that CDF is a non-decreasing function and $m_1$ is median. Thus $\prob(\gamma\leq u)<1/2$ must hold and it clearly follows that $\prob(\gamma<u) < 1/2$ so we can choose $\hatgamma$ to be any median to get this minimum value. Similarly we see that if $u>m_2$, then $\prob(\gamma\geq u) < 1/2$ so that we can choose $\hatgamma$ to be any median to get the minimum.
We have thus shown that the median of $\gamma$ minimises the conditional error given each fixed value of $u$. 

% unconditional error - proof
Since any median minimises the conditional error with each fixed $u$, it follows that any median minimises also the conditional error integrated over $u\in[0,1]$ which is the unconditional error. 
%
% unconditional error - alternative proof
Alternatively, we can also see this as follows. We write
\begin{align*}
    \conderr_{u,\hatgamma} &= \prob(\gamma < u\cond u)\indic_{\hatgamma \geq u} + \prob(\gamma \geq u\cond u)\indic_{\hatgamma < u} \\
    &= \prob(\gamma < u\cond u)\indic_{\hatgamma \geq u} + (1-\prob(\gamma < u\cond u))(1-\indic_{\hatgamma \geq u}) \\
    &= \indic_{\hatgamma \geq u}(2\prob(\gamma < u\cond u)-1) + c_u,
\end{align*}
where $c_u=\prob(\gamma \geq u\cond u)$ does not depend on $\hatgamma$. 
It follows that
\begin{align*}
    \unconderr_{\hatgamma} = \int_0^{\hatgamma}(2\prob(\gamma < u\cond u)-1) \ud u + \int_0^1 c_u \ud u.
    %\label{eq:uncerrf}
\end{align*}
The claim then follows from the facts $\prob(\gamma < m_1)\leq 1/2$, $\prob(\gamma < m) = 1/2$ for $m\in(m_1,m_2]$ and $\prob(\gamma < u)>1/2$ for $u>m_2$. 
\end{proof}

%% Proof for the formula of the conditional and unconditional error
%Next we justify (\ref{eq:conderr_gp}) and (\ref{eq:unconderr_gp}). 
\noindent\textbf{Justification for (\ref{eq:conderr_gp}) and (\ref{eq:unconderr_gp}).} 
The former equation is obtained as follows 
\begin{align*}
    \conderr_{t,u,\hatgamma}(\Btheta,\Btheta') %&= \prob(\alpha < u\cond u)\indic_{\med(\alpha) \geq u} + \prob(\alpha \geq u\cond u)\indic_{\med(\alpha) < u} \\
    &= \prob(\gamma_f(\Btheta,\Btheta') < u\cond u)\indic_{\med(\gamma_f(\Btheta,\Btheta')) \geq u} + \prob(\gamma_f(\Btheta,\Btheta') \geq u\cond u)\indic_{\med(\gamma_f(\Btheta,\Btheta')) < u} \\
    &= \prob(\log\gamma_f(\Btheta,\Btheta') < \logu\cond u)\indic_{\med(\gamma_f(\Btheta,\Btheta')) \geq u} \!+\! \prob(\log\gamma_f(\Btheta,\Btheta') \geq \logu\cond u)\indic_{\med(\gamma_f(\Btheta,\Btheta')) < u} \\
    &= \Phi\left( \frac{\mu_t(\Btheta,\Btheta') - \logu}{\sigma_t(\Btheta,\Btheta')} \right)\!\indic_{\mu_t(\Btheta,\Btheta') < \logu} 
    + \Phi\left( \frac{ \logu - \mu_t(\Btheta,\Btheta')}{\sigma_t(\Btheta,\Btheta')} \right)\!\indic_{\mu_t(\Btheta,\Btheta') \geq \logu} \\
    &= \Phi\left( \frac{\min\{\mu_t(\Btheta,\Btheta') - \logu, \logu - \mu_t(\Btheta,\Btheta') \}}{\sigma_t(\Btheta,\Btheta')} \right) \\
    &= \Phi\left( -\frac{|\mu_t(\Btheta,\Btheta') - \logu|}{\sigma_t(\Btheta,\Btheta')} \right).
\end{align*}

We state a Lemma that we need to derive (\ref{eq:unconderr_gp}):
\begin{lemma} \label{lemma:integral1}
Suppose $\sigma>0$ and $0\leq a \leq b$. Then 
\begin{align*}
    \int_a^b \Phi\left( \frac{\log u - \mu}{\sigma} \right) \ud u 
    = \e^{\mu}\left[ e^{\beta}\Phi\left(\frac{\beta}{\sigma}\right)
    - e^{\alpha}\Phi\left(\frac{\alpha}{\sigma}\right)
    + \e^{\sigma^2/2}\left( \Phi\left(\frac{\alpha-\sigma^2}{\sigma}\right) - \Phi\left(\frac{\beta-\sigma^2}{\sigma}\right) \right)
    \right], %\label{eq:integral1}
\end{align*}
where $\alpha \eqdef \log a - \mu$ and $\beta \eqdef \log b - \mu$. 
\end{lemma}
\begin{proof}
We first use change of variables $x={(\log u - \mu)}/{\sigma}$ to compute
\begin{align*}
    \int_a^b \Phi\left( \frac{\log u - \mu}{\sigma} \right) \ud u 
    = \sigma\e^{\mu}\int_{\alpha/\sigma}^{\beta/\sigma} \e^{\sigma x}\Phi(x) \ud x. 
\end{align*}
The final equality then follows by using the equation 101.000 in \citet{Owen1980} and some straightforward simplifications. 
\end{proof}

To shorten the notation, we write $\mu_t$ for $\mu_t(\Btheta,\Btheta')$ and $\sigma_t$ for $\sigma_t(\Btheta,\Btheta')$. We can write
\begin{align*}
    \unconderr_{t,\hatgamma}(\Btheta,\Btheta') &= \int_0^1 \Phi\left( -\frac{|\mu_t - \log u|}{\sigma_t} \right) \ud u,
\end{align*}
%
% CASE 1:
from which we see that if $\mu_t\geq 0$, then 
\begin{align*}
    \unconderr_{t,\hatgamma}(\Btheta,\Btheta') &= \int_0^1 \Phi\left( \frac{\log u - \mu_t}{\sigma_t} \right) \ud u.
\end{align*}
The first case of (\ref{eq:unconderr_gp}) then follows immediately by using Lemma \ref{lemma:integral1} and some straightforward simplifications. 

% CASE 2:
Suppose now that $\mu_t< 0$. Then $0<\e^{\mu_t}<1$ and we can write
\begin{align}
    \unconderr_{t,\hatgamma}(\Btheta,\Btheta') &= \int_0^{\e^{\mu_t}}\Phi\left( \frac{\log u - \mu_t}{\sigma_t} \right) \ud u + \int_{\e^{\mu_t}}^{1}\Phi\left( \frac{\mu_t - \log u}{\sigma_t} \right) \ud u \nonumber \\
    &= 1-\e^{\mu_t} + \int_0^{\e^{\mu_t}}\Phi\left( \frac{\log u - \mu_t}{\sigma_t} \right) \ud u - \int_{\e^{\mu_t}}^{1}\Phi\left( \frac{\log u - \mu_t}{\sigma_t} \right) \ud u. \label{eq:tyuio}
\end{align}
We use Lemma \ref{lemma:integral1} to compute both integrals in (\ref{eq:tyuio}) and after some straightforward computations we obtain the second case of (\ref{eq:unconderr_gp}).

%\begin{proof}[Proof of Proposition \ref{prop:error_formulas}]
\begin{proof}[Proposition \ref{prop:error_formulas}]
%% general part
Based on the GP model, we have
\begin{align*}
    \By^*\cond\Btheta^*,\ddata_t \sim \Normal_b(m_t(\Btheta^*),c_t(\Btheta^*,\Btheta^*) + \BLambda^{\!*}). 
\end{align*}
Then, by Lemma 5.1 in \citet{Jarvenpaa2019_sl}, it follows that\footnote{Lemma 5.1 in fact shows the result for the GP mean and variance functions in a single $\Btheta$-location only but it is easy to see that the result immediately extends for the more general case considered here.}
\begin{align*}
    \begin{bmatrix}
    m_{t+b}^*(\Btheta) \\ m_{t+b}^*(\Btheta')
    \end{bmatrix} \!\cond\Btheta^*,\ddata_t
    &\sim \Normal_2\left( \begin{bmatrix} m_t(\Btheta) \\ m_t(\Btheta') \end{bmatrix}, \begin{bmatrix} \tau_t^2(\Btheta;\Btheta^*) & \omega_t(\Btheta,\Btheta';\Btheta^*) \\ \omega_t(\Btheta,\Btheta';\Btheta^*) & \tau_t^2(\Btheta';\Btheta^*) \end{bmatrix} \right), \\
    c_{t+b}^*(\Btheta,\Btheta') &= c_t(\Btheta,\Btheta') - \omega_t(\Btheta,\Btheta';\Btheta^*), 
\end{align*}
where $*$ is used to emphasise that these quantities depend on $\Btheta^*$ and possibly also $\By^*$ via $\ddata^*$. 
It follows that
\begin{align*}
% m_{t}(\Btheta') - m_{t}(\Btheta) + a(\Btheta,\Btheta')
    \mu_{t+b}^*(\Btheta,\Btheta') %&= m_{t+b}^*(\Btheta') - m_{t+b}^*(\Btheta) + a(\Btheta,\Btheta') 
    \cond\Btheta^*,\ddata_t
    &\sim \Normal_1(\mu_t(\Btheta,\Btheta'), \xi_t^2(\Btheta',\Btheta;\Btheta^*)),
\end{align*}
where $\mu_t(\Btheta,\Btheta')$ is given by (\ref{eq:mu}) and $\xi_t^2(\Btheta',\Btheta;\Btheta^*)$ by (\ref{eq:xi}). 
We also see that
\begin{align}
\begin{split}
    \sigma_{t+b}^{2*}(\Btheta,\Btheta') &= s_t^2(\Btheta)-\tau_t^2(\Btheta;\Btheta^*) + s_t^2(\Btheta')-\tau_t^2(\Btheta';\Btheta^*) - 2(c_t(\Btheta,\Btheta')-\omega_t^2(\Btheta',\Btheta;\Btheta^*)) \\
    &= s_t^2(\Btheta) + s_t^2(\Btheta') - 2c_t(\Btheta,\Btheta') 
    - (\tau_t^2(\Btheta;\Btheta^*)+\tau_t^2(\Btheta';\Btheta^*) - \omega_t^2(\Btheta',\Btheta;\Btheta^*)) \\
    &= \sigma_{t}^{2}(\Btheta,\Btheta') - \xi_t^2(\Btheta',\Btheta;\Btheta^*). 
\end{split} \label{eq:updgpvarformula}
\end{align}
To shorten the notation, we once again drop ``$(\Btheta,\Btheta')$'' from various formulas. For example, we write $\mu_t$ for $\mu_t(\Btheta,\Btheta')$ and $\xi_t(\Btheta^*)$ for $\xi_t(\Btheta',\Btheta;\Btheta^*)$. % and $\sigma_t$ for $\sigma_t(\Btheta,\Btheta')$. 

%% conderror
We write the conditional error as
\begin{align}
    \conderr_{t+b,u,\hatgamma}(\Btheta,\Btheta') 
    = 2\Phi\left(\frac{\logu-\mu_{t+b}^*}{\sigma_{t+b}^*}\right)\indic_{\mu_{t+b}^*\geq\logu} 
    - \indic_{\mu_{t+b}^*\geq\logu} 
    + 1 - \Phi\left(\frac{\logu-\mu_{t+b}^*}{\sigma_{t+b}^*}\right).
    \label{eq:cv1}
\end{align}
We then compute
\begin{align*}
    \mean_{\mu_{t+b}^*\cond\Btheta^*,\ddata_t}\left[ \indic_{\mu_{t+b}^*\geq\logu} \right] &= 1 - \int_{-\infty}^{\logu}\Normal(\mu_{t+b}^*\cond\mu_{t},\xi_t^2(\Btheta^*)) %\\
    = 1 - \Phi\left(\frac{\logu-\mu_{t}}{\xi_t(\Btheta^*)}\right), \\
    \mean_{\mu_{t+b}^*\cond\Btheta^*,\ddata_t}\left[ 1 - \Phi\left(\frac{\logu-\mu_{t+b}^*}{\sigma_{t+b}^*}\right)\right] &= 1 - \Phi\left(\frac{\logu-\mu_{t}}{\sigma_{t}}\right),
\end{align*}
where we have used equation 3.82 in \citet{Rasmussen2006} and the fact $\Phi(z)=1-\Phi(-z)$. 
The first term in (\ref{eq:cv1}) requires some more work. We write
\begin{align}
    %\begin{split}
    &\mean_{\mu_{t+b}^*\cond\Btheta^*,\ddata_t}\left[ 2\Phi\left(\frac{\logu-\mu_{t+b}^*}{\sigma_{t+b}^*}\right)\indic_{\mu_{t+b}^*\geq\logu} \right] \nonumber \\
    &\myquad= 2\int_{\logu}^{\infty}  \Phi\left(\frac{\logu-\mu_{t+b}^*}{\sigma_{t+b}^*}\right) \Normal(\mu_{t+b}^* \cond \mu_t, \xi_t^2(\Btheta^*)) \ud \mu_{t+b}^* \nonumber \\
    %\end{split}
    %
    &\myquad= 2\Phi\left(\frac{\logu-\mu_{t}}{\sigma_{t}}\right) - 2\int_{-\infty}^{\logu}  \Phi\left(\frac{\logu-\mu_{t+b}^*}{\sigma_{t+b}^*}\right) \Normal(\mu_{t+b}^* \cond \mu_t, \xi_t^2(\Btheta^*)) \ud \mu_{t+b}^* \nonumber \\
    &\myquad= 2\Phi\left(\frac{\logu-\mu_{t}}{\sigma_{t}}\right) - 2\int_{-\infty}^{\frac{\logu-\mu_{t}}{\xi_t(\Btheta^*)}}  \Phi\left(\frac{\logu-\mu_{t}-\xi_t(\Btheta^*)x}{\sigma_{t+b}^*}\right) \Normal(x \cond 0, 1) \ud x, \label{eq:xcvb}
\end{align}
where we used the transformation $x=(\mu_{t+b}^*-\mu_{t})/\xi_t(\Btheta^*)$. To compute the integral in (\ref{eq:xcvb}) we use the equation 10,010.1 in \citet{Owen1980}. After some straightforward computations we obtain 
\begin{align}
    \int_{-\infty}^{\frac{\logu-\mu_{t}}{\xi_t(\Btheta^*)}}  \Phi\left(\frac{\logu-\mu_{t}-\xi_t(\Btheta^*)x}{\sigma_{t+b}^*}\right) \Normal(x \cond 0, 1) \ud x 
    = \bvn\left(\frac{\logu-\mu_{t}}{\sigma_{t}},\frac{\logu-\mu_{t}}{\xi_t(\Btheta^*)}; \frac{\xi_t(\Btheta^*)}{\sigma_{t}}\right), \label{eq:bvnformula}
\end{align}
where $\bvn(h,k,\rho)$ is the pdf of a bivariate Gaussian with unit variances and correlation coefficient $\rho$ evaluated at $(h,k)\T$. We use the connection between $\bvn$ and \owen{} given by equation 3.1 in \citet{Owen1980}
(the first case of which applies here because $hk=(\logu-\mu_{t})^2/(\sigma_{t}\xi_t(\Btheta^*))\geq 0$ and because $hk=0 \iff h=k=0$ hold with (\ref{eq:bvnformula})) and the fact $T(h,0)=0$ for any $h\in\reals$, to further obtain
\begin{align*}
    &\bvn\left(\frac{\logu-\mu_{t}}{\sigma_{t}},\frac{\logu-\mu_{t}}{\xi_t(\Btheta^*)}; \frac{\xi_t(\Btheta^*)}{\sigma_{t}}\right) \\
    &\myquad= \half\Phi\left(\frac{\logu-\mu_{t}}{\sigma_{t}}\right) 
    + \half\Phi\left(\frac{\logu-\mu_{t}}{\xi_t(\Btheta^*)}\right) 
    - T\left( \frac{\logu-\mu_t}{\sigma_t}, \frac{\sqrt{\sigma_t^2 - \xi_t^2(\Btheta^*)}}{ \xi_t(\Btheta^*)} \right).
\end{align*}
Once we combine the equations above, we see that all the $\Phi(\cdot)$-terms cancel out and we are left with (\ref{eq:exp_var}). 

%% unconderr
The formula (\ref{eq:exp_unconderr}) follows immediately from above because we can change the order of integration over $u\in[0,1]$ and the expectation with respect to $\pi(\mu_{t+b}^*\cond\mu_{t},\xi_t(\Btheta^*))$ using Fubini's theorem. 

%% variance
By using (\ref{eq:kappavar}) and the fact $\Phi(z)=1-\Phi(-z)$, we write the expected variance of $\accrejindic_{u,f}$ as 
\begin{align*}
    L_t^{\textnormal{v},u}(\Btheta,\Btheta';\Btheta^*) 
    &= \int_{-\infty}^{\infty} \left[ \Phi\left(\frac{\logu-\mu_{t+b}^*}{\sigma_{t+b}^*}\right) - \Phi^2\left(\frac{\logu-\mu_{t+b}^*}{\sigma_{t+b}^*}\right)\right] \Normal(\mu_{t+b}^* \cond \mu_t, \xi_t^2(\Btheta^*)) \ud \mu_{t+b}^*.
\end{align*}
We then recognise that this integral is of the same form as in the proof of Lemma 3.1 in \citet{Jarvenpaa2018_acq} from which (\ref{eq:exp_var}) follows. 
\end{proof}

%\begin{proof}[Proof of Proposition \ref{prop:xi}]
\begin{proof}[Proposition \ref{prop:xi}]
The result for the variance of the log-\mh{} ratio follows immediately from $\sigma^2_{t+b}(\Btheta,\Btheta')=\sigma_t^2(\Btheta,\Btheta') - \xi_t^2(\Btheta,\Btheta';\Btheta^*)$ which is given by (\ref{eq:updgpvarformula}). 
The \owen{} satisfies $T(h,a) = \frac{1}{2\pi}\int_{0}^{a}\frac{\e^{-h^2(1+x^2)/2}}{1+x^2} \ud x$. This fact shows that the function $a\mapsto T(h,\sqrt{a}), a\geq 0$ is strictly increasing with any fixed $h\in\reals$. It follows that $L_t^{\conderr,u}(\Btheta,\Btheta';\Btheta^*)$ is minimised when $\frac{\sigma_t^2(\Btheta,\Btheta') - \xi_t^2(\Btheta,\Btheta';\Btheta^*)}{ \xi_t^2(\Btheta,\Btheta';\Btheta^*)}$ is minimised which clearly happens when $\Btheta^*$ is chosen as in (\ref{eq:xiopt}) since $0 \leq \xi_t^2(\Btheta,\Btheta';\Btheta^*) \leq \sigma_t^2(\Btheta,\Btheta')$. (These two inequalities follows from (\ref{eq:gpomegav2}) and the fact that variance is always non-negative.) 
This reasoning holds with any $u>0$ so that $L_t^{\unconderr}(\Btheta,\Btheta';\Btheta^*)$ is also minimised by this choice of $\Btheta^*$. 
The proof for the case of $L_t^{\textnormal{v},u}(\Btheta,\Btheta';\Btheta^*)$ is similar as for $L_t^{\conderr,u}(\Btheta,\Btheta';\Btheta^*)$. 
\end{proof}

\section{Some theoretical analysis} \label{appsec:extra_analysis}

In the following we present some theoretical analysis for understanding \gpmh{}. The key results are briefly summarised in the main text as Section \ref{sec:theory}. 

\subsection{Evaluation locations} \label{subsec:eval_loc_analysis}

% general
We analyse the optimal evaluations locations. We mostly limit our attention to the sequential case $b=1$ where $\Btheta^*$ consists of a single parameter. Throughout this section we assume $\Btheta,\Btheta'\in\Theta$ are arbitrary distinct points.

\subsubsection{Analysis of \texorpdfstring{$\tau_t$}{tau}-function}

%% tau, noiseless
%In the following we suppose $s_t(\Btheta)>0$. (If $s_t(\Btheta)=0$, then we would already know $f$ exactly at $\Btheta$.)
We start by stating some results regarding $\tau_t^2(\Btheta;\Btheta^*)$ in (\ref{eq:gptau}) which gives the reduction of GP variance at $\Btheta$ resulting from evaluations at $\Btheta^*$. 
When $b=1$ we can write (\ref{eq:gptau}) as 
\begin{align*}
    \tau_t^2(\Btheta;\Btheta^*) = \frac{c_t^2(\Btheta,\Btheta^*)}{s_t^2(\Btheta^*) + \sigma_n^2(\Btheta^*)}.
\end{align*}
Consider first the case $\sigma_n(\Btheta) = 0$. We see immediately that then $\tau_t^2(\Btheta;\Btheta) = s_t^2(\Btheta)$ so that the intuitive choice $\Btheta^* = \Btheta$ gives the maximal reduction of variance at $\Btheta$. % not necessarily unique

%% tau, constant noise setting
Consider now the situation $\sigma_n(\Btheta) = \sigma_n>0$ for all $\Btheta$, that is, the noise level is constant. % as is typically assumed e.g.~in Bayesian optimisation. 
Then we have $\tau_t^2(\Btheta;\Btheta) = [s_t^2(\Btheta)/(s_t^2(\Btheta) + \sigma_n^2)]s_t^2(\Btheta) < s_t^2(\Btheta)$ whenever $s_t(\Btheta)>0$. 
We can write $c_t(\Btheta,\Btheta^*) = \rho_t(\Btheta,\Btheta^*)s_t(\Btheta)s_t(\Btheta^*)$, where $\rho_t(\cdot,\cdot)$ is the GP correlation function and 
\begin{align*}
    \tau_t^2(\Btheta;\Btheta^*) = \frac{\rho_t^2(\Btheta,\Btheta^*)s_t^2(\Btheta)s_t^2(\Btheta^*)}{s_t^2(\Btheta^*) + \sigma_n^2}.
\end{align*}
If $|\rho_t(\Btheta,\Btheta^*)|=1$ and $s_t^2(\Btheta^*)>s_t^2(\Btheta)>0$ then we obtain
\begin{align*}
    \tau_t^2(\Btheta;\Btheta^*) = \frac{s_t^2(\Btheta)s_t^2(\Btheta^*)}{s_t^2(\Btheta^*) + \sigma_n^2} 
    = \frac{s_t^4(\Btheta)}{s_t^2(\Btheta) + \frac{s_t^2(\Btheta)}{s_t^2(\Btheta^*)}\sigma_n^2} > \tau_t^2(\Btheta;\Btheta)
\end{align*}
which shows that choosing $\Btheta^* = \Btheta$ does not produce the maximal reduction of variance at $\Btheta$ in general.
%
%% tau, general setting
This is obviously also the case when $\sigma_n(\Btheta)>0$ is not constant with respect to $\Btheta$. For example, if $0<s_t(\Btheta)<\infty$ and $\sigma_n(\Btheta)=\infty$ then obviously $\tau_t^2(\Btheta;\Btheta)=0$ but it is possible that $\tau_t^2(\Btheta;\Btheta^*)>0$ for some $\Btheta^*\neq\Btheta$.

\subsubsection{Analysis of \texorpdfstring{$\xi_t$}{xi}-function} \label{sssec:xi}

%% general about xi
The preliminary analysis above showed that an evaluation at $\Btheta$ may not maximally reduce the uncertainty about $f$ at $\Btheta$ (or be a sensible choice at all) unless the evaluation is exact. 
We now similarly analyse $\xi_t^2(\Btheta,\Btheta';\Btheta^*)$ in (\ref{eq:xi}).
To this end, we first observe that we can alternatively write (\ref{eq:xi}) as 
\begin{align}
    \xi_t^2(\Btheta,\Btheta';\Btheta^*) &=  (c_t(\Btheta,\Btheta^*) - c_t(\Btheta',\Btheta^*))[c_t(\Btheta^*,\Btheta^*) + \BLambda^{\!*}]^{-1}(c_t(\Btheta^*,\Btheta) - c_t(\Btheta^*,\Btheta')) \geq 0. \label{eq:gpomegav2}
\end{align}

%% noiseless case: two evaluations enough
First, suppose $\sigma_n(\Btheta)=\sigma_n(\Btheta')=0$. Then $\xi_t^2(\Btheta,\Btheta';[\Btheta,\Btheta']) = s_t^2(\Btheta)+s_t^2(\Btheta')-2c_t(\Btheta,\Btheta') = \sigma_t^2(\Btheta,\Btheta')$ which is easily verified using (\ref{eq:gpomegav2}) and some straightforward computations. That is, when two evaluations are to be used in the non-noisy setting %and can be jointly selected 
so that $b=2$, the intuitive choice $\Btheta^*=[\Btheta,\Btheta']$ produces the maximal reduction of uncertainty. 
Similarly, if we have already evaluated at $\Btheta$ and have one evaluation to use so that $b=1$, then $\xi_t^2(\Btheta,\Btheta';\Btheta^*) = \tau_t^2(\Btheta';\Btheta^*)$ which shows that the intuitive choice $\Btheta^*=\Btheta'$ produces the maximal reduction of uncertainty which in this case equals $s_t^2(\Btheta')=\sigma_t^2(\Btheta,\Btheta')$. % in this non-noisy setting. % this is a bit imprecise; could better discuss the case where the points are perfectly correlated...

Let us now suppose $b=1$. We can then write (\ref{eq:gpomegav2}) as 
\begin{align}
    \xi_t^2(\Btheta,\Btheta';\Btheta^*) &= \frac{[c_t(\Btheta,\Btheta^*) - c_t(\Btheta',\Btheta^*)]^2}{s_t^2(\Btheta^*) + \sigma_n^2(\Btheta^*)}. %\\
    %
    %= \frac{s_t^2(\Btheta^*)[\rho_t(\Btheta,\Btheta^*)s_t(\Btheta) - \rho_t(\Btheta',\Btheta^*)s_t(\Btheta')]^2}{s_t^2(\Btheta^*) + \sigma_n^2(\Btheta^*)}.
    \label{eq:xin}
\end{align}
%
%% noisy case, first iteration (no data yet)
We suppose $\sigma_n(\Btheta) = \sigma_n \geq 0$ for all $\Btheta$. As analysing (\ref{eq:xin}) analytically is hard in general when $b=1$, we limit our attention to the $t=0$ case where the GP posterior equals the GP prior. For simplicity, we also assume a stationary covariance function of the form $k(\Btheta,\Btheta') = \sigmaf^2\kappa(||\Btheta - \Btheta'||_{\BLL})$ where $\sigmaf>0$, $\kappa:\realsp\rightarrow[0,1]$ is a strictly decreasing function with $\kappa(0)=1$ and $\lim_{r\rightarrow\infty}\kappa(r)=0$, $\BLL$ is a positive definite matrix and $||\Btheta||^2_{\BLL}\eqdef\Btheta\T\BLL\Btheta$. 
For example, the choice $\BLL=\diag(l_1^2,\ldots,l_p^2)^{-1}$ where $l_i>0$ are the length-scales and $\kappa(r)=\exp(-r^2/2)$ gives the (anisotropic) squared exponential (SE) covariance function. 
%We have
%
% \begin{align}
%     \rho_0(\Btheta,\Btheta') = \frac{c_0(\Btheta,\Btheta')}{s_0(\Btheta)s_0(\Btheta')} 
%     %
%     =  \frac{k(\Btheta,\Btheta')}{\sqrt{k(\Btheta,\Btheta)k(\Btheta',\Btheta')}}
%     %
%     = \kappa(||\Btheta - \Btheta'||_{\BLL}).
%     \label{eq:rho0}
% \end{align}
%
Then (\ref{eq:xin}) can be written as 
\begin{align}
    \xi_t^2(\Btheta,\Btheta';\Btheta^*) &= c[\kappa(||\Btheta - \Btheta^*||_{\BLL}) - \kappa(||\Btheta' - \Btheta^*||_{\BLL})]^2, 
    \label{eq:xisimple}
\end{align}
where $c>0$ does not depend on $\Btheta^*$. 
%% "useless" points
Now, (\ref{eq:xisimple}) shows that if $\Btheta^*$ is far enough (as measured using the norm $||\cdot||_{\BLL}$) from both $\Btheta$ and $\Btheta'$, then $\xi_t^2(\Btheta,\Btheta';\Btheta^*)\approx 0$ so that the uncertainty associated with the \mh{} accept/reject decision will decrease only little. 
Surprisingly, this is also the case if $\kappa(||\Btheta - \Btheta^*||_{\BLL}) = \kappa(||\Btheta' - \Btheta^*||_{\BLL})$, which is equivalent to $||\Btheta - \Btheta^*||_{\BLL}=||\Btheta' - \Btheta^*||_{\BLL}$. That is, points $\Btheta^*$ that are equally far from $\Btheta$ and $\Btheta'$ are uninformative in this case. 

%% "useful" points
Under the above assumptions, we can reason that if $\Btheta$ and $\Btheta'$ are very far from each other, then choosing $\Btheta^*=\Btheta$ or $\Btheta^*=\Btheta'$ will approximately maximise $\xi_t^2(\Btheta,\Btheta';\Btheta^*)$. Suppose next that we have SE covariance function. 
Then $\nabla_{\!\Btheta^*}\xi_t^2(\Btheta,\Btheta';\Btheta^*)$ is proportional to
\begin{align*}
%\nabla_{\Btheta^*}\xi_t^2(\Btheta,\Btheta';\Btheta^*) 
%&= 2c\!
\left[\e^{-\half||\Btheta - \Btheta^*||_{\BLL}^2} - \e^{-\half||\Btheta' - \Btheta^*||_{\BLL}^2}\right] \!\left[\e^{-\half||\Btheta - \Btheta^*||_{\BLL}^2}\BLL(\Btheta-\Btheta^*) - \e^{-\half||\Btheta' - \Btheta^*||_{\BLL}^2}\BLL(\Btheta'-\Btheta^*)\right]\!. 
%\label{eq:nablaxi}
\end{align*}
This formula shows that the gradient $\nabla_{\!\Btheta^*}\xi_t^2(\Btheta,\Btheta';\Btheta^*)$ is nonzero both at $\Btheta^*=\Btheta$ and $\Btheta^*=\Btheta'$ because $\Btheta\neq\Btheta'$ so these points are not local or global optima (except possibly in the special case where they belong to the boundary of the set $\Theta$). 
%
%% xi, more general case
%If the covariance function is non-stationary or if $\sigma_n^2(\Btheta)$ is not constant, the situation is more complicated but our analysis above shows that the optimal point is not $\Btheta$ or $\Btheta'$ in general. 
Our numerical experiments in Section \ref{sec:exp} further demonstrate that it matters in practice whether we optimise $\xi_t^2(\Btheta,\Btheta';\Btheta^*)$ globally or over $\Btheta^*\in\{\Btheta,\Btheta'\}$.

\subsection{Worst case upper bounds for the repeated, noisy evaluations} \label{subsec:nrevals}

% general
We provide some analysis on the number of repeated, noisy log-likelihood evaluations, denoted by $\n$ in this section, needed at an individual iteration of Algorithm \ref{alg:gpmh} to satisfy the condition $\conderr_{\n,u,\hatgamma}(\Btheta,\Btheta')\leq\epsi$ or $\unconderr_{\n,\hatgamma}(\Btheta,\Btheta')\leq\epsi$. % to render the (un)conditional error smaller than $\epsi$. % corresponding to a transition from $\Btheta$ to $\Btheta'$. 
%We only derive certain worst case bounds as more general analysis seems difficult. 
%It would be advantageous to know in advance how many likelihood evaluations Algorithm \ref{alg:gpmh} requires given some error tolerance $\epsi$ and the number of \mh{} samples $i_{\text{MH}}$. 
%
We first recognise some special cases: If $\epsi \geq 1/2$, then obviously no log-likelihood evaluations are needed. On the other hand, if $\epsi=0$ and $\pi(\Btheta')>0$ then the log-likelihood must be known exactly at both $\Btheta$ and $\Btheta'$ which requires arbitrarily many evaluations in the noisy case. 
As seen in Figure \ref{fig:cond_vs_uncond_errors}, if $\mu_t(\Btheta,\Btheta') = \log(u)$ then $\conderr_{t,u,\hatgamma}(\Btheta,\Btheta')=\Phi(0)=1/2$ even if $\sigma_t(\Btheta,\Btheta')$ is negligible yet nonzero. %i.e.~even if $f$ is known accurately at both $\Btheta$ and $\Btheta'$. 
This situation however happens with vanishing probability. %, we hence cannot obtain a useful deterministic bound for $\conderr_{t,u,\hatgamma}(\Btheta,\Btheta')$. %On the other hand, if $\mu_t(\Btheta,\Btheta')>0$ which means that the proposed point $\Btheta'$ has higher posterior value than the current one $\Btheta$ (and provided that $q$ is also symmetric), only a finite number of evaluations is needed with any $u\in[0,1]$. 
%
%We obtain the following worst case bounds: 
%
\begin{proposition} \label{prop:bound1}
Suppose that the GP prior model of Section \ref{subsec:gp_model} but with a stationary covariance function as in Section \ref{sssec:xi} holds and that $\n\geq 0$ evaluation locations are then chosen as to minimise either the conditional error (\ref{eq:conderr_bound1}) or unconditional error (\ref{eq:unconderr_bound1}) computed between distinct parameters $\Btheta, \Btheta'\in\Theta$ in an optimal fashion of Section \ref{sec:acq}. 
Suppose $0<\epsi<1/2$ and $\sigma_n(\Btheta)>0$ for all $\Btheta\in\Theta$. 
Then\footnote{The fact $c_{\n} = \min\{2\sigmaf,2\sqrt{2}\bar{\sigma}_n/\sqrt{\n}\} \leq 2\sqrt{2}\bar{\sigma}_n/\sqrt{\n}$ could be used to slightly simplify the bounds when $\n$ is even and $\n\geq2$.} it holds that 
\begin{align}
    \prob(\conderr_{\n,u,\hatgamma}(\Btheta,\Btheta') \geq \epsi) &\leq 1-\e^{2\Phi^{-1}(\epsi)c_{\n}}, \label{eq:conderr_bound1} \\
    \unconderr_{\n,\hatgamma}(\Btheta,\Btheta') &\leq \max_{\mu\leq 0}\left\{ \Phi\left(\frac{\mu}{c_{\n}}\right) + \e^{\mu+{c_{\n}^2}/{2}}\left( \Phi\left( -\frac{\mu+c_{\n}^2}{c_{\n}} \right) - 2\Phi(-c_{\n})\right)\!\right\}, \label{eq:unconderr_bound1}
\end{align}
where $\prob(\cdot)$ is with respect to $u\sim\Unif([0,1])$, $c_{\n}\eqdef 2\min\{\sigmaf,\bar{\sigma}_n/\sqrt{\lfloor \n/2\rfloor}\}$ and $\bar{\sigma}_n \eqdef \max\{ \sigma_n(\Btheta), \sigma_n(\Btheta') \}$. 
\end{proposition}
%

% worst case bound etc.
Importantly, Proposition \ref{prop:bound1} concerns the worst case situation where no log-likelihood evaluations are yet obtained (that is, the GP posterior equals the GP prior) and where the evaluations are always such that $\mu_{\n}(\Btheta,\Btheta')=\Phi^{-1}(\epsi)\sigma_{\n}(\Btheta,\Btheta')$ as seen in the proof of (\ref{eq:conderr_bound1}) in \appe{} \ref{appsec:proofs2}. Furthermore, Proposition \ref{prop:bound1} is for the batch case where $\n$ evaluations are selected simultaneously (but we hypothesise similar result holds in the sequential case). 
We observe numerically that the maximum of (\ref{eq:unconderr_bound1}) is typically found in $\mu\in(-0.7,0)$. %However, such values of $\mu_t$ do not occur on each iteration of Algorithm \ref{alg:gpmh} but only occasionally. 
% remarks
We see that $\prob(\conderr_{\n,u,\hatgamma}\geq \epsi)\rightarrow 0$ and $\unconderr_{\n,\hatgamma}(\Btheta,\Btheta')\rightarrow 0$ as $\n\rightarrow\infty$. 

% "average" bounds
%We next analyse the typical values of $\mu_t$ encountered during Algorithm \ref{alg:gpmh} to gain some insight on a typical number of simulations $n$ needed to make the (un)conditional error smaller than $\epsi$. 
%
% Suppose Normal target
We next develop revised bounds that are more representative of practice though they still assume a certain worst case situation. We consider a Gaussian target density $\Normal_p(\Bmu,\BSigma)$ where $\BSigma\in\reals^{p\times p}$ is positive definite and w.l.o.g.~we set $\Bmu=\Bzeros$. We consider a Gaussian proposal $q(\Btheta'\cond\Btheta)=\Normal_p(\Btheta'\cond\Btheta,s^2\BSigma)$ where $s>0$ is a fixed scaling parameter. 
We suppose that the artificial scenario, where $\Btheta$ is first drawn from the target and $\Btheta'$ is then drawn from the proposal so that $\Btheta\sim\Normal_p(\Bzeros,\BSigma), \Btheta'\cond\Btheta \sim \Normal_p(\Btheta,s^2\BSigma)$, represents a typical \mh{} iteration. Then we obtain %(see \appe{} \ref{appesec:proofs} for the derivation) 
\begin{align}
    \mean_{\Btheta,\Btheta'}(f(\Btheta')-f(\Btheta)) = -\half ps^2, \quad \Var_{\Btheta,\Btheta'}(f(\Btheta')-f(\Btheta)) = \half ps^2(s^2+2),
    \label{eq:mean_var_fdiff}
\end{align}
where $f$ is the Gaussian log-target density and where the expectation and variance are wrt.~the randomness due to sampling $\Btheta$ and $\Btheta'$. 
In particular, if we use the common recommended choice $s^2=2.4^2/p$ \citep[Chapter~12.2]{Gelman2013}, we obtain $\mean_{\Btheta,\Btheta'}(f(\Btheta')-f(\Btheta)) = -2.88$ and a relatively large standard deviation $\stdev_{\Btheta,\Btheta'}(f(\Btheta')-f(\Btheta)) = 2.4\sqrt{2.88/p+1}$. % which suggests that most transitions are such that ... 

The bounds of Proposition \ref{prop:bound1} are then revised by using the distribution of $\mu_{\n}(\Btheta,\Btheta')$, which we assume to be the same as that of $f(\Btheta')-f(\Btheta)$ described above, in place of its worst case choice. These computations are not completely analytic but the resulting bounds are easily evaluated numerically. The details are postponed to \appe{} \ref{appsec:proofs2}.

% figure of the bounds
Figure \ref{fig:bounds} demonstrates the bounds. According to the bounds of Proposition \ref{prop:bound1} (dashed lines in (a) and (b)) hundreds of repeated evaluations are needed if $\epsi$ is small. As expected, the revised bounds (solid lines in (a) and (c)), based on the assumed distribution of $\mu_t$ in the Gaussian case, are tighter. %In particular, if we use unconditional error and $\epsi=0.2$, maximum of $40$ evaluations is enough. 
These results indicate a potential shortcoming of \gpmh{}: Many repeated evaluations may be needed at least in some individual worst case iterations of the \gpmh{} algorithm. %However, the fact that the information brought by the evaluations is reused in the later iterations via the GP model is not acknowledged in the theoretical analysis so these results however do not necessarily imply that \gpmh{} would not be sample-efficient. 
%It is clear that $\epsi$ should not be set too small. Otherwise the number of required likelihood evaluations can become too large for \gpmh{} to be useful. 

\begin{figure}[hbtp] %% bounds
\centering
\begin{subfigure}{0.32\textwidth}
\includegraphics[width=\textwidth]{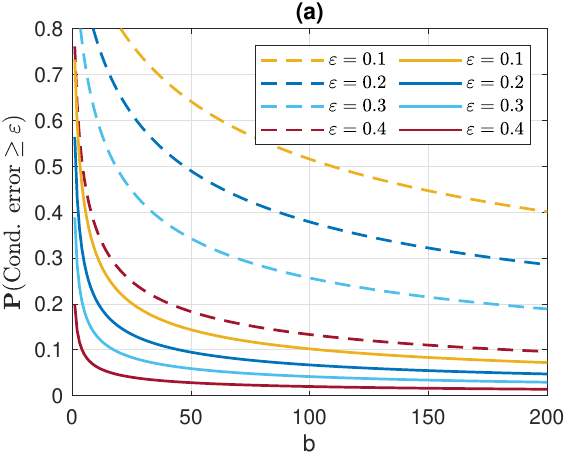}
\end{subfigure}
%\hspace{0.05cm}
\begin{subfigure}{0.32\textwidth}
\includegraphics[width=\textwidth]{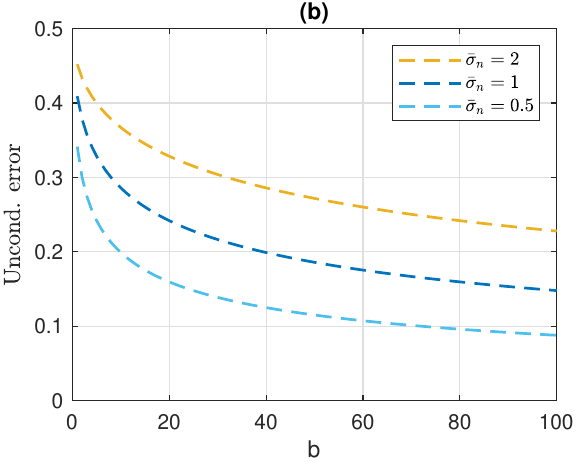}
\end{subfigure}
%\hspace{0.05cm}
\begin{subfigure}{0.32\textwidth}
\includegraphics[width=\textwidth]{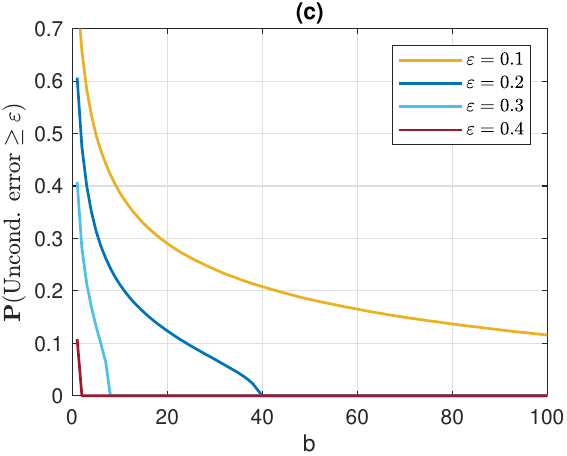}
\end{subfigure}
\caption{(a) Dashed lines show the worst case bound (\ref{eq:conderr_bound1}) for the probability that the conditional error is larger than $\epsi$. Solid lines show the corresponding bound based on the assumed distribution of $\mu_{\n}$ in the Gaussian case (\ref{eq:extraboundconderr}) with $p=5$ and $s^2=2.4^2/p$. In both cases $\bar{\sigma}_n=1$. (b) The worst case bound (\ref{eq:unconderr_bound1}) for the unconditional error for three different values of $\bar{\sigma}_n$. (c) The upper bounds for the probability that the unconditional error is larger than $\epsi$ in the assumed Gaussian case (\ref{eq:extraboundunconderr}).} \label{fig:bounds}
\end{figure}

\subsection{Proofs and mathematical derivations of Section \ref{subsec:nrevals}} \label{appsec:proofs2}

%\begin{proof}[Proof of Proposition \ref{prop:bound1}]
\begin{proof}[Proposition \ref{prop:bound1}]
We denote $\lambda_{\n} \eqdef -\Phi^{-1}(\epsi)\sigma_{\n}(\Btheta,\Btheta')$. Note that $\lambda_{\n}>0$ since $\epsi\in(0,1/2)$. We then obtain 
\begin{align}
\begin{split}
    \prob(\conderr_{\n,u,\hatgamma} \geq \epsi) &= \prob\left( \Phi\left( -\frac{|\mu_{\n}(\Btheta,\Btheta') - \logu|}{\sigma_{\n}(\Btheta,\Btheta')} \right) \geq \epsi \right) \\
    &= \prob\left( |\mu_{\n}(\Btheta,\Btheta') - \logu| \leq -\Phi^{-1}(\epsi)\sigma_{\n}(\Btheta,\Btheta') \right) \\
    %
    %&= \prob\left( (\mu_{\n}(\Btheta,\Btheta') - \logu)^2 \leq [-\Phi^{-1}(\epsi)]^2\sigma^2_n(\Btheta,\Btheta') \right) \\
    %
    %&= 1 - \prob\left( \mu_{\n}(\Btheta,\Btheta') - \logu \geq \lambda_{\n} \textnormal{ or } \logu - \mu_{\n}(\Btheta,\Btheta') \geq \lambda_{\n} \right) \\
    &= 1 - \prob\left( \{\mu_{\n}(\Btheta,\Btheta') - \logu \geq \lambda_{\n} \} \cup \{ \logu - \mu_{\n}(\Btheta,\Btheta') \geq \lambda_{\n} \} \right) \\
    &= 1 - \prob\left( \mu_{\n}(\Btheta,\Btheta') - \logu \geq \lambda_{\n} \right) - \prob\left(\logu - \mu_{\n}(\Btheta,\Btheta') \geq \lambda_{\n} \right) \\ 
    &= 1 - \prob\left( u \leq \e^{\mu_{\n}(\Btheta,\Btheta') - \lambda_{\n}} \right) - \prob\left( u > \e^{\mu_{\n}(\Btheta,\Btheta') + \lambda_{\n}} \right) \\
    &= 1 - \min\{ \e^{\mu_{\n}(\Btheta,\Btheta') - \lambda_{\n}} ,1\} - (1-\min\{ \e^{\mu_{\n}(\Btheta,\Btheta') + \lambda_{\n}} ,1\}) \\
    &= \max\{ 1 - \e^{\mu_{\n}(\Btheta,\Btheta') - \lambda_{\n}} ,0\} + \min\{ \e^{\mu_{\n}(\Btheta,\Btheta') + \lambda_{\n}}-1 ,0\}. \label{eq:conderrmaxmin}
\end{split}
\end{align}
%
% % case 1: \mu_t = 0 ACTUALLY REDUNDANT TO CONSIDER SEPARATELY
% Consider the case $\mu_{\n}(\Btheta,\Btheta')=0$. Then clearly
% %
% \begin{align*}
%     \prob(\conderr_{\n,u,\hatgamma} \geq \epsi) 
%     &= 1 - \e^{-\lambda_{\n}}. 
% \end{align*}
% %

% case 2: \mu_t > 0
%If $\mu_{\n}(\Btheta,\Btheta')>0$, then simple computation shows that
If $\mu_{\n}(\Btheta,\Btheta')\geq 0$, then simple computation shows that
\begin{align*}
    \prob(\conderr_{\n,u,\hatgamma} \geq \epsi) &= \max\{ 1 - \e^{\mu_{\n}(\Btheta,\Btheta') - \lambda_{\n}} ,0\} \\
    &=\begin{cases} 0 &\textnormal{if } \lambda_{\n} \leq \mu_{\n}(\Btheta,\Btheta'), \\
    1-\e^{\mu_{\n}(\Btheta,\Btheta') - \lambda_{\n}} &\textnormal{otherwise}. \end{cases}
\end{align*}
Since the function $\mu_{\n}\mapsto 1-\e^{\mu_{\n}-\lambda_{\n}}$ is decreasing for $\mu_{\n} \geq 0$ and because $\prob(\conderr_{\n,u,\hatgamma} \geq \epsi)=0$ for $\lambda_{\n} \leq \mu_{\n}(\Btheta,\Btheta')$, we see that $\prob(\conderr_{\n,u,\hatgamma} \geq \epsi) \leq 1 - \e^{-\lambda_{\n}}$ for any $\mu_{\n}(\Btheta,\Btheta') \geq 0$.

% case 3: \mu_t < 0
If $\mu_{\n}(\Btheta,\Btheta')<0$, then we see that 
\begin{align*}
    \prob(\conderr_{\n,u,\hatgamma} \geq \epsi) &= 1 - \e^{\mu_{\n}(\Btheta,\Btheta') - \lambda_{\n}}  + \min\{ \e^{\mu_{\n}(\Btheta,\Btheta') + \lambda_{\n}}-1 ,0\} \\
    &= \begin{cases}1-\e^{\mu_{\n}(\Btheta,\Btheta') - \lambda_{\n}} &\textnormal{if } -\lambda_{\n} \leq \mu_{\n}(\Btheta,\Btheta'), \\
    \e^{\mu_{\n}(\Btheta,\Btheta') + \lambda_{\n}} - \e^{\mu_{\n}(\Btheta,\Btheta') - \lambda_{\n}} &\textnormal{otherwise}. \end{cases} 
\end{align*}
Since the function $\mu_{\n}\mapsto 1-\e^{\mu_{\n}-\lambda_{\n}}$ is decreasing, its maximum in $-\lambda_{\n}\leq\mu_{\n}<0$ occurs with $\mu_{\n}=-\lambda_{\n}$. As $\mu_{\n}\mapsto \e^{\mu_{\n}-\lambda_{\n}}\e^{\mu_{\n} + \lambda_{\n}} - \e^{\mu_{\n} - \lambda_{\n}} = 2\e^{\mu_{\n}}\sinh(\lambda_{\n})$ is increasing in $\mu_{\n}<-\lambda_{\n}$, the choice $\mu_{\n}=-\lambda_{\n}$ gives an upper bound also when $\mu_{\n}<-\lambda_{\n}$. 
We have thus shown
\begin{align}
    \prob(\conderr_{\n,u,\hatgamma} \geq \epsi) \leq 1-\e^{-2\lambda_{\n}} = 1-\e^{2\Phi^{-1}(\epsi)\sigma_{\n}(\Btheta,\Btheta')}, \label{eq:firstpart}
\end{align}
which also works when $\mu_{\n}(\Btheta,\Btheta') \geq 0$. 

Next we obtain
\begin{align}
\begin{split}
    \sigma_{\n}(\Btheta,\Btheta') &= \sqrt{s_{\n}^2(\Btheta)+s_{\n}^2(\Btheta')-2c_{\n}(\Btheta,\Btheta'))} \\
    &\leq \sqrt{s_{\n}^2(\Btheta)+s_{\n}^2(\Btheta') + 2s_{\n}(\Btheta)s_{\n}(\Btheta')} \\
    &= s_{\n}(\Btheta)+s_{\n}(\Btheta'). 
\end{split} \label{eq:ub_s1s2}
\end{align}
We further bound (\ref{eq:ub_s1s2}) in terms of the ${\n}$ evaluations. We note that since the optimal method for minimising either the conditional or unconditional error also minimises $\sigma_{\n}(\Btheta,\Btheta')$, any method for selecting the ${\n}$ evaluation locations will give an upper bound. 

Suppose that $\Btheta_{1:t}$ consists of $t\geq 1$ evaluation locations that all are $\Btheta$. We denote $\onevector \eqdef [1,\ldots,1]\T\in\reals^t$ and $\Id$ is a $t\times t$ identity matrix. % and $\sigma_n = \sigma_n(\Btheta)$. 
Then
\begin{align}
\begin{split}
    s_t^2(\Btheta) &= k(\Btheta,\Btheta) - k(\Btheta,\Btheta_{1:t})[k(\Btheta_{1:t},\Btheta_{1:t}) + \sigma_n^2(\Btheta)\Id]^{-1}k(\Btheta_{1:t},\Btheta) \\
    &= \sigmaf^2 - \sigmaf^2\onevector\T\Big[\onevector\onevector\T + \frac{\sigma_n^2(\Btheta)}{\sigmaf^2}\Id\Big]^{-1}\onevector. 
\end{split}\label{eq:smf}
\end{align}
We use Sherman–Morrison formula to compute $\Big[\onevector\onevector\T + \frac{\sigma_n^2(\Btheta)}{\sigmaf^2}\Id\Big]^{-1} = \frac{\sigmaf^2}{\sigma_n^2(\Btheta)}\Id - \frac{(\sigmaf^2/\sigma_n^2(\Btheta))^2}{1+t\sigmaf^2/\sigma_n^2(\Btheta)}\onevector\onevector\T$. After plugging this formula to (\ref{eq:smf}) and some straightforward calculations, we see that
\begin{align}
    s_t(\Btheta) &= \frac{\sigmaf}{\sqrt{1+{t\sigmaf^2}/{\sigma_n^2(\Btheta)}}}
    \leq \min\left\{\sigmaf, \frac{\sigma_n(\Btheta)}{\sqrt{t}}\right\}, \label{eq:bdx}
    %
    %\leq \frac{\sigma_n(\Btheta)}{\sqrt{n}}. 
\end{align}
which also works when $t=0$ because then ${\sigma_n(\Btheta)}/{\sqrt{t}}=\infty$ so that $s_0(\Btheta)=\sigmaf$. 

Suppose now that we have $m\geq0$ evaluations at $\Btheta$ and $m'\geq0$ evaluations at $\Btheta'$ such that $m+m'\leq {\n}$. Then
\begin{align*}
    s_{\n}(\Btheta)+s_{\n}(\Btheta') &\leq s_m(\Btheta)+s_{m'}(\Btheta') \\
    \overset{\eqref{eq:bdx}}&{\leq} \min\left\{\sigmaf,\frac{\sigma_n(\Btheta)}{\sqrt{m}}\right\} + \min\left\{\sigmaf,\frac{\sigma_n(\Btheta')}{\sqrt{m'}}\right\} \\
    &\leq \min\left\{2\sigmaf, \left( \frac{1}{\sqrt{m}} + \frac{1}{\sqrt{m'}} \right)\bar{\sigma}_n\right\},
\end{align*}
where the first inequality follows from the fact that the GP variance function is non-increasing as a set function of data, that is, adding more evaluation locations to data cannot increase the GP variance.  
If we choose in particular ${m}={m'}=\lfloor {\n}/2\rfloor$, we obtain
\begin{align}
    s_{\n}(\Btheta)+s_{\n}(\Btheta') \leq  2\min\left\{\sigmaf,\bar{\sigma}_n/\sqrt{\lfloor {\n}/2\rfloor}\right\} = c_{\n}. \label{eq:lastpart}
\end{align}
Combining (\ref{eq:firstpart}), (\ref{eq:ub_s1s2}) and (\ref{eq:lastpart}) produces the final bound (\ref{eq:conderr_bound1}).

% UNCOND ERROR
To prove (\ref{eq:unconderr_bound1}), we first notice that $\unconderr_{\n,\hatgamma}(\Btheta,\Btheta')$ is decreasing function with respect to $\mu_{\n}$ %(as earlier, we simplify the formulas by writing $\mu_{\n}$ for $\mu_{\n}(\Btheta,\Btheta')$ and similarly for $\sigma_{\n}$) 
when $\mu_{\n}\geq0$ so that the choice $\mu_{\n}=0$ maximises $\unconderr_{\n,\hatgamma}(\Btheta,\Btheta')$ in $\mu_{\n} \geq 0$. 

Suppose now $\mu_{\n} \leq 0$. For this case we already derived the formula:
\begin{align*}
    \unconderr_{\n,\hatgamma}(\Btheta,\Btheta') &= \Phi\left(\frac{\mu_{\n}}{\sigma_{\n}}\right) + \e^{\mu_{\n}+{\sigma_{\n}^2}/{2}}\left( \Phi\left( -\frac{\mu_{\n}+\sigma_{\n}^2}{\sigma_{\n}} \right) - 2\Phi(-\sigma_{\n})\right). %\label{eq:unconderrsec}
\end{align*}
We maximise it with respect to $\mu_{\n}$ and use the inequality for $\sigma_{\n}$ derived in the first part of this proof and the fact that $\sigma_{\n}\mapsto \unconderr_{\n,\hatgamma}(\Btheta,\Btheta')$ is strictly increasing function for $\sigma_{\n}>0$ (which we see directly from (\ref{eq:conderr_gp})) to obtain
\begin{align*}
    \unconderr_{\n,\hatgamma}(\Btheta,\Btheta') &\leq \max_{\mu\leq 0}\left\{ \Phi\left(\frac{\mu}{\sigma_{\n}}\right) + \e^{\mu+{\sigma_{\n}^2}/{2}}\left( \Phi\left( -\frac{\mu+\sigma_{\n}^2}{\sigma_{\n}} \right) - 2\Phi(-\sigma_{\n})\right)\!\right\} \\
    &\leq \max_{\mu\leq 0}\left\{ \Phi\left(\frac{\mu}{c_{\n}}\right) + \e^{\mu+{c_{\n}^2}/{2}}\left( \Phi\left( -\frac{\mu+c_{\n}^2}{c_{\n}} \right) - 2\Phi(-c_{\n})\right)\!\right\},
\end{align*}
which is the desired bound. 
\end{proof}

% Analysis - approx. bound
\noindent\textbf{Justification for (\ref{eq:mean_var_fdiff}).}
%We next justify equations (\ref{eq:mean_var_fdiff}). 
We write
\begin{align}
    f(\Btheta')-f(\Btheta) &= \log\Normal(\Btheta'\cond \Bzeros, \BSigma) - \log\Normal(\Btheta\cond \Bzeros, \BSigma) = \half(\Btheta\T\BSigma^{-1}\Btheta - \Btheta\Tpr\BSigma^{-1}\Btheta').
\end{align}
Since $\BSigma$ is positive definite, we have Cholesky factorisation $\BSigma=\BL\BL\T$ so that $\BSigma^{-1}=\BL\Tinv\BL^{-1}$, where $\BL\Tinv\eqdef(\BL^{-1})\T=(\BL\T)^{-1}$. Consider random vectors $\Bpsi = \BL^{-1}\Btheta$ and $\Bpsi' = \BL^{-1}\Btheta'$. Clearly $\Btheta\T\BSigma^{-1}\Btheta=\Bpsi\T\Bpsi$, $\Btheta\Tpr\BSigma^{-1}\Btheta'=\Bpsi\Tpr\Bpsi'$ and $[\Bpsi\T, \Bpsi\Tpr]\T$ is Gaussian distributed. We compute $\mean(\Bpsi)=\BL^{-1}\mean(\Btheta)=\Bzeros$ and $\Var(\Bpsi)=\BL^{-1}\Var(\Btheta)\BL\Tinv=\BL^{-1}\BSigma\BL\Tinv=\Id$. 
we also have $\mean(\Bpsi')=\mean(\mean(\BL^{-1}\Btheta'\cond\Btheta))=\mean(\BL^{-1}\mean(\Btheta'\cond\Btheta))=\mean(\BL^{-1}\Btheta)=\Bzeros$ and 
\begin{align*}
    \Var(\Bpsi') &= \mean(\Var(\BL^{-1}\Btheta'\cond\Btheta)) + \Var(\mean(\BL^{-1}\Btheta'\cond\Btheta)) \\
    &= \mean(\BL^{-1}\Var(\Btheta'\cond\Btheta)\BL\Tinv) + \Var(\BL^{-1}\mean(\Btheta'\cond\Btheta)) \\
    &= \mean(s^2\BL^{-1}\BSigma\BL\Tinv) + \Var(\BL^{-1}\Btheta) \\
    &= s^2\Id + \BL^{-1}\BSigma\BL\Tinv \\
    &= (s^2+1)\Id.
\end{align*}
Since we can write $\Btheta'=\Btheta+\Br$, where $\Br\sim\Normal(\Bzeros,s^2\BSigma)$, it follows that 
\begin{align*}
    \cov(\Bpsi,\Bpsi') &= \cov(\BL^{-1}\Btheta, \BL^{-1}(\Btheta+\Br)) \\
    &= \cov(\BL^{-1}\Btheta,\BL^{-1}\Btheta) + \cov(\BL^{-1}\Btheta, \BL^{-1}\Br) \\
    &= \Var(\BL^{-1}\Btheta) + \BL^{-1}\cov(\Btheta,\Br)\BL\Tinv \\
    &= \Id.
\end{align*}
We have thus shown
\begin{align}
\begin{split}
    f(\Btheta')-f(\Btheta) &= \half(\Bpsi\T\Bpsi - \Bpsi\Tpr\Bpsi') = \half\sum_{i=1}^p (\psi_i^2 - \psi_i'{}^2), \\
    \begin{bmatrix}
    \Bpsi \\ \Bpsi'
    \end{bmatrix}
    &\sim \Normal_{2p}\left( \begin{bmatrix} \Bzeros \\ \Bzeros \end{bmatrix}, \begin{bmatrix} \Id & \Id \\ \Id & (s^2+1)\Id \end{bmatrix} \right).
\end{split} \label{eq:psijoint}
\end{align}
This shows that the distribution of $f(\Btheta')-f(\Btheta)$ does not depend on $\BSigma$ and is approximately Gaussian by the central limit theorem (which applies because the random variables $\psi_i^2 - \psi_i'{}^2, i=1,\ldots,p$ are independent and have finite variance by (\ref{eq:psijoint})) when $p$ is large. 

The expectation and variance are now obtained as\footnote{One could also write $\Bpsi\T\Bpsi - \Bpsi\Tpr\Bpsi' = [\Bpsi\T,\Bpsi\Tpr]\begin{bmatrix} \Id & \zeromatrix \\ \zeromatrix & -\Id \end{bmatrix} \begin{bmatrix} \Bpsi \\ \Bpsi' \end{bmatrix}$ and then use known formulas for computing the expectation and variance of this quadratic form to obtain the same results.}:
\begin{align*}
    \mean_{\Btheta,\Btheta'}(f(\Btheta')-f(\Btheta)) &= \half\sum_{i=1}^p \left(\mean(\psi_i^2)-\mean(\psi_i'{}^2)\right) = -\half ps^2, \\
    \Var_{\Btheta,\Btheta'}(f(\Btheta')-f(\Btheta)) &= \frac{1}{4}\sum_{i=1}^p \left(\Var(\psi_i^2) + \Var(\psi_i'{}^2) - 2 \cov(\psi_i^2,\psi_i'{}^2)\right) = \half ps^2(s^2+2), 
\end{align*}
\sloppy
where we have additionally used the facts $\Var(\psi_i^2) = \mean(\psi_i^4)-\mean(\psi_i^2)^2$, $\cov(\psi_i^2,\psi_i'{}^2) = \mean(\psi_i^2\psi_i'{}^2)-\mean(\psi_i^2)\mean(\psi_i'{}^2)$ and well-known formulas for the moments of zero-mean Gaussian distribution.

%% bounds - simulation
\noindent\textbf{Derivation of the revised upper bounds.} 
Proposition \ref{prop:bound1} shows the worst case upper bounds with respect to $\mu_{\n}$. 
Here we derive revised bounds where we instead use a specific distribution for $\mu_{\n}$ under the Gaussian target and proposal assumption. That is, we assume $\mu_{\n}$ follows (for each possible $\n$) the same distribution as $f(\Btheta')-f(\Btheta)$ in (\ref{eq:psijoint}). %The resulting bounds cannot be computed completely analytically. 
Under the assumptions of Proposition \ref{prop:bound1} and using (\ref{eq:conderrmaxmin}) we obtain 
\begin{align}
\begin{split}
    &\prob(\conderr_{\n,u,\hatgamma}(\Btheta,\Btheta') \geq \epsi) \\
    &\myquad= \int_{\reals} \prob(\conderr_{\n,u,\hatgamma}(\Btheta,\Btheta') \geq \epsi\cond\mu_{\n}) \pi(\mu_{\n}) \ud \mu_{\n} \\
    &\myquad= \int_{\reals} \left(\max\{ 1 - \e^{\mu_{\n} - \lambda_{\n}} ,0\} + \min\{ \e^{\mu_{\n} + \lambda_{\n}}-1 ,0\}\right) \pi(\mu_{\n}) \ud \mu_{\n} \\
    &\myquad\leq \int_{\reals} \left(\max\{ 1 - \e^{\mu_{\n} + \Phi^{-1}(\epsi)c_{\n}} ,0\} + \min\{ \e^{\mu_{\n} - \Phi^{-1}(\epsi)c_{\n}}-1 ,0\}\right) \pi(\mu_{\n}) \ud \mu_{\n} \\
    &\myquad\approx \frac{1}{r}\sum_{i=1}^r \left(\max\{ 1 - \e^{\mu_{\n}^{(i)} + \Phi^{-1}(\epsi)c_{\n}} ,0\} + \min\{ \e^{\mu_{\n}^{(i)} - \Phi^{-1}(\epsi)c_{\n}}-1 ,0\}\right), %\quad \mu_{\n}^{(i)}\simiid\pi(\mu_{\n})
\end{split} \label{eq:extraboundconderr}
\end{align}
where $\mu_{\n}^{(i)}\simiid\pi(\mu_{\n}), i=1,\ldots,r$ (for any $\n$). Samples from $\pi(\mu_{\n})$ can be obtained by first drawing $[\psi_j^{(i)},\psi_j'^{(i)}]\T\simiid\Normal_2\Big(\begin{bmatrix}0\\0\end{bmatrix},\begin{bmatrix}1& 1\\1&s^2+1\end{bmatrix}\Big)$ for $j=1,\ldots,p$ and then computing $\mu_{\n}^{(i)}=\sum_{i=1}^p(\psi_j^{(i)2} - \psi_j'^{(i)2})/2$.

We write the unconditional error as $E(\mu_{\n},\sigma_{\n})$ when it is considered as a function of $\mu_{\n}$ and $\sigma_{\n}$. We then obtain
\begin{align}
\begin{split}
    \prob(\unconderr_{\n,\hatgamma}(\Btheta,\Btheta') \geq \epsi) =
    \prob(E(\mu_{\n},\sigma_{\n}) \geq \epsi) %\\
    &\leq \prob(E(\mu_{\n},c_{\n}) \geq \epsi) %\\
    %
    %&= \mean\indic_{E(\mu_{\n},c_{\n}) \geq \epsi} \\
    %
    \approx \frac{1}{r}\sum_{i=1}^r \indic_{E(\mu_{\n}^{(i)},c_{\n}) \geq \epsi},
\end{split} \label{eq:extraboundunconderr}
\end{align}
where $c_{\n}$ is as in Proposition \ref{prop:bound1} and where $\mu_{\n}^{(i)}\simiid\pi(\mu_{\n}), i=1,\ldots,r$ can be simulated as above.

\section{Additional illustrations} \label{appsec:add_viz}

Figures \ref{fig:loglik_gp2} and \ref{fig:loglik_gp3} show how a GP prior with a non-zero mean function and an additional evaluation near the right boundary of the parameter space, respectively, produce more intuitive estimates of the SL posterior in the illustrative example of Section \ref{subsubsec:post_estimator}. In the former case large uncertainty of the likelihood near the right boundary however remains. 

\begin{figure}[hbt] % loglik modelling/estimator demo - non-zero mean function
\centering
\includegraphics[width=0.8\textwidth]{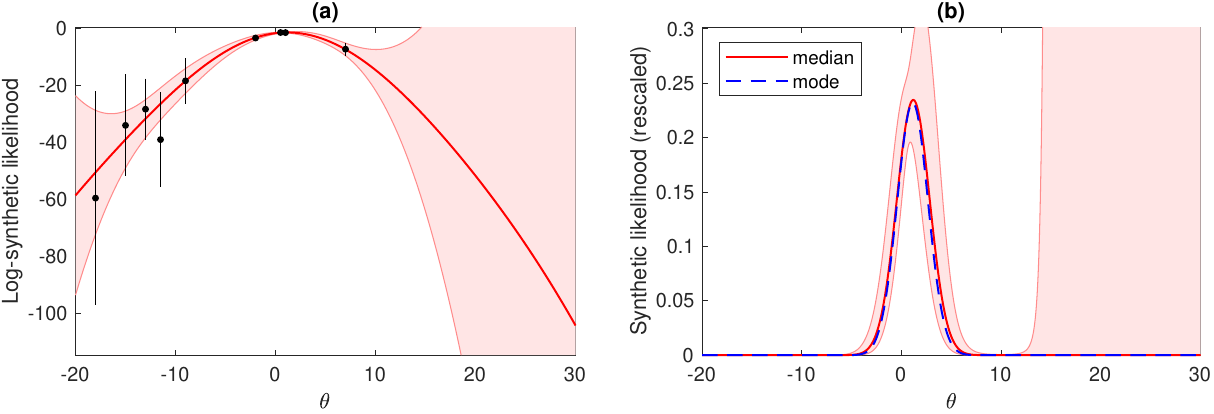}
\caption{As in Figure \ref{fig:loglik_gp1} but here a GP prior with a special mean function $m_0(\theta) = \beta_1 + \beta_2\theta + \beta_3\theta^2$ as described in Section \ref{subsec:gp_model} is used instead of a zero-mean GP prior.} \label{fig:loglik_gp2}
\end{figure}

\begin{figure}[hbt] % loglik modelling/estimator demo - additional evaluation
\centering
\includegraphics[width=0.8\textwidth]{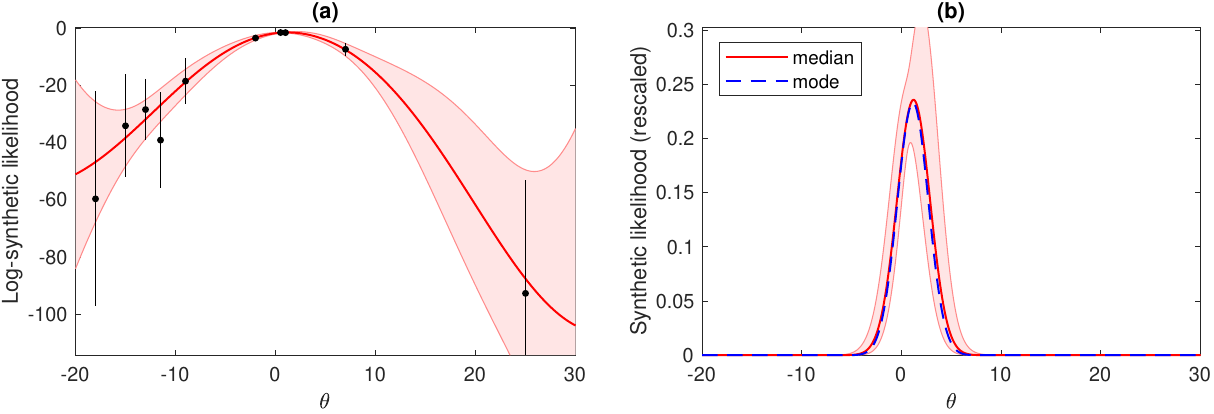}
\caption{As in Figure \ref{fig:loglik_gp1} but here an additional evaluation at $\theta=25$ is used for GP fitting.} \label{fig:loglik_gp3}
\end{figure}

\section{Additional discussion on modelling assumptions and implementation} \label{appsec:implementation_details}

We make some remarks on the Gaussian noise assumption (\ref{eq:noisemodel}) in Section \ref{appsec:gpremarks}. In Section \ref{appsec:gpmodellingdetails} we then discuss how potential practical difficulties with GP fitting are handled in our \gpmh{} algorithm. % and provide some remarks on modelling log-likelihood functions. %These details are not explicitly shown in Algorithm \ref{alg:gpmh} for clarity. 

\subsection{Remarks on the Gaussian noise assumption} \label{appsec:gpremarks}

% log-lik case
An unbiased, Gaussian distributed estimate of the \emph{log}-likelihood function is ideal for \gpmh{}. %If the log-likelihood is of the form $\sum_{i=1}^n g(\Bx_i\cond\Btheta)$ then an unbiased estimate is obtained e.g.~via resampling a small number of the $g(\Bx_i\cond\Btheta)$ terms. 
Such approximately Gaussian estimate is available in the likelihood-free generalised Bayesian updating case of Section \ref{subsec:bacterial_model}. 
% note on SL
The SL estimate of Section \ref{subsec:lfi} is in general not unbiased even if the Gaussianity assumption holds. % and some approximation error is hence introduced. 
The numerical results in \citet{Jarvenpaa2019_sl} however suggest that the resulting approximation error is small for typical $N$ and that the Gaussian noise assumption is sensible. %As the true, intractable likelihood is approximated with SL in the first place, losing unbiasedness may not be worrisome. 

% lik case
An unbiased estimator of the (approximate) likelihood is available in some applications such as in the standard ABC case of Section \ref{subsec:abc}. %, see e.g.~\citet[Chapter~4]{Sisson2019} for details. 
Also, an unbiased estimate for the likelihood of the parameters of a state-space model with intractable transition density can be obtained via bootstrap particle filtering. 
The ABC estimate is considered by \citet{Wilkinson2014} and the state-space model case by \citet{Drovandi2018} in the context of GP surrogate modelling but without clear justifications. 

% analysis
Suppose $\mean(p(\Btheta)) = \pi_{\text{lik}}(\Btheta) \eqdef \pi(\Bxobs\cond\Btheta)$ where the expectation is wrt.~the noise in the evaluation. % and where the exact likelihood function of interest is denoted by $\pi_{\text{lik}}(\Btheta)$. 
Let $y(\Btheta) \eqdef \log p(\Btheta)$. In general, we have
\begin{align*}
    \mean(y(\Btheta)) = \mean(\log p(\Btheta)) \neq \log \mean(p(\Btheta)) = \log \pi_{\text{lik}}(\Btheta),
\end{align*}
that is, an unbiased estimator of the likelihood function cannot be transformed to an unbiased estimate of the log-likelihood by just taking the logarithm. Similar observations, but in the opposite direction and in the context of pseudo-marginal \mh{}, are made in \citet[Section~4.2]{Bardenet2017} and \citet[Section~2]{Llorente2021}.

Suppose further that $p(\Btheta)$ follows log-Normal distribution with the mean $\pi_{\text{lik}}(\Btheta)$ and variance $\sigma^2(\Btheta) \geq 0$. Then we have
\begin{align*}
    y(\Btheta) \sim \Normal\left(\log\left(\frac{\pi_{\text{lik}}(\Btheta)^2}{\sqrt{\pi_{\text{lik}}(\Btheta)^2 + \sigma^2(\Btheta)}}\right), \log\left(1 + \frac{\sigma^2(\Btheta)}{\pi_{\text{lik}}(\Btheta)^2}\right)\right)
\end{align*}
by the basic properties of the log-Normal distribution. It follows that
\begin{align*}
    \mean(y(\Btheta)) = \log\pi_{\text{lik}}(\Btheta) + \log\left(\frac{\pi_{\text{lik}}(\Btheta)}{\sqrt{\pi_{\text{lik}}(\Btheta)^2 + \sigma^2(\Btheta)}}\right)
    \leq \log\pi_{\text{lik}}(\Btheta).
\end{align*}
Hence, if $\sigma^2(\Btheta) \ll \pi_{\text{lik}}(\Btheta)^2$ in this case, then $y(\Btheta)$ might work as a reasonable approximation for an unbiased estimate of $\log\pi_{\text{lik}}(\Btheta)$. %That is, if the noise level is low, taking the logarithm of the unbiased likelihood estimate might work. %One should keep in mind that additional approximation error is still introduced even if the Gaussian noise assumption holds otherwise. 

\subsection{Implementation details on handling log-likelihood evaluations} \label{appsec:gpmodellingdetails}

% handling NaN/Inf etc. loglik values
As mentioned in the main text, the log-likelihood function can behave irregularly in some boundary regions of the parameter space. %For example, log-likelihood evaluations there might result arbitrarily small outputs. 
%Similarly, the Gaussianity assumption of SL may not hold there producing log-SL evaluations driven by numerical error. 
This is usually not problematic for standard \mh{} (unless one tries to initialise \mh{} from such a region) because any proposal that results an infeasible log-likelihood evaluation is simply rejected. However, handling such log-likelihood values in GP-based methods requires care because including individual values with substantially different magnitudes to $\ddata_t$ typically leads to poor GP fits or numerical issues. In B(O)LFI methods one would additionally need to somehow ensure that the global optimum of the acquisition function does not lie on such problematic parameter regions where the GP fit cannot be trusted. %In B(O)LFI one would additionally need to rule out such regions when globally optimising the acquisition function and preferably without making the optimisation any harder. 
In the following we discuss how this practical problem is handled in our \gpmh{} implementation.

% dealing with invalid logliks
We treat a log-likelihood evaluation $y_j$ at $\Btheta_j\in\Theta$ as \emph{invalid} in our algorithms if any of the following holds: $y_j$ is a complex number or ``NaN'' (not-a-number), $|y_j|>10^5$ or $\sigma_n(\Btheta_j)>10^3$. %We discuss the appropriateness of these values in more detail later. 
Invalid $(y_j,\Btheta_j)$ is never included to $\ddata_t$ and hence not used for GP fitting. 

Invalid evaluations can result at different stages of \gpmh{}. First, if $2t_{\text{init}}$ tries do not produce $t_{\text{init}}$ valid initial evaluations, the algorithm is terminated as of having too poor initialisation. %This never happened in our experiments but is often guaranteed to occur if the initialisation is too far from the highest density region. 
%
% naive and epoer
If an invalid evaluation is obtained at the proposed point, the proposal is rejected and the algorithm continues as normal (neither $\ddata_t$ nor GP is updated). If the invalid evaluation is obtained at the current point then this means that the algorithm has likely proceeded to a point which should ideally have been rejected in an earlier iteration but were not based on the GP model. Because it may take long before the algorithm would manage to exit the problematic region, in this case the algorithm is terminated. %(Another option would be to restart the algorithm e.g.~from the most recent sample point at which a valid evaluation was obtained.) %We rarely observed this behaviour.
%
% more tricky epoe case
\epoe{} often evaluates at a point which is neither the current nor the proposed point. If an invalid evaluation is encountered in this case, it is neglected and a new evaluation is obtained using \naive{} instead so that the algorithm can proceed as in the \naive{}/\epoer{} case. Unless terminated, the algorithm switches back to \epoe{} after this exceptional step. 

The above heuristic procedure often allows the algorithm to recover if a problematic likelihood evaluation is obtained. A user can rerun the algorithm using a better initialisation if it was terminated. %Valid evaluations obtained during the previous run can be reused to initialise the next one. 
Note that technically any parameter can produce an invalid evaluation under the Gaussian noise assumption (\ref{eq:noisemodel}). Apart from some unusual or pathological situations, this is however expected to happen so rarely in the highest density region that the potential bias caused by always neglecting a proposal that produced an invalid evaluation is not taken into account. 
A practical difficulty is that the range of the log-likelihood function in the highest density region is rarely known in advance and depends, among other things, on the scaling of the data space, possible model misspecification and whether potential additive constant terms in the log-likelihood formula (which cancel out in the \mh{} accept/reject test but still affect the GP modelling) are neglected. Our definition of the invalid log-likelihood evaluation is based on our numerical experiments and some simple analyses but may require adjustments in some other settings due to these reasons.

\section{Additional details on noisy synthetic densities} \label{appsec:toymodel_details}

We here summarise the details of the three 6D toy log-densities originally presented in \citet{Jarvenpaa2019_sl} (where also their 2D versions were used for illustration and are shown as Figure D.2 of their supplementary material) and used in Section \ref{subsec:toymodels} of this article. 
These log-densities, which we denote as $f_{\textnormal{6D}}$, are constructed so that $f_{\textnormal{6D}}(\Btheta) = f_{\textnormal{2D}}(\Btheta_{1:2}) + f_{\textnormal{2D}}(\Btheta_{3:4}) + f_{\textnormal{2D}}(\Btheta_{5:6})$. 
The 2D log densities $f_{\textnormal{2D}}$ are then defined so that the `Simple' log-density results when $f_{\textnormal{2D}}(\Btheta)=-\Btheta\T \BS^{-1}_{\rho}\Btheta/2$ where $\rho=0.25$, 
the `Banana' results when $f_{\textnormal{2D}}(\Btheta)=-[\theta_1, \theta_2 + \theta_1^2 + 1] \BS^{-1}_{\rho}[\theta_1, \theta_2 + \theta_1^2 + 1]\T/2$ where $\rho=0.9$ and, finally, 
the `Multimodal' log-density is obtained using $f_{\textnormal{2D}}(\Btheta)=-[\theta_1, \theta_2^2-2] \BS^{-1}_{\rho}[\theta_1, \theta_2^2-2]\T/2$ where $\rho=0.5$. 
Above we have defined $\BS_{\rho}\in\reals^{2\times2}$ so that $(S_{\rho})_{11}=(S_{\rho})_{22}=1$ and $(S_{\rho})_{12}=(S_{\rho})_{21}=\rho$. 
The 2D structure of these models is used to aid computing the ground-truth posterior but is not taken into account in the GP modelling. 

% parameter bounds / priors
The above log-densities are additionally modified by specifying bounds for their six parameters. % although this is not strictly necessary. 
In practice this is achieved by using the above log-densities in the place of the target log-likelihood function for the \gpmh{} algorithm and by setting the uniform priors $\Unif([-16,16]^6), \, \Unif(\prod_{i=1}^3([-6,6]\!\times\![-20,2]))$ and $\Unif([-6,6]^6)$ for Simple, Banana and Multimodal densities, respectively. 
%
% initial points/cov matrices for M-H
We use the following initial points in our experiments: $\Btheta^{(0)}=-8\onevector$ for Simple and $\Btheta^{(0)}=-3\onevector$ for both Banana and Multimodal. The initial covariance matrix of the Gaussian proposal is $\BSigma_{0}=\Id$ for all three test cases.

\section{Additional results and experiments} \label{appsec:experiment_details}

\subsection{Noisy synthetic densities} \label{appsec:synth_add_res} % additional results for the 6D toy models

% results: x-axis MH iteration
We here complement Section \ref{subsec:toymodels} with additional results. 
First, Figure \ref{fig:res6d_iter} demonstrates how the quality of the marginal posterior approximation of \gpmh{} and \mhblfi{} develops as a function of iteration $i$, that is, as more approximate \mh{} samples are obtained. 
We observe that $10^5$ iterations are enough for Banana and Multimodal while $10^4$ iterations is already sufficient for the Simple log-density. Since the results by \mhblfi{}, shown in the bottom row of Figure \ref{fig:res6d_iter}, are based on a separate MH sampling with chain length $10^5$, its convergence is not directly affected by $i$. The \mhblfi{} results during the initial iterations are nevertheless poor ($\TV\approx1$) because the number of collected log-likelihood evaluations is obviously small initially and the resulting GP approximation hence inaccurate. %As more evaluations are collected the accuracy of the GP fits, and consequently the resulting posterior approximations, tends to increase. 
Both methods eventually produce approximations with similar quality with each $\epsi$. % as Figure \ref{fig:res6d_eval} also shows. % which is in line with the discussion in Section \ref{subsec:gpmcmcasblfi}. 

\begin{figure}[hbtp]
\centering
\includegraphics[width=0.97\textwidth]{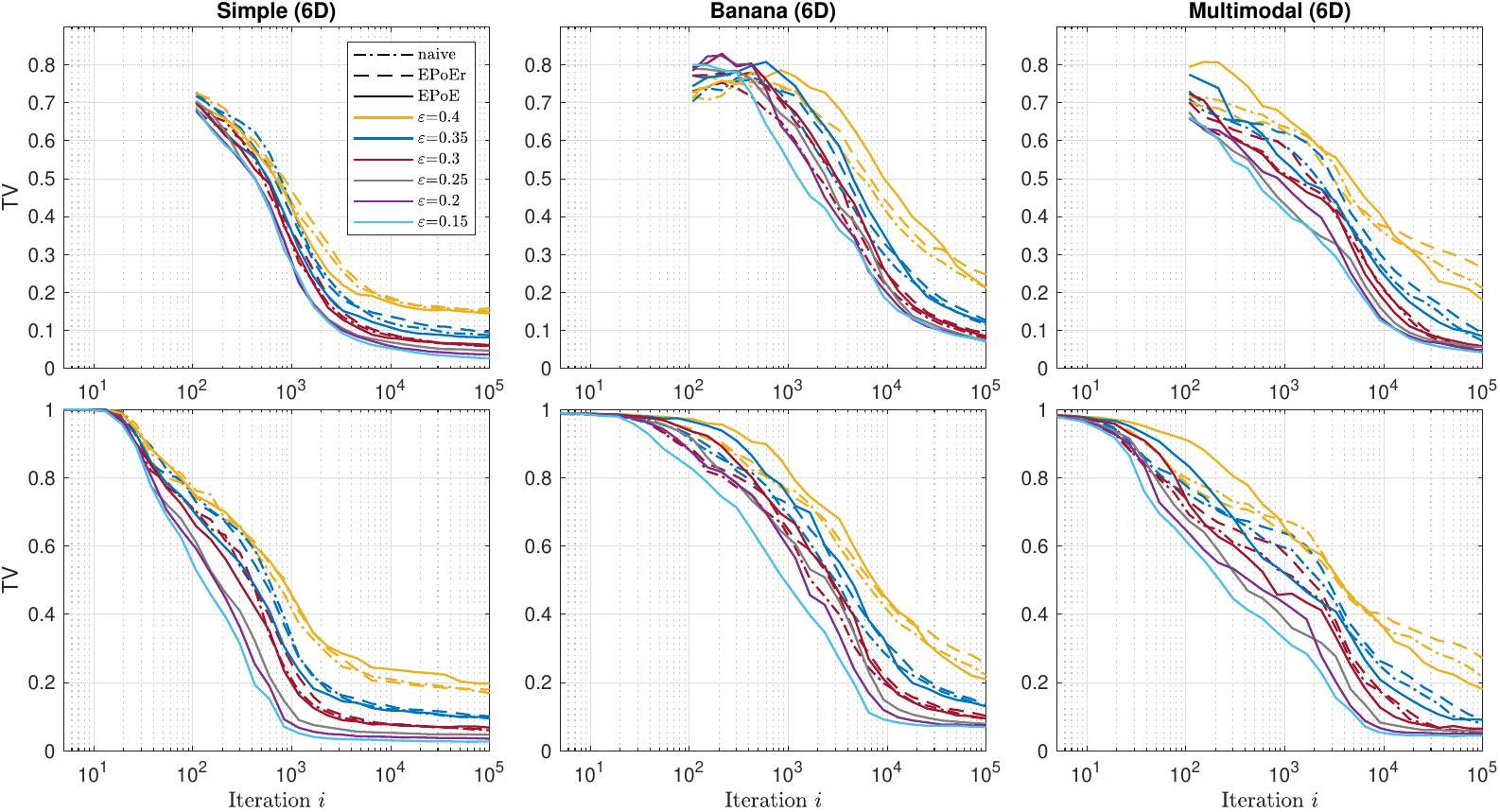}
\caption{Accuracy of the marginal posterior approximation as a function of iteration $i$ of Algorithm \ref{alg:gpmh}. Each line shows the median result over the $100$ repeats. \emph{Top row} shows \gpmh{} and the \emph{bottom row} the corresponding results by \mhblfi{}. The results by \gpmh{} at the early iterations ($i<10^2$) are not shown because sampling error is necessarily large there.} \label{fig:res6d_iter}
\end{figure}

% proportion of evaluations
Figure \ref{fig:res6d_cdfevals} shows that most of the log-density evaluations of the experiments of Section \ref{subsec:toymodels} occur already in the early stage of the algorithm which is neglected as burn-in. This is also the case with our other experiments and explains why \gpmh{} and \mhblfi{} produce similar results. The \epoe{} results with $\epsi=0.25, 0.2$ and $0.15$ are not shown as they were very similar to those with $\epsi=0.3$.

\begin{figure}[hbtp] % synthetic 6D, proportion of evaluations as a function of iteration
\centering
\includegraphics[width=0.97\textwidth]{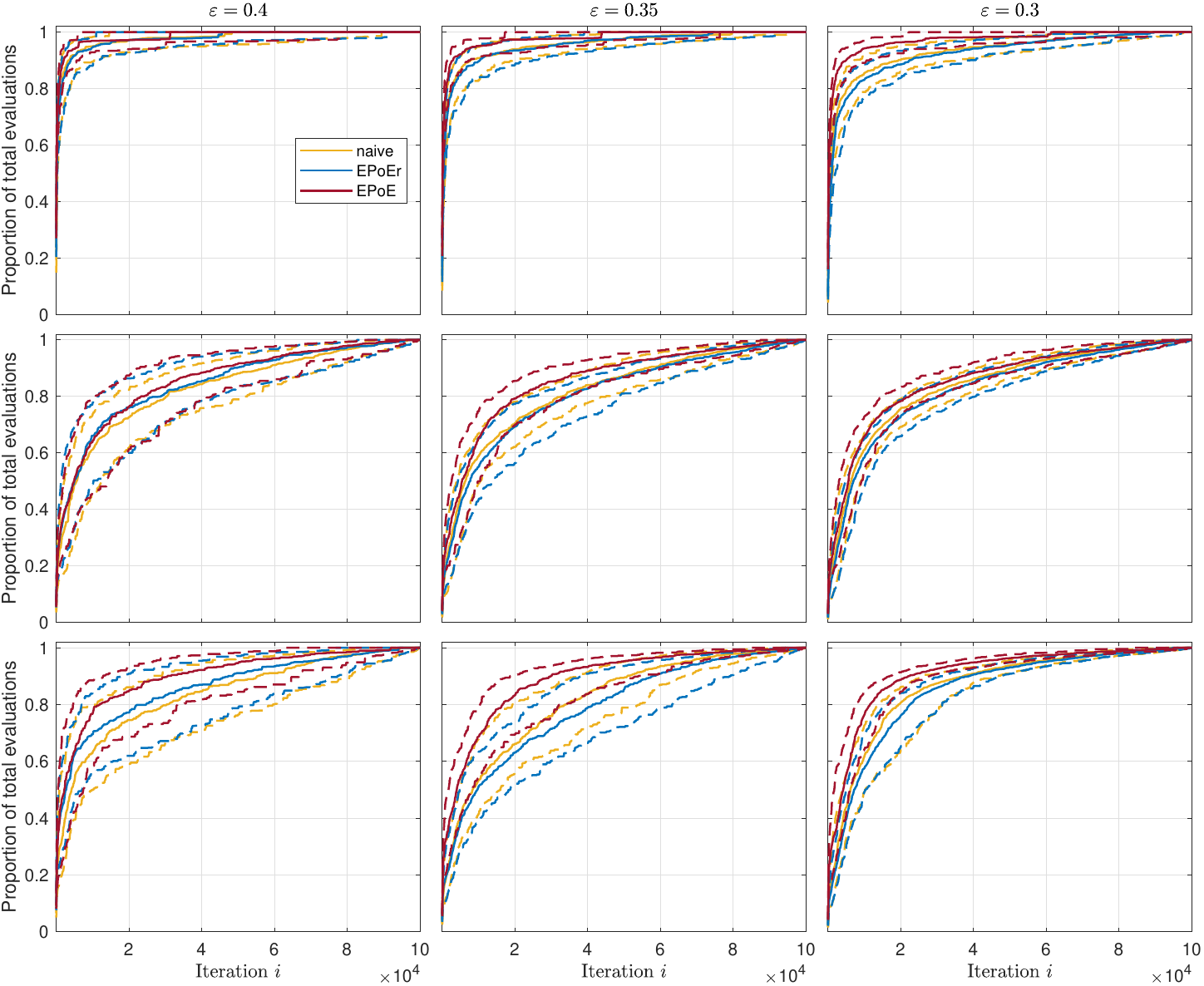}
\caption{Proportion of the total log-density evaluations collected as a function of \gpmh{} iteration $i$. \emph{Top row}: Simple (6D), \emph{middle row}: Banana (6D), \emph{bottom row}: Multimodal (6D) test density. The solid lines show the median and the dashed lines the $75\%$ quantile computed over the $100$ repeated runs.} \label{fig:res6d_cdfevals}
\end{figure}

% doubled noise
Figure \ref{fig:res6d_iter_highernoise} demonstrates the quality of the marginal posterior approximation as in Section \ref{subsec:toymodels} but when the noise levels have been doubled to $\sigma_n=4$ for Simple and $\sigma_n=2$ for Banana and Multimodal. All methods still produce reasonable results but more evaluations are of course needed. Our threshold for the maximum number of evaluations $10^3$ is always met in Banana and Multimodal cases using \epoer{} and \naive{} with $\epsi=0.3$ or \epoe{} with $\epsi=0.15$. \epoe{} is clearly more sample-efficient than \epoer{} and \naive{}. \epoe{} produces similar or slightly better sample-efficiency than \BLFI{} with IMIQR. 

\begin{figure}[hbtp] % synthetic 6D but with more noise
\centering
\includegraphics[width=0.97\textwidth]{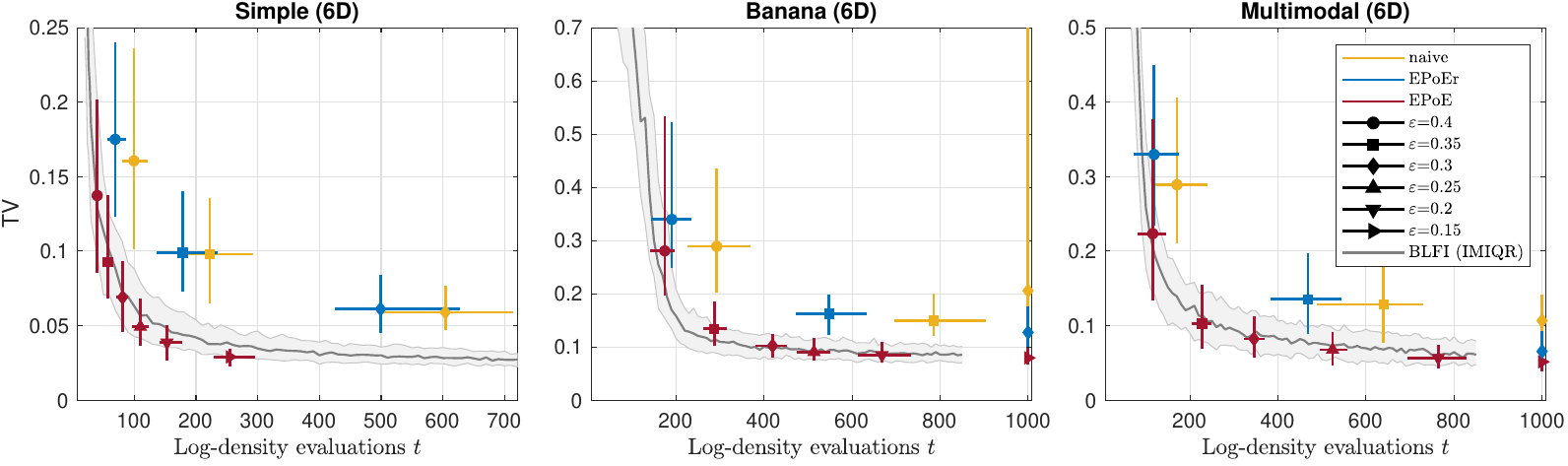}
\caption{Accuracy of the marginal posterior approximation as a function of the required log-density evaluations at the final iteration $i_{\text{MH}}=10^5$. These results are as in Figure \ref{fig:res6d_eval} except that the noise levels were doubled and only the \gpmh{} results are shown.} \label{fig:res6d_iter_highernoise}
\end{figure}

\subsection{Ricker model} \label{appsec:ricker_thetaricker}

% (standard) Ricker - general about the model
We briefly consider the (scaled) Ricker model used before e.g.~by \citet{Wood2010,Gutmann2016,Price2018} to analyse LFI methods. 
In this model the population evolves as
%$
\begin{equation*}
N_{t+1} = r N_{t} \exp(-N_{t} + \epsi_t), %\quad \epsi_t \simiid \Normal(0,\sigma^2_{\epsi}),
\end{equation*}
%$
for $t=1,\ldots,T$. The process noise and observation models are as in Section \ref{subsubsec:thetaricker}. Our experimental set-up is the same as for the theta-Ricker model except for the following differences: 
The priors are $\log(r)\sim\Unif([3,5]), \phi\sim\Unif([4,20]),\sigma_{\epsi}\sim\Unif([0,0.8])$ and the ``true'' parameter is $\Btheta_{\text{true}}=(3.8,10,0.3)$. 
We set $T=50$, $t_{\text{init}}=10$, $i_{\text{MH}} = 10^5$, $\Btheta^{(0)}=(3.4,  8.0,  0.15)$ and $\BSigma_0=\diag(0.1,  1.0,  0.1)^2$.

% Ricker - results
Figure \ref{fig:res_ricker_eval} shows the results in a similar fashion as before. We observe that less log-SL evaluations are needed as for theta-Ricker. %All \gpmh{} methods produce good accuracy when $\epsi \lesssim 0.25$. 
%The naive method produces the most accurate results when $\epsi=0.2$ and likely for the same reason as in the case of the bacterial infections model though the differences are here small. %The theta-Ricker model has two additional parameters so the intelligent search of evaluation locations becomes more important and possible GP model misspefication affects less. 
\BLFI{} with IMIQR produces slightly better sample-efficiency and accuracy than the \gpmh{} methods in this experiment. 
%Figure \ref{fig:ricker_post_example} in \appe{} \ref{appsec:experiment_details} demonstrates a typical estimated posterior and shows that the correlation structure is also estimated well. 
% examples of posterior approximation and demonstration of the evaluation locations
% Ricker model:
Figure \ref{fig:ricker_post_example} shows a typical example of the estimated posterior obtained using $\epsi=0.2$ and \epoe{}. %\epoer{} and \naive{} methods produced also similar approximations (but with the cost of additional log-likelihood evaluations). 
We can see that both the marginals and the correlation structure is estimated well.

\begin{figure}[hbtp]
\centering
\includegraphics[width=0.75\textwidth]{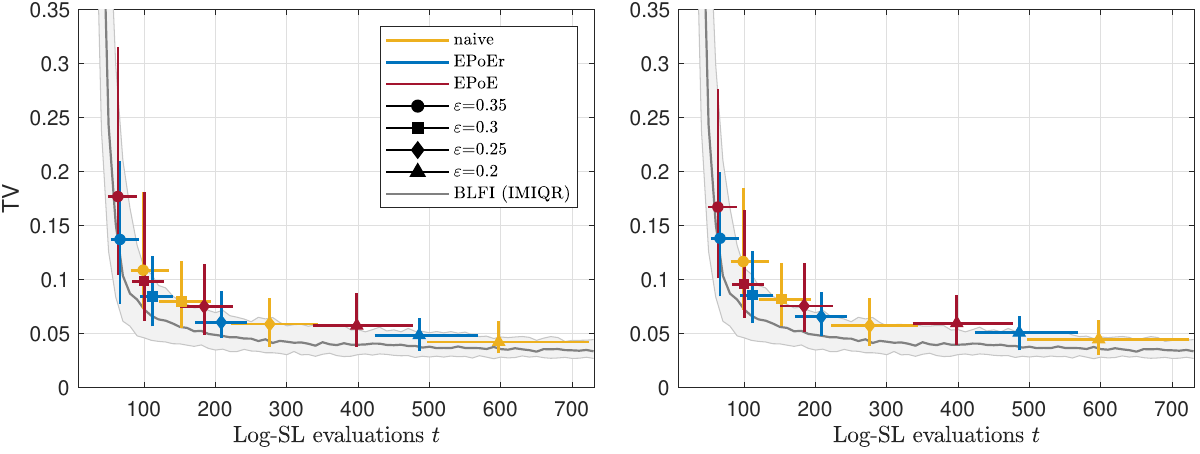}
\caption{Marginal posterior approximation accuracy as in Figure \ref{fig:res_thetaricker_eval} but for Ricker experiment and at $i_{\text{MH}}=10^5$ iterations. \emph{Left plot}: \gpmh{}, \emph{right plot}: \mhblfi{}.} \label{fig:res_ricker_eval}
\end{figure}

\begin{figure}[hbt] % example posterior - Ricker
\centering
\includegraphics[width=0.7\textwidth]{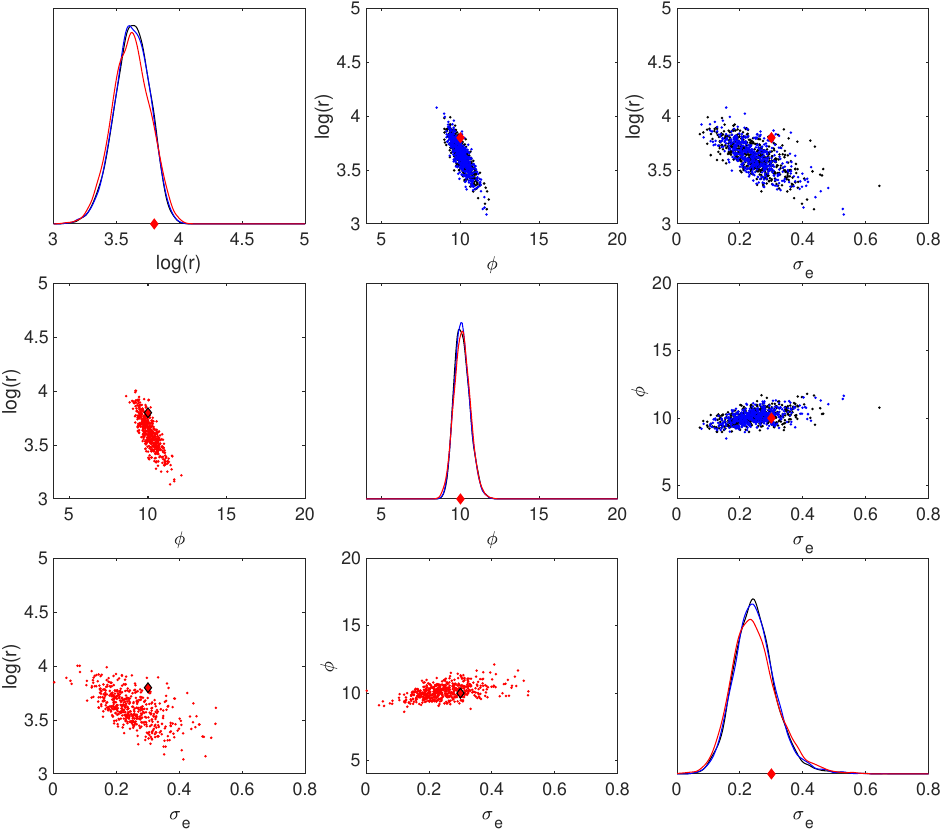}
\caption{Comparison of the ground-truth posterior (red dots/line) and a typical example of estimated posterior (black and blue dots/line) in the case of Ricker experiment. See the caption of Figure \ref{fig:thetaricker_post_example} for more detailed description.} \label{fig:ricker_post_example}
\end{figure}

\subsection{Theta-Ricker model}

% additional results for the theta-Ricker
We complement the theta-Ricker experiments of Section \ref{subsubsec:thetaricker} with some additional results. 
% discussion on the demo of eval locations
Figure \ref{fig:thetaricker_eval_example} demonstrates that \naive{}, \epoer{} and \epoe{} all tend to produce fairly similar evaluation locations. The evaluations by \epoe{} are however slightly more evenly distributed and diverse as those by the other two methods. 
%\epoe{} still requires less evaluations on average to reach similar approximation accuracy as the other methods. 

\begin{figure}[hbtp] % example evaluation locations - theta-Ricker
\centering
\includegraphics[width=0.85\textwidth]{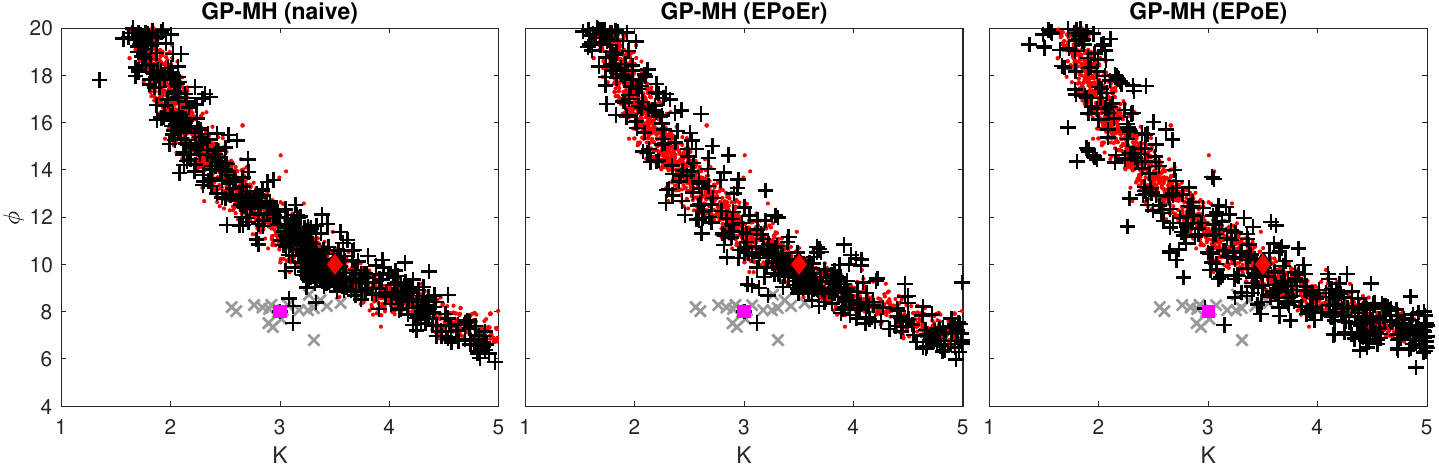}
\caption{Typical realisations of the log-SL evaluation locations projected to $(K,\phi)$-space in the theta-Ricker experiment. See the caption of Figure \ref{fig:simple_eval_example} for other details.} \label{fig:thetaricker_eval_example}
\end{figure}

% more extra results
Figure \ref{fig:res_thetaricker_iter} shows how the quality of the posterior approximation develops as a function of iteration $i$ of Algorithm \ref{alg:gpmh}. 
Finally, Figure \ref{fig:res_thetaticker_cdfevals} demonstrates, similarly to Figure \ref{fig:res6d_cdfevals}, the proportion of the collected log-SL evaluations. 

\begin{figure}[hbtp] % theta-Ricker, x-axis==iteration
\centering
\includegraphics[width=0.75\textwidth]{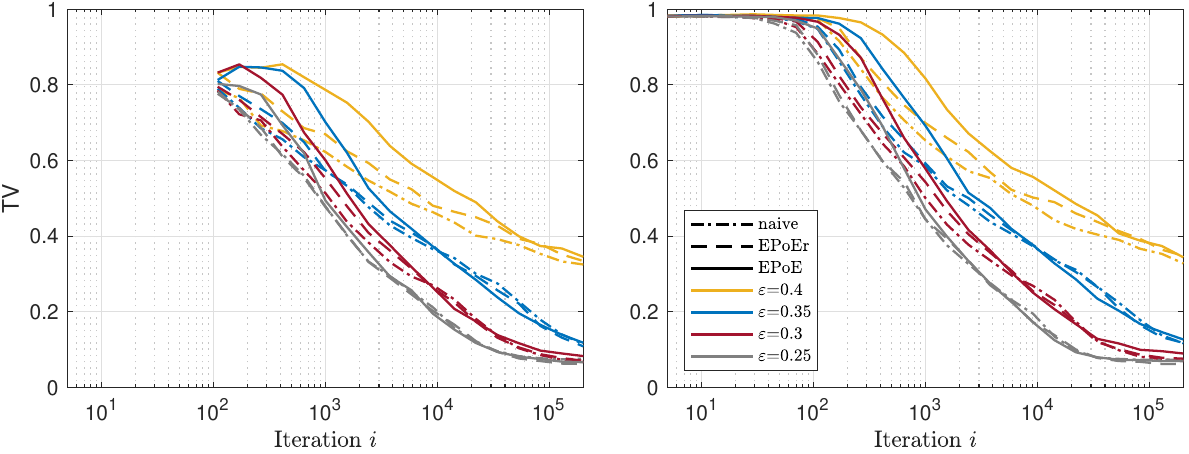}
\caption{Accuracy of the marginal posterior approximation in the theta-Ricker experiment as a function of iteration $i$ of Algorithm \ref{alg:gpmh}. \emph{Left plot}: \gpmh{}, \emph{right plot}: \mhblfi{}.} \label{fig:res_thetaricker_iter}
\end{figure}

\begin{figure}[hbtp] % theta-Ricker, proportion of evaluations as a function of iteration
\centering
\includegraphics[width=0.97\textwidth]{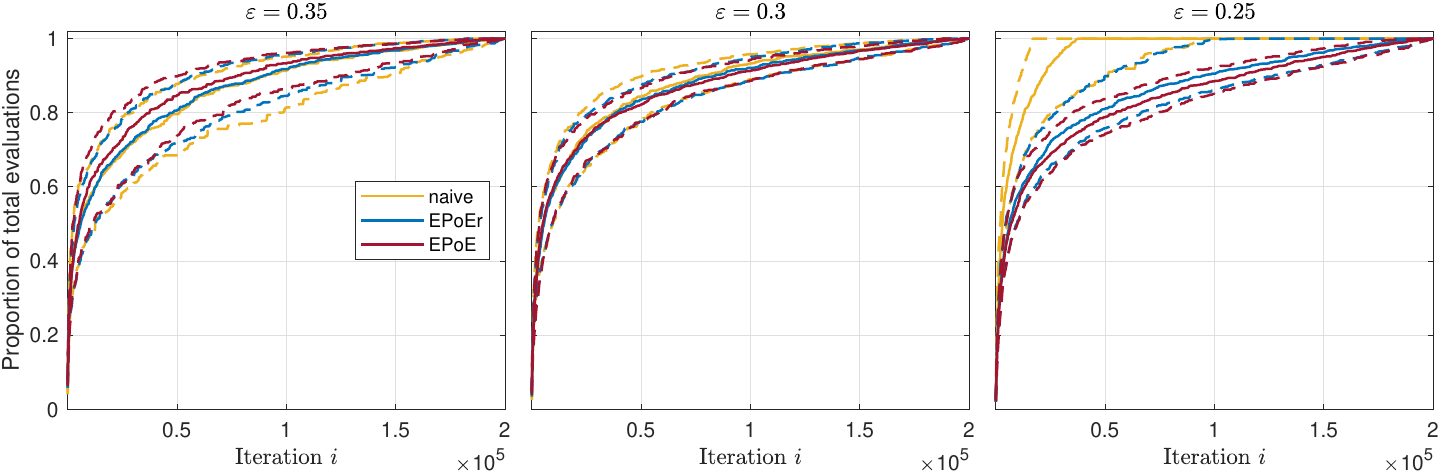}
\caption{Proportion of the total log-SL evaluations collected as a function of \gpmh{} iteration $i$ in the theta-Ricker experiment. The solid lines show the median and the dashed lines the $75\%$ quantile computed over the $100$ repeated runs.} \label{fig:res_thetaticker_cdfevals}
\end{figure}

%\subsubsection*{References}
%\newpage
%\clearpage
%\cleardoublepage
\bibliography{references_v2}

\end{document}